\documentclass[runningheads]{llncs}
\newtheorem{observation}[theorem]{Observation}

\newcommand{\mcfulltree}{\textsc{MC[Hyper$^2$LTL, Tree]} }
\newcommand{\mcfulldag}{\textsc{MC[Hyper$^2$LTL, Acyclic]} }

\newcommand{\fpfp}{Fixpoint Hyper$^2$LTL$_{fp}$ }

\newcommand{\mcfpfpdag}{\textsc{MC[Fixpoint, Acyclic]} }
\newcommand{\mcfpfptree}{\textsc{MC[Fixpoint, Tree]} }

\usepackage{ltl}
\usepackage{todonotes}
\usepackage{tikz}
\usepackage{pgf}
\usepackage{comment}
\usepackage{orcidlink}

\usepackage{cleveref}
\usepackage{subcaption}
\usepackage{array}

\usetikzlibrary{decorations.pathreplacing}
\usetikzlibrary{decorations.pathmorphing}
\usetikzlibrary{decorations.markings}
\usetikzlibrary{shapes,shapes.symbols,automata,arrows}
\usetikzlibrary{calc}
\usetikzlibrary{patterns}
\usetikzlibrary{positioning}
\usetikzlibrary{fit}
\usetikzlibrary{fadings}

\pgfdeclarelayer{background}
\pgfdeclarelayer{foreground}
\pgfsetlayers{background,main,foreground}

\usepackage{enumitem}
\usepackage{amsmath}
\allowdisplaybreaks

\usepackage[T1]{fontenc}
\usepackage{graphicx}
\begin{document}

\title{Complexity of Model Checking Second-Order Hyperproperties on Finite Structures\thanks{This work was supported in part by the Israel Science Foundation (ISF grant No. 655/25) and by the European Union (ERC Grant
HYPER, No. 101055412).}}

\titlerunning{Model Checking Second-Order Hyperproperties on Finite Structures} %TODO optional, please use if title is longer than one line

\author{Bernd Finkbeiner\inst{1}\orcidlink{0000-0002-4280-8441} \and Hadar Frenkel\inst{2}\orcidlink{0000-0002-3566-0338} \and Tim	Rohde\inst{3}\orcidlink{0009-0000-2419-1604}}
\authorrunning{Finkbeiner, Frenkel, and Rohde}
\institute{CISPA Helmholtz Center for Information Security, Saarbrücken, Germany \and Bar Ilan University, Ramat Gan, Israel \and Saarland University, Saarbrücken, Germany \\\email{finkbeiner@cispa.de, hadar.frenkel@biu.ac.il, tiro00001@stud.uni-saarland.de}}
\maketitle              % typeset the header of the contribution

\begin{abstract}
	
	We study the model checking problem of Hyper$^2$LTL over finite structures. Hyper$^2$LTL is a second-order hyperlogic, that extends the well-studied logic HyperLTL by adding quantification over sets of traces, to express complex hyperproperties such as epistemic and asynchronous hyperproperties. While Hyper$^2$LTL is very expressive, its expressiveness comes with a price, and its general model checking problem is undecidable. This motivates us to study the model checking problem for Hyper$^2$LTL over finite structures -- tree-shaped or acyclic graphs, which are particularly useful for monitoring purposes. 	
We show that Hyper$^2$LTL model checking is decidable on finite structures:
It is in PSPACE (in the size of the model) on tree-shaped models and in EXPSPACE on acyclic models.
Additionally, we show that for  
 an expressive fragment of Hyper$^2$LTL, namely the \fpfp fragment, the model checking problem is much simpler and is P-complete on tree-shaped models and EXP-complete on acyclic models. Last, we present some preliminary results that take into account not only the size of the model, but also the formula size. 
\end{abstract}

	\section{Introduction} 

Hyperproperties~\cite{DBLP:journals/jcs/ClarksonS10} generalize trace properties to reason about sets of sets of traces. 
Hyperproperties capture a wide range of properties, from information-flow and security policies~\cite{DBLP:conf/sp/GoguenM82a}, to complex epistemic properties such as knowledge and common knowledge~\cite{DBLP:books/mit/FHMV1995}. 
HyperLTL~\cite{ClarksonFKMRS14} is the prominent logic to capture temporal hyperproperties, and it extends LTL (Linear Temporal Logic~\cite{DBLP:conf/focs/Pnueli77}) by adding first-order quantification over trace variables. 
This way, HyperLTL can express many important information-flow and security-related properties, as well as knowledge in multi-agent systems. However, first-order quantification over traces is not enough to express all properties of interest. Two main examples that HyperLTL cannot express are common knowledge~\cite{DBLP:conf/fossacs/BozzelliMP15} and asynchronous hyperproperties~\cite{DBLP:journals/pacmpl/GutsfeldMO21,DBLP:conf/lics/BozzelliPS21,DBLP:conf/cav/BaumeisterCBFS21}.
Many efforts have been made in recent years to capture such complex hyperproperties, demonstrated by the range of logics suggested to express knowledge and common knowledge~\cite{DBLP:conf/fossacs/BozzelliMP15,DBLP:journals/jacm/HalpernM90,DBLP:books/mit/FHMV1995,Meyden98}, and even more so, the wide range of logics for asynchronous hyperproperties~\cite{BartocciHNC23,DBLP:conf/cav/BaumeisterCBFS21,BeutnerF23,DBLP:journals/corr/abs-2404-16778,DBLP:conf/lics/BozzelliPS21,DBLP:conf/concur/BozzelliPS22,DBLP:journals/pacmpl/GutsfeldMO21,DBLP:conf/mfcs/KontinenSV23,KontinenSV24,DBLP:conf/mfcs/KrebsMV018}. 

Recently, Beutner et al.~\cite{DBLP:conf/cav/BeutnerFFM23} 
observed that to express such complex hyperproperties, a second-order quantification, that is, quantification over \emph{sets of traces}, is needed. To this end,~\cite{DBLP:conf/cav/BeutnerFFM23} suggest the logic Hyper$^2$LTL, which extends HyperLTL with second-order quantification, and is able to express a wide range of hyperproperties such as common knowledge, asynchronous hyperproperties, and Mazurkiewicz trace theory~\cite{DR1995}. Yet, the expressiveness of Hyper$^2$LTL comes with a price, and its model checking problem is in general undecidable~\cite{DBLP:conf/cav/BeutnerFFM23}. To handle Hyper$^2$LTL model checking algorithmically,~\cite{DBLP:conf/cav/BeutnerFFM23} identify a fragment of Hyper$^2$LTL, namely Fixpoint Hyper$^2$LTL$_{fp}$, for which the model checking problem is still undecidable, but due to its definition 
using a fixpoint construction, it allows to compute sound over- and under-approximations of the second order sets. Using these approximations,~\cite{DBLP:conf/cav/BeutnerFFM23} provides a sound but necessarily incomplete model checking algorithm for an expressive fragment~of~Hyper$^2$LTL. 

In this work, motivated by the inherent undecidability of the logic, but also by practical applications such as monitoring~\cite{DBLP:conf/atal/BeutnerFF024}, we focus on the complexity of model checking second-order hyperproperties over finite structures. In particular, we analyze the complexity for two types of data structures: tree-shaped and acyclic models, both for the full Hyper$^2$LTL and for the \fpfp fragment. This study completes the complexity analysis picture for Hyper$^2$LTL: The complexity of HyperLTL has been studied for general structures~\cite{DBLP:journals/lmcs/FortinKTZ25} and for finite structures~\cite{MonitorHLTL}; for
second-order hyperproperties, 
Frenkel and Zimmermann~\cite{csl2025} provide a detailed complexity analysis for Hyper$^2$LTL and its fragments over general structures, and show that the problem is highly undecidable. 
However, no analysis is done for finite structures. 
Here, we address this gap, and show that when restricting the system to finite structures, the model checking problem becomes decidable, and its complexity drops significantly.

This study is also motivated by the complexity of the monitoring problem, particularly for second-order hyperproperties.  
In monitoring (in contrast to model checking), we do not have access to the whole system, but only to its executions. 
Based on the observed executions, we wish to decide whether the property holds or is violated in the (unknown) system. 
Since we can never observe infinite executions, the monitoring problem is defined over finite traces.
A (sequential) monitor for hyperproperties uses  a finite structure to store the set of traces seen so far, and repeatedly model checks this growing set against the specification.
In~\cite{DBLP:conf/atal/BeutnerFF024}, Beutner et al. address the monitoring problem of second-order hyperproperties by defining the problem and suggesting 
monitoring algorithms. However, they do not address the complexity aspect. 

The data structures we study here -- tree-shaped and acyclic structures -- 
are used in practical monitoring tools, such as~\cite{DBLP:journals/sttt/FinkbeinerHST20,DBLP:conf/atal/BeutnerFF024}; the trace logs seen in the monitoring process may be stored in the form of a simple linear collection of the traces seen so far or, for space efficiency, organized by common prefixes into a tree-shaped structure, or by both prefixes and suffixes into an acyclic structure.

In this work, we mostly analyze the complexity in terms of the size of the model, again motivated by applications such as
monitoring
and bounded model checking~\cite{DBLP:conf/tacas/HsuSB21}, where the model grows iteratively, yet the specification remains unchanged.  
Therefore, all complexities given in this work are in the size of the model, except when mentioned otherwise.
\begin{table}[t]
	\begin{center}
		\begin{tabular}{c||c|c}
			& Tree-shaped & Acyclic\\
			\hline\hline
			\fpfp & P-complete (Thm.~\ref{proofs:fppcomplete:complete})& EXP-complete (Thm.~\ref{proofs:fpexpcomplete:complete})\\
			\hline
			Hyper$^2$LTL $(\exists\forall)^k$ & $\Sigma^p_{k+1}$-complete (Thm.~\ref{proofs:hierarchycomplete:complete}) & $\Sigma^{\textit{EXP}}_{k+1}$-complete (Thm.~\ref{proofs:exphierarchy:complete})\\
			\hline
			Hyper$^2$LTL $(\forall\exists)^k$ & $\Pi^p_{k+1}$-complete (Thm.~\ref{proofs:hierarchycomplete:complete}) & $\Pi^{\textit{EXP}}_{k+1}$-complete (Thm.~\ref{proofs:exphierarchy:complete})\\
			\hline
			Hyper$^2$LTL & PSPACE (Cor.~\ref{proofs:hierarchycomplete:pspace}) & EXPSPACE (Cor.~\ref{proofs:exphierarchy:expspace})
		\end{tabular}
	\end{center}
	\caption{Complexity of Hyper$^2$LTL model checking in the size of the  structure. $(\exists\forall)^k$ and $(\forall\exists)^k$ denote  formulas with $k$ second-order quantifier alternations.}
	\label{summarytable1}
    \vspace{-5mm}
\end{table}
We summarize the results of our complexity analysis in Table~\ref{summarytable1}.
Note that our upper bound proofs in fact also provide concrete model checking algorithms for the respective problem.
In general, our results show that the model checking problem is exponentially easier on tree-shaped models than on general acyclic models, which conforms with the fact that trees have exponentially less traces, and are a less efficient representation of trace-sets. 
The more interesting observation following our results, is that model checking \fpfp formulas is much easier than model checking general Hyper$^2$LTL formulas. Our results demonstrate that the advantages of the fixpoint fragment are not limited only to infinite traces, but also appear on finite models, making the logic very appealing for finite settings, such monitoring and bounded model checking.
Since the reductions we provide are very technical, we use the body of the paper to present the main parts of our constructions, and refer the reader to the appendix for the full technical details and proofs. 
\vspace{-2mm}

\paragraph*{Related Work}\label{relatedwork}
For LTL, Kuhtz and Finkbeiner~\cite{DBLP:conf/concur/KuhtzF11} studied the impact of structural restrictions on the complexity of LTL model checking, and Fionda and Greco~\cite{DBLP:conf/aaai/FiondaG16} studied model checking complexity for different fragments of LTL over finite~traces.
For HyperLTL, Bonakdarpour and Finkbeiner~\cite{MonitorHLTL} provided a complexity analysis over finite structures, and showed that the complexity of HyperLTL model checking (generally NONELEMENTARY) drops significantly when~the model is restricted to be acyclic or tree-shaped. 

Over
general structures, Frenkel and Zimmermann~\cite{csl2025} showed that model checking Hyper$^2$LTL is highly undecidable (equivalent to truth in 3rd order arithmetic), but when restricting to the \fpfp fragment, the complexity reduces significantly, to $\Sigma_1^1$ (yet remains undecidable). 
Attempts to identify easier fragments for model checking were made in the context of epistemic and asynchronous logics subsumed by Hyper$^2$LTL, either by restricting~the formula, or the model structure: 
The logic LTL$_{K, C}$ extends LTL with epistemic operators to reason about knowledge and common knowledge. 
While in general its model checking problem is undecidable, it becomes decidable when removing either the \emph{until} temporal operator, or the \emph{common knowledge} epistemic operator~\cite{LTLKC}. 
Baumeister et al.~\cite{DBLP:conf/cav/BaumeisterCBFS21} identified syntactic fragments 
of A-HLTL, an asynchronous extension of HyperLTL, for which model checking is decidable. Hsu et al.~\cite{AHLTLBMC} showed that 
A-HLTL model checking is decidable over acyclic~graphs. 

Our work studies Hyper$^2$LTL, that subsumes all of the above, and 
concerns not only decidability, but completes the picture by providing a detailed complexity analysis for Hyper$^2$LTL and its fragments, over finite structures.

	\section{Preliminaries}\label{preliminaries}

\paragraph*{Kripke Structures}  
Let ${AP}$ be a finite set of atomic propositions.
A \emph{Kripke structure} is a tuple $K = (S, s_0, \delta, L)$ such that
$S$ is a finite set of states; $s_0\in S$ is the initial state;
$\delta\subseteq S\times S$ is the transition relation; and
$L : S\mapsto 2^{{AP}}$ is a labeling function.
If $(s,s') \in \delta$, we say that $s'$ is a \emph{successor} of $s$. 
We require that  $\forall s\in S. \exists s'\in S. (s,s') \in \delta$. That is, every state has at least one successor. Given a Kripke structure $K$, its underlying graph is the directed graph with the set $S$ as vertices, and edges defined by the transition relation $\delta$.\footnote{
While the data structures we study yield finite traces, 
to stick with the infinite trace semantics of Hyper$^2$LTL, and with the usual definition of Kripke structures, we use infinite traces with repeating last state to represent finite traces, similar to the choice of~\cite{MonitorHLTL}, instead of defining finite trace semantics as done e.g. for LTL in~\cite{MCSafety,truncp}.}
\begin{itemize}[nosep]
    \item 
A Kripke structure is \emph{acyclic} if the underlying graph $(S, \delta)$ is acyclic, 
except for states that have no outgoing transitions other than a self-loop. That is, if $(s,s)\in \delta$ then there is no $s'\neq s \in S$ such that $(s,s')\in \delta$. 

\item 
A Kripke structure is \emph{tree-shaped} if: (1) for every state $s'\neq s_0$ there is exactly one state $s\neq s'$ such that $(s, s')\in \delta$; (2) there is no $s$ such that $(s, s_0)\in\delta$; and (3) $(s,s)\in \delta$ iff there is no $s'\neq s$ such that $(s,s')\in \delta$.
\end{itemize}
A \emph{trace} $t\in(2^{{AP}})^\omega$ is an infinite sequence over the alphabet $2^{{AP}}$. 
We denote by $t[i]$ the $i$-th position of $t$, and by $t[i,\infty]$ the suffix of $t$ starting at position~$i$. 
Here, we only consider traces of the form $t = t[0] \cdots t[k] t[k+1] \cdots$ s.t. $\forall i\geq k. t[i] = t[i+1]$ for some $k\in \mathbb{N}$. The smallest such $k$ is the \emph{length} of $t$.
We denote the set of all traces of a Kripke structure $K$ by~$\textit{Traces}(K)$.

\paragraph*{Hyper$^2$LTL}
Hyper$^2$LTL~\cite{DBLP:conf/cav/BeutnerFFM23} extends the logic HyperLTL~\cite{ClarksonFKMRS14} by second-order quantification, that is, quantification over sets of traces. Let $\mathcal V$ be a finite set of trace variables, $\mathfrak V$ be a finite set of trace-set variables containing a special variable $\mathfrak G$ (which is reserved to denote the set of all traces we reason about), and let $AP$ be a finite set of atomic propositions.
Hyper$^2$LTL formulas are defined using the following grammar, for $\pi \in \mathcal V, X\in \mathfrak V, a\in AP$. 
\begin{gather*}
	\psi ::= a_\pi\mid \neg\psi\mid \psi\lor\psi\mid\psi\LTLuntil\psi\mid\LTLnext\psi 
	\\ \varphi ::= \forall X.\varphi\mid\exists X.\varphi\mid\forall \pi\in X. \varphi \mid \exists \pi\in X.\varphi \mid \psi
\end{gather*}
where $\LTLuntil, \LTLnext$ are the temporal operators \emph{until} and \emph{next} (their semantics is given below). We also allow the usual Boolean derivations $\mathit{true},\mathit{false}, \wedge, \rightarrow, \leftrightarrow$; and the temporal derivations $\LTLfinally\psi  \equiv \textit{true}\LTLuntil \psi$ (\emph{eventually}), and $\LTLglobally\psi \equiv \neg\LTLfinally\neg\psi$ (\emph{globally}). 

\emph{Semantics:}
Given a set of traces $T$,
we interpret Hyper$^2$LTL formulas w.r.t. a trace assignment $\Pi : \mathcal V\rightharpoonup T$, and a trace-set assignment $\Delta:\mathfrak V \rightharpoonup 2^T$ (both are partial functions). 
For the trace assignment $\Pi$, 
we write $\Pi[i, \infty]$ for the assignment $\Pi'(\pi) = \Pi(\pi)[i, \infty]$,
and $\Pi[\pi\mapsto t]$ for the assignment that maps $\pi$ to $t$, and maps every $\pi' \neq \pi$ to~$\Pi(\pi')$.
Similarly, for the trace-set assignment $\Delta$, we write $\Delta[X\mapsto A]$ for the assignment that maps $X$ to $A$, and every trace-set variable $Y\neq X$ to $\Delta(Y)$.
The semantics of Hyper$^2$LTL are defined over the relation $\models_T$, w.r.t. the set $T$ of traces, and the assignments $\Pi$ and $\Delta$, as follows. 
\begin{align*}
	&\Pi,\Delta \models_T a_\pi &\text{ iff } &a\in\Pi(\pi)[0]\\
	&\Pi,\Delta\models_T \neg\psi &\text{ iff } & \Pi,\Delta\not\models_T\psi\\
	&\Pi,\Delta\models_T \psi_1\lor\psi_2 &\text{ iff } &\Pi,\Delta\models_T\psi_1 \text { or }\Pi, \Delta\models_T\psi_2\\
	&\Pi,\Delta\models_T \LTLnext \psi &\text{ iff } &\Pi[1,\infty],\Delta\models_T\psi\\
	&\Pi,\Delta\models_T \psi_1\LTLuntil\psi_2 &\text{ iff } &\exists i\ge 0: \Pi[i,\infty],\Delta\models_T\psi_2 
\text{, and } \\ &&&
        \forall 0\le j < i: \Pi[j, \infty], \Delta\models_T\psi_1\\
	&\Pi, \Delta\models_T \forall\pi\in X. \varphi &\text{ iff } &\text{for all } t\in \Delta(X):\Pi[\pi\mapsto t], \Delta\models_T\varphi\\
	&\Pi, \Delta\models_T \exists\pi\in X. \varphi &\text{ iff } &\text{there exists } t\in \Delta(X):\Pi[\pi\mapsto t], \Delta\models_T\varphi\\
	&\Pi, \Delta\models_T \forall X.\varphi &\text{ iff } &\text{for all } A\subseteq T: \Pi,\Delta[X\mapsto A]\models_T\varphi\\
	&\Pi, \Delta\models_T \exists X.\varphi &\text{ iff } &\text{there exists } A\subseteq T: \Pi,\Delta[X\mapsto A]\models_T\varphi
\end{align*}
When the set of traces $T$ is clear from the context, we use $\models$ instead of $\models_T$. 
We say that a Kripke structure $K$ with traces $\mathit{Traces}(K)$ satisfies Hyper$^2$LTL formula $\varphi$ iff $\emptyset, [\mathfrak G\mapsto\textit{Traces}(K)]\models_{\textit{Traces}(K)}\varphi$, where $\emptyset$ stands for the empty trace assignment, and $[\mathfrak G\mapsto\textit{Traces}(K)]$ is the trace-set assignment that only maps the special trace-set variable $\mathfrak G$ to $\textit{Traces}(K)$, and non of the %is empty for all 
other trace-set variables. This is since we are only interested in satisfaction of Hyper$^2$LTL formulas with no free variables (other than $\mathfrak G$). 
In this case, we write $K\models\varphi$.\footnote{The semantics we use slightly differs from the original Hyper$^2$LTL semantics introduced in \cite{DBLP:conf/cav/BeutnerFFM23}, which uses an additional special trace-set variable $\mathfrak U$ to represent the set of all possible traces (not restricted only to the given model).
	Since we analyze the model checking problem for finite models, semantics that reason over all possible traces are not suited for our approach, and so we do not allow the use of this variable. This matches the closed-world semantics defined by~\cite{csl2025}.}

\paragraph*{Second-Order Quantifier Alternations}
The
{second-order quantifier alternations} of a Hyper$^2$LTL formula $\varphi$, is the number of times $\varphi$ alternates from universal to existential second-order quantification or vice versa.
All first-order quantifications are ignored. 
For example, the formula $\exists A. \forall\pi\in A.\exists B.\forall C.\exists D.\psi$, where $\psi$ is quantifier-free, has two second-order quantifier alternations: one from existential to universal set quantification, and one back to existential quantification.

 We denote Hyper$^2$LTL formulas with an outermost second-order existential quantifier by $\Sigma$-Hyper$^2$LTL. Such formulas that have $k$ second-order quantifier alternations, we denote by $\Sigma_k$-Hyper$^2$LTL.
Similarly, $\Pi$-Hyper$^2$LTL is the fragment of Hyper$^2$LTL formulas with outermost second-order universal quantifier, and $\Pi_k$-Hyper$^2$LTL are such formulas with $k$ alternations.
We say that a formula is \emph{quantifier free} if it contains neither first-order nor second-order quantifiers.

\paragraph*{Syntactic Sugar}
Before we proceed to the definition of the \fpfp fragment, we define some abbreviations we will use in the remainder of the paper.
\begin{itemize}[nosep]
    \item $\pi =_{\textit{AP}} \pi'$ stands for $ \bigwedge_{a\in \textit{AP}} \LTLglobally(a_\pi\leftrightarrow a_{\pi'})$ denoting equality between $\pi$ and~$\pi'$.
    \item 
$\exists! \pi \in X.\varphi(\pi)$ stands for
 $\exists\pi \in X.\forall\pi' \in X. \varphi(\pi)\land(\varphi(\pi')\rightarrow\pi =_{\textit{AP}}\pi')$, i.e., there {exists exactly one trace} for which $\varphi$ holds.
 \item $\pi\triangleright X$ stands for 
 $\exists\pi'\in X.\pi'=_{\textit{AP}} \pi$, i.e., $\pi$ is a member of the set $X$. If~$\triangleright$ is not under the scope of temporal operators, 
 we can bring
 the formula 
 into a proper Hyper$^2$LTL syntax.
 \item
 A trace $t$ \emph{encodes a 
	(binary) number $b$ with proposition $p$}
if:
$p$ holds in the $i$-th step of $t$ iff $b$ has a 1 at position $i$.
For traces $t$, $t'$ of equal length,~the formula $\LTLglobally(p_t\leftrightarrow p_{t'})$ holds iff $t$ and~$t'$ encode the same number $b$ with proposition $p$. All our encodings are in binary, and we start with the least significant bit.
    \item
    A trace $\pi$ is \emph{marked} with some atomic proposition $a$ if $\pi$ satisfies $\LTLfinally a_\pi$.
\end{itemize}
Additionally to this syntactic sugar, we often give our formulas with quantifiers under boolean operators.
It improves readability and allows us to nicely split formulas into subformulas.
Syntactically, this is not permitted but such a formula can be easily transformed into valid Hyper$^2$LTL syntax by pulling all quantifiers to the front.
Note that such a transformation is not possible if a quantifier appears under a temporal operator.

\paragraph*{Fixpoint Hyper$^2$LTL}
Along with Hyper$^2$LTL,~\cite{DBLP:conf/cav/BeutnerFFM23} introduces two sub-fragments: Hyper$^2$LTL$_{fp}$ and Fixpoint Hyper$^2$LTL$_{fp}$.
Hyper$^2$LTL$_{fp}$ restricts second-order quantification to largest or smallest sets that satisfy a given formula.
Fixpoint Hyper$^2$LTL$_{fp}$ further
 restricts second-order quantification so that the quantified sets are uniquely defined by least fixpoints.
This makes it more applicable to model checking, as demonstrated by~\cite{DBLP:conf/cav/BeutnerFFM23}, that proposed a partial model checking algorithm for the Fixpoint Hyper$^2$LTL$_{fp}$ fragment, based on over- and under-approximations of the fixpoints.

In this work, we analyze the complexities of the full Hyper$^2$LTL logic, and of the \fpfp fragment, but not the complexity of the intermediate Hyper$^2$LTL$_{fp}$ fragment. 
We expect that the model checking problem for Hyper$^2$LTL$_{fp}$ is of similar complexity as for Hyper$^2$LTL.
This is since Hyper$^2$LTL$_{fp}$ restricts second-order quantification, but neither gives a constructive way to compute all quantified sets nor restricts the number of sets that a formula can quantify over significantly. This is also demonstrated in the results of~\cite{csl2025}, where, over general (not finite) structures,  Hyper$^2$LTL and  Hyper$^2$LTL$_{fp}$ are at the same level of the arithmetical hierarchy, whereas the \fpfp fragment is much less complex (see Table~1 of~\cite{csl2025}). 
In contrast to Hyper$^2$LTL$_{fp}$, Fixpoint Hyper$^2$LTL$_{fp}$ restricts the number of sets that a formula can quantify over to be exactly one and gives a constructive way to compute this set.
This restrictive quantification is still expressive enough to express complex properties such as common knowledge, and yields an undecidable logic for the general case. 
The syntax of \fpfp is defined as follows:
\begin{gather*}
\psi ::= a_\pi\mid \neg\psi\mid \psi\lor\psi\mid\psi\LTLuntil\psi\mid\LTLnext\psi
	\\ \varphi ::= (X, \curlyvee,\varphi_{fp}).\varphi\mid\forall \pi\in X. \varphi \mid \exists \pi\in X.\varphi \mid \psi
\end{gather*}
Where $\varphi_{fp}$ is a conjunction of formulas of the form: 
$
\forall\pi_1\in X_1.\dots\forall\pi_n\in X_n.\psi_{\textit{step}}\rightarrow\pi_M\triangleright X
$, 
where $X_1, \dots, X_n$ are previously quantified trace-set variables, or $X$ itself; $\psi_{\textit{step}}$ is a quantifier free formula; and $1\le M\le n$. The semantics of \fpfp are the same as the semantics of Hyper$^2$LTL, except for second-order quantification, which is defined as follows. 
\begin{align*}
	\Pi, \Delta\models_T (X, \curlyvee, \varphi_{fp}).\varphi \text{ iff there exists } A\in \textit{sol}(\Pi,\Delta,(X, \curlyvee, \varphi_{fp})) 
    \text{  s.t.  }\\  \Pi,\Delta[X\mapsto A]\models_T\varphi
\end{align*}
where $\textit{sol}$ function returns the smallest sets (w.r.t. to set inclusion) satisfying~$\varphi_{fp}$:
\begin{align*}
	&  \textit{sol}(\Pi,\Delta,(X, \curlyvee, \varphi_{fp})) := \\
   & \{A\subseteq T\mid \Pi,\Delta[X\mapsto A]\models_T\varphi_{fp}\land\forall A'\subsetneq A.\Pi, \Delta[X\mapsto A']\not\models_T\varphi_{fp}\}
\end{align*}
Due to its definition using the $\psi_{step}$ formulas, each fixpoint iteration adds a unique set of traces to the constructed set.
The resulting fixpoint is therefore unique, i.e., $\textit{sol}$ consists~of~only one set (\cite{csl2025} Sec.7). Note that this defines a \emph{least} fixpoint, but for conciseness, and since we only analyze least fixpoints in this work, we neglect the \emph{least} from the name of the fragment.\footnote{In~\cite{DBLP:conf/cav/BeutnerFFM23}, it is termed the least fixpoint fragment. A largest fixpoint fragment is not defined in~\cite{DBLP:conf/cav/BeutnerFFM23}. Still, the analysis for a largest fixpoint fragment should be dual.}
 As discussed in~\cite{DBLP:conf/cav/BeutnerFFM23}, the fixpoint fragment is very appealing as,
 together with the fact that it is expressive enough to express properties such as common knowledge and asynchronous hyperproperties, it also allows (over general structures) to compute sound approximations of the second-order sets, thus providing a partial model checking algorithm for the fragment. In this work, we establish the efficiency of the fixpoint fragment also on finite structures, making it even more appealing.

\begin{example}\label{ex:fp}
The Fixpoint Hyper$^2$LTL$_{fp}$ formula 
$\varphi = (X, \curlyvee, \forall \pi\in X.\forall\pi'\in\mathfrak G. (\LTLglobally a _{\pi'} \lor \LTLglobally(a_\pi\leftrightarrow a_{\pi'})\lor\LTLglobally(b_\pi\leftrightarrow b_{\pi'})) \rightarrow\pi'\triangleright X). \forall \pi\in X. \neg b_{\pi}$
defines a set~$X$ that initially contains all traces satisfying $\LTLglobally a$.
With each fixpoint iteration, a trace $\pi'$ is added to $X$ if there is a trace $\pi\in X$ which globally agrees with $\pi'$
on $a$ or on~$b$.
See Fig.~\ref{prelims:example} for concrete traces.
The formula $\varphi$ can be used to reason about common knowledge in multi agent systems, as shown in~\cite{DBLP:conf/cav/BeutnerFFM23}. 
\end{example}
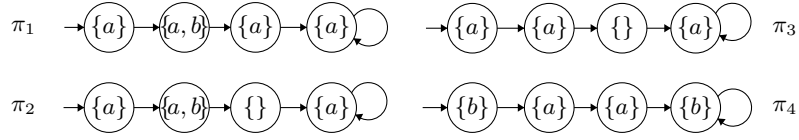
\begin{figure}[t]
    \centering
\begin{tikzpicture}[scale=0.11]
\tikzstyle{every node}+=[inner sep=0pt]
\draw [black] (10.3,-25.3) circle (3);
\draw (0,-25.3) node {$\pi_1$};
\draw (10.3,-25.3) node {$\{a\}$};
\draw [black] (19.4,-25.3) circle (3);
\draw (19.4,-25.3) node {\footnotesize{$\{\!a,b\!\}$}};
\draw [black] (28,-25.3) circle (3);
\draw (28,-25.3) node {$\{a\}$};
\draw [black] (37.2,-25.3) circle (3);
\draw (37.2,-25.3) node {$\{a\}$};
\draw [black] (10.3,-35) circle (3);
\draw (0,-35) node {$\pi_2$};
\draw (10.3,-35) node {$\{a\}$};
\draw [black] (19.4,-35) circle (3);
\draw (19.4,-35) node {\small$\{\!a,b\!\}$};
\draw [black] (28,-35) circle (3);
\draw (28,-35) node {$\{\}$};
\draw [black] (37.2,-35) circle (3);
\draw (37.2,-35) node {$\{a\}$};
\draw [black] (54.2,-25.4) circle (3);
\draw (54.2,-25.4) node {$\{a\}$};
\draw [black] (63.5,-25.4) circle (3);
\draw (63.5,-25.4) node {$\{a\}$};
\draw [black] (72.3,-25.4) circle (3);
\draw (72.3,-25.4) node {$\{\}$};
\draw [black] (81.2,-25.4) circle (3);
\draw (81.2,-25.4) node {$\{a\}$};
\draw (92,-25.4) node {$\pi_3$};
\draw [black] (54.2,-35) circle (3);
\draw (54.2,-35) node {$\{b\}$};
\draw [black] (63.5,-35) circle (3);
\draw (63.5,-35) node {$\{a\}$};
\draw [black] (81.2,-35) circle (3);
\draw (81.2,-35) node {$\{b\}$};
\draw (92,-35) node {$\pi_4$};

\draw [black] (72.3,-35) circle (3);
\draw (72.3,-35) node {$\{a\}$};
\draw [black] (4.9,-25.3) -- (7.3,-25.3);
\fill [black] (7.3,-25.3) -- (6.5,-24.8) -- (6.5,-25.8);
\draw [black] (13.3,-25.3) -- (16.4,-25.3);
\fill [black] (16.4,-25.3) -- (15.6,-24.8) -- (15.6,-25.8);
\draw [black] (22.4,-25.3) -- (25,-25.3);
\fill [black] (25,-25.3) -- (24.2,-24.8) -- (24.2,-25.8);
\draw [black] (31,-25.3) -- (34.2,-25.3);
\fill [black] (34.2,-25.3) -- (33.4,-24.8) -- (33.4,-25.8);
\draw [black] (39.88,-23.977) arc (144:-144:2.25);
\fill [black] (39.88,-26.62) -- (40.23,-27.5) -- (40.82,-26.69);
\draw [black] (4.8,-35) -- (7.3,-35);
\fill [black] (7.3,-35) -- (6.5,-34.5) -- (6.5,-35.5);
\draw [black] (13.3,-35) -- (16.4,-35);
\fill [black] (16.4,-35) -- (15.6,-34.5) -- (15.6,-35.5);
\draw [black] (22.4,-35) -- (25,-35);
\fill [black] (25,-35) -- (24.2,-34.5) -- (24.2,-35.5);
\draw [black] (31,-35) -- (34.2,-35);
\fill [black] (34.2,-35) -- (33.4,-34.5) -- (33.4,-35.5);
\draw [black] (39.613,-33.238) arc (153.86581:-134.13419:2.25);
\fill [black] (40.07,-35.84) -- (40.56,-36.65) -- (41.01,-35.75);
\draw [black] (48.6,-25.4) -- (51.2,-25.4);
\fill [black] (51.2,-25.4) -- (50.4,-24.9) -- (50.4,-25.9);
\draw [black] (57.2,-25.4) -- (60.5,-25.4);
\fill [black] (60.5,-25.4) -- (59.7,-24.9) -- (59.7,-25.9);
\draw [black] (66.5,-25.4) -- (69.3,-25.4);
\fill [black] (69.3,-25.4) -- (68.5,-24.9) -- (68.5,-25.9);
\draw [black] (83.616,-23.641) arc (153.78241:-134.21759:2.25);
\fill [black] (84.07,-26.25) -- (84.56,-27.05) -- (85,-26.15);
\draw [black] (75.3,-25.4) -- (78.2,-25.4);
\fill [black] (78.2,-25.4) -- (77.4,-24.9) -- (77.4,-25.9);
\draw [black] (48.2,-35) -- (51.2,-35);
\fill [black] (51.2,-35) -- (50.4,-34.5) -- (50.4,-35.5);
\draw [black] (57.2,-35) -- (60.5,-35);
\fill [black] (60.5,-35) -- (59.7,-34.5) -- (59.7,-35.5);
\draw [black] (83.88,-33.677) arc (144:-144:2.25);
\fill [black] (83.88,-36.32) -- (84.23,-37.2) -- (84.82,-36.39);
\draw [black] (66.5,-35) -- (69.3,-35);
\fill [black] (69.3,-35) -- (68.5,-34.5) -- (68.5,-35.5);
\draw [black] (75.3,-35) -- (78.2,-35);
\fill [black] (78.2,-35) -- (77.4,-34.5) -- (77.4,-35.5);
\end{tikzpicture}

 \caption{The formula from \Cref{ex:fp} is satisfied by the above set of traces $T=\{\pi_1,\pi_2,\pi_3,\pi_4\}$. The set $X$ initially contains $\pi_1$, $\pi_2$ is added as it agrees with $\pi_1$ on $b$, and  $\pi_3$ is added since it agrees with $\pi_2$ on $a$. $\pi_4$ does not agree with any other trace on any proposition, thus $X=\{\pi_1,\pi_2,\pi_3\}$.\vspace{-3mm}
}
\label{prelims:example}
\end{figure}

\vspace{-3mm}

\paragraph*{Problem Statement}
The model checking problem is to decide, given a Kripke structure~$K$ and a Hyper$^2$LTL formula $\varphi$, whether $K\models\varphi$. 
We study the following model checking (MC) problems.
If not stated otherwise, we 
study the \emph{data complexity}, 
i.e., the complexities are w.r.t. the size of $K$ and not the size of $\varphi$.
\begin{itemize}[nosep]
	\item \mcfpfptree$\!\!$: $K$ is tree-shaped and $\varphi$ is in Fixpoint Hyper$^2$LTL$_{fp}$.% formula.
	\item \mcfpfpdag$\!\!$: $K$ is acyclic and $\varphi$ is in Fixpoint Hyper$^2$LTL$_{fp}$.
	\item \mcfulltree: $K$ is tree-shaped and $\varphi$ is not further restricted.
	\item \mcfulldag: $K$ is acyclic and $\varphi$ is not further restricted.
\item For the full logic, we also define the following problems, where \textsf{struct} can be \textsc{tree} or  \textsc{acyclic}, and the structure $K$ is according to  \textsf{struct}:
\subitem \textsc{MC[$\Sigma_k$-Hyper$^2$LTL,} \textsf{struct}]: $\varphi$ is in $\Sigma_k$-Hyper$^2$LTL. 
	\subitem \textsc{MC[$\Pi_k$-Hyper$^2$LTL,} \textsf{struct}]:  $\varphi$ is in $\Pi_k$-Hyper$^2$LTL.
\end{itemize}

	\section{P-completeness of \mcfpfptree}\label{proofs:fppcomplete}
In this section, we show that \mcfpfptree is P-complete. We start by proving that it is in P. Then we give a reduction in logarithmic space from the Horn-satisfiability problem~\cite{CompuComplexity} to show that it is P-hard. All reductions we present in the paper are polynomial time. 

\begin{observation}
\label{proofs:treeremark}
A tree-shaped Kripke structure with $n$ states has at most $n$ traces. In addition, each leaf of the tree corresponds to exactly one trace. 
\end{observation}

\begin{lemma}[\cite{MonitorHLTL}]\label{lemma:noquantifiershyperltl}
    Let $K$ be a tree-shaped Kripke structure. 
    Under a fixed trace assignment, a quantifier-free HyperLTL formula $\psi$ can be evaluated on $K$ in polynomial time in the size of $K$. 
\end{lemma}
\begin{corollary}\label{lemma:noquantifiers}
     Let $K$ be a tree-shaped Kripke structure. 
    Under a fixed trace and trace-set assignments, a quantifier-free Hyper$^2$LTL formula $\psi$ can be evaluated on $K$ in polynomial time in the size of $K$.  
\end{corollary}

\Cref{lemma:noquantifiers} follows from the fact that given a concrete trace-set assignment, verifying Hyper$^2$LTL reduces to verifying HyperLTL formulas. 

Since we can represent a finite set of traces $T$ as a tree with $|T|$ branches,~we~get:
\begin{corollary}\label{cor:noquantifiers}
    Let $T$ be a finite set of traces. 
    Under a fixed trace assignment, a quantifier-free Hyper$^2$LTL formula $\psi$ can be evaluated on $T$ in polynomial time in the size of $T$. 
\end{corollary}

\begin{lemma}\label{proofs:fppcomplete:inp}
\mcfpfptree is in P.
\end{lemma}

\begin{proof}
Let $\varphi$ be a fixed \fpfp formula and let $K$ be a tree-shaped Kripke structure.
We recursively evaluate $\varphi$ over $K$ as follows.
\begin{itemize}[nosep]
    \item Every second-order (fixpoint) quantifier $(X, \curlyvee, \varphi_{fp}).\varphi'$ quantifies over exactly one set $X$, which we iteratively construct by following the fixpoint formula $\varphi_{fp}$.
We can keep track of $X$ by marking leaves in the tree that correspond to traces in $X$.
We evaluate the inner formula $\varphi'$ under the current instantiation of $X$.
\item For every first-order quantifier $\mathbb Q \pi\in X.\varphi'$, we iterate over every possible instantiation of~$\pi$ and evaluate if $\varphi'$ holds for the current instantiation.
\end{itemize}
We prove that this process is polynomial in the size of the structure, by induction over Fixpoint Hyper$^2$LTL$_{fp}$ formulas:
\begin{itemize} [nosep]
\item If $\varphi$ is quantifier-free, it can be evaluated 
in polynomial time under fixed trace and trace-set assignments (Cor.~\ref{lemma:noquantifiers}).
\item For $\varphi = \forall\pi\in X.\varphi'$ or $ \varphi =\exists\pi\in X.\varphi'$: 
By induction hypothesis, $\varphi'$ can be evaluated in polynomial time under a fixed trace assignment. 
The number of possible assignments for $\pi$ is polynomial in the size of the Kripke structure (due to Obs.~\ref{proofs:treeremark}), and therefore 
the satisfaction of first-order quantifiers can be evaluated in polynomial time in $K$.
\item For $\varphi = (X, \curlyvee,\varphi_{fp}).\varphi'$: 
We only have
restricted fixpoint sets, which are uniquely defined by the inner formula $\varphi_{fp}$.
Each set can contain at most polynomially many traces (Obs.~\ref{proofs:treeremark}), so 
we need at most polynomially many fixpoint iterations to evaluate it.
Each such iteration can be done in polynomial time, by the induction hypothesis.\qed
\end{itemize}%\vspace{-4mm}
\end{proof}

\subsection{P-hardness of \mcfpfptree} \label{sec:hornsat}

To show 
P-hardness we reduce from the Horn satisfiability problem.

\begin{definition}[Horn-satisfiability~\cite{CompuComplexity}] 
Let $L = \{x_1, \dots, x_k\}$ be a set of $k$ variables, and let $h = c_1 \wedge c_2 \wedge \ldots \wedge c_n$ be a formula 
consisting of a conjunction of~$n$ Horn clauses, each of them of the form $c_i = (\neg l_1\lor \neg l_2\lor l_3)$, where $l_1, l_2, l_3\in L \cup \{\top,\bot \}$.
The Horn-satisfiability problem is to find a Boolean assignment for the variables $x_1, \dots, x_k$ such that $h$ is satisfied.
\end{definition}

\begin{theorem}\label{proofs:fppcomplete:complete}
	\mcfpfptree is P-complete.
\end{theorem}

\begin{proof}
	Containment in P follows from \Cref{proofs:fppcomplete:inp}. 
    We show P-hardness by a log-space and polynomial time reduction from~the Horn Satisfiability problem, which is P-hard~\cite{CompuComplexity}.
	Given a Horn formula $h$,
	we construct a tree structure $K$ and a \fpfp formula~$\varphi$ (which does not depend on $h$) such that $K\models\varphi$ iff $h$ is true. We hereby describe our construction and refer to~App.~\ref{app:fppcomplete} for its full correctness proof. Here, and in the following proofs throughout the paper, when we use \textbf{encode} we mean \textbf{encode in binary} (as we stated in Sec.~\ref{preliminaries}, all our encodings are in binary). 
Intuitively, we build a tree structure $K$ that has,
for each variable $x_i$, one branch representing a positive assignment to $x_i$, and 
one branch representing a negative assignment to $x_i$.
Additionally, $K$ has one branch for every clause~$c_i$.
For simplicity, we add variables $x_{k+1}$ and $x_{k+2}$ to represent $\top$ and $\bot$, and restrict the valid assignments such that $x_{k+1}$ has to be assigned {true} and $x_{k+2}$ has to be assigned {false}.
On every branch of $K$, we encode the index of the corresponding variable(s) with atomic propositions $\textit{pos, neg}_1,\textit{neg}_2$.

The \fpfp formula constructs a set $A$ representing an assignment to~$h$. Intuitively, if $A$ contains the positive trace of a variable $x_i$, it means that the corresponding assignment assigns $x_i$ with true; and if $A$ contains the negative trace of $x_i$, then $x_i$ is assigned with false.
We define 
$A$ iteratively by adding an assignment trace to $A$ only if there is a clause that cannot be satisfied if that variable would be assigned to the other value.

We now describe our reduction formally. 
Let $h = c_1 \wedge c_2 \wedge \ldots \wedge c_n$ be a Horn formula with $k$ variables. 
We build a Kripke structure $K = (S, s_0, \delta, L)$ whose underlying graph is a rooted tree. We refer to Fig.~\ref{proofs:fppcomplete:example} for an example.
The state-space $S$ consists of: the initial state~$s_0$,  $2\cdot\lceil\log_2(k+3)\rceil$ states for each variable (including $\top, \bot$), and $\lceil\log_2(k+3)\rceil$ states for each clause in $h$.
We denote the states corresponding to a variable $x_i$ by $s_{i, 1}, s_{i, 2}, \cdots$ and $s'_{i, 1}, s'_{i,2},\cdots$.
The states corresponding to the $i$-th clause of $h$ are denoted $t_{i, 1}, t_{i, 2},\cdots$.
The relation $\delta$ is: %defined as follows:
\begin{align*}
\delta(s_0) &= \{s_{i,1}\mid i\in[1,k]\}\cup\{s'_{i, 1}\mid i\in[1,k]\}\cup\{t_{j, 1}\mid j\in[1, n]\}\\
\delta(\sigma_{i,m}) &= \left\{\begin{matrix}
\{\sigma_{i, m+1}\}&\text{if }m < \lceil\log_2(k+3)\rceil\\
\{\sigma_{i,m}\}&\text{otherwise}\end{matrix}\right.\quad\quad\forall \sigma\in\{s, s', t\}
\end{align*}
That is, there is a transition from $s_0$ to the states $s_{i, 1}, s'_{i, 1}$ for $i\in[1, k]$ and $t_{j, 1}$ for $j\in[1, n]$, and there are branches of the type $s,s',t$ for each $i,j$. Note that each trace can contain states only of one of the types $s,s'$, or $t$, as this is a tree-shaped structure. 
The set $AP$ is $\{\textit{pos},\textit{neg}_1, \textit{neg}_2, \textit{c}, \textit{a}\}$ where $c$ denotes clause branches, and $a$ denotes the branches corresponding to $\top$ and $\bot$.
Traces consisting of states $s_{i, m}$
encode $i$ with $\textit{pos}$, and traces consisting of states $s'_{i, m}$
encode $i$ with $\textit{neg}_1$.
A trace consisting of the states $t_{i, m}$ representing the $i$-th clause $(\neg x_d \vee \neg x_e \vee x_f)$ encodes: $d$ with $\textit{neg}_1$, $e$ with $\textit{neg}_2$, and $f$ with $\textit{pos}$.
This results in a tree-shaped Kripke structure:
For each variable $x_i$, we have two branches representing a positive and negative assignment to $x_i$, and for each clause we have a branch labeled with $c$.
The branches representing correct assignment to $\top$ and $\bot$ are marked with \textit{a}.
The \fpfp formula $\varphi$, for which we model check $K$, is: 
$\varphi: = (A, \curlyvee, \varphi_{fp}\land\varphi'_{fp}).
\forall\pi\in A.\forall\pi'\in A. \neg\square(\textit{pos}_\pi\leftrightarrow\textit{neg}_{1,\pi'})$.
Where $\varphi'_{fp} = \forall\pi\in \mathfrak G. \bigcirc a_\pi\rightarrow\pi~\triangleright A$, and 
$\varphi_{fp}$ is defined as follows:%\todo{in Eq1, seperate the equations $\neq$ }
\begin{align*}
&\forall \pi\in \mathfrak G.\forall \alpha\in A.\forall\alpha'\in A.\forall\beta\in \mathfrak G.
\LTLnext c_\pi\land \neg\LTLnext c_\alpha \land \neg\LTLnext c_{\alpha'}\land\neg\LTLnext c_\beta & (1) \\
&\land(\square(\textit{pos}_\alpha\leftrightarrow\textit{neg}_{1,\pi})
\land\square(\textit{pos}_{\alpha'}\leftrightarrow\textit{neg}_{2,\pi})
\land\square(\textit{pos}_{\beta}\leftrightarrow\textit{pos}_\pi) & (2)\\
&\lor \square(\textit{pos}_\alpha\leftrightarrow\textit{neg}_{1,\pi})
\land\square(\textit{neg}_{1,\alpha'}\leftrightarrow\textit{pos}_\pi)
\land\square(\textit{neg}_{1,\beta}\leftrightarrow\textit{neg}_{2,\pi}) & (3)\\
&\lor \square(\textit{pos}_\alpha\leftrightarrow\textit{neg}_{2,\pi})
\land\square(\textit{neg}_{1,\alpha'}\leftrightarrow\textit{pos}_\pi)
\land\square(\textit{neg}_{1,\beta}\leftrightarrow\textit{neg}_{1,\pi})) \rightarrow \beta~\triangleright A & (4)
\end{align*}
$\varphi$ quantifies over the smallest set $A$, representing the current assignment to $h$, and is only satisfied if there is no variable for which both the positive the negative traces are in $A$. 
$\varphi_{fp}$ ensures that a trace $\beta$ is only added to $A$ (line~4, right) if there exist assignment traces $\alpha, \alpha'$ which are already in $A$ and a clause trace~$\pi$ (line~1), which satisfy the condition in lines 2-4.
This condition (lines 2-4) holds exactly if $\alpha$ and~$\alpha'$ represent an assignment to two of the variables in the clause $c$ represented by $\pi$, such that $c$ can only be satisfied if the third variable, represented by $\beta$, is assigned according to~$\beta$.
This is the case only if $\alpha$ and $\alpha'$ represent a positive assignment to both negative variables and $\beta$ represents a positive assignment to the positive variable (line 2) or $\alpha$ represents a positive assignment to one of the negative variables and $\alpha'$ represents a negative assignment to the positive variable and $\beta$ represents a negative assignment to the other negative variable (lines 3-4).
Then $h$ is satisfiable iff $\varphi$ holds in $K$, and since $K$ is polynomial in the size of~$h$, we conclude the~proof.\qed
\begin{figure}[t]
\begin{minipage}{.47\textwidth}
    \resizebox{1\linewidth}{!}{
\begin{tikzpicture}[shorten >=1pt,node distance=2cm,on grid,auto]
  \tikzstyle{every state}=[fill={rgb:black,0;white,10},minimum size=40pt,inner sep=0pt]
  \node[state,initial, ] (s_0)  {$\{\}$};
  \node[state, node distance=3cm] (x31) [below of=s_0] {$\{\textit{pos}, a\}$};
  \node[state, node distance=1.5cm] (nx21) [left of=x31] {$\{\}$};
  \node[state, node distance=1.5cm] (x21) [left of=nx21] {$\{\}$};
  \node[state, node distance=1.5cm] (nx11) [left of=x21] {$\{\textit{neg}_1\}$};
  \node[state, node distance=1.5cm] (x11)  [left of=nx11] {$\{\textit{pos}\}$};
  \node[state, node distance=1.5cm] (nx31) [right of=x31] {$\{\textit{neg}_1\}$};
  \node[state, node distance=1.5cm] (x41) [right of=nx31] {$\{\}$};
  \node[state, node distance=1.5cm] (nx41) [right of=x41] {$\{a\}$};
  \node[state, node distance=1.5cm, align=center] (c11) [right of=nx41] {$\{\textit{neg}_1,$\\$\textit{neg}_2, c\}$};
  \node[state, node distance=1.5cm] (c21) [right of=c11] {$\{\textit{neg}_1, c\}$};
  \node[state] (x32) [below of=x31] {$\{\textit{pos}\}$};
  \node[state] (nx22) [below of=nx21] {$\{\textit{neg}_1\}$};
  \node[state] (x22) [below of=x21] {$\{\textit{pos}\}$};
  \node[state] (nx12) [below of=nx11] {$\{\}$};
  \node[state] (x12) [below of=x11] {$\{\}$};
  \node[state] (nx32) [below of=nx31] {$\{\textit{neg}_1\}$};
  \node[state] (x42) [below of=x41] {$\{\}$};
  \node[state] (nx42) [below of=nx41] {$\{\}$};
  \node[state, align=center] (c12) [below of=c11] {$\{\textit{neg}_2$\\$\textit{pos}\}$};
  \node[state] (c22) [below of=c21] {$\{\textit{neg}_2\}$};
  \node[state] (x33) [below of=x32] {$\{\}$};
  \node[state] (nx23) [below of=nx22] {$\{\}$};
  \node[state] (x23) [below of=x22] {$\{\}$};
  \node[state] (nx13) [below of=nx12] {$\{\}$};
  \node[state] (x13) [below of=x12] {$\{\}$};
  \node[state] (nx33) [below of=nx32] {$\{\}$};
  \node[state] (x43) [below of=x42] {$\{\textit{pos}\}$};
  \node[state] (nx43) [below of=nx42] {$\{\textit{neg}_1\}$};
  \node[state] (c13) [below of=c12] {$\{\}$};
  \node[state] (c23) [below of=c22] {$\{\textit{pos}\}$};
  \path[->]
  (s_0) edge  node {} (x31)
  (s_0) edge  node {} (nx21)
  (s_0) edge  node {} (x21)
  (s_0) edge  node {} (nx11)
  (s_0) edge [bend right=10] node {} (x11)
  (s_0) edge  node {} (nx31)
  (s_0) edge  node {} (x41)
  (s_0) edge  node {} (nx41)
  (s_0) edge [bend left=5] node {} (c11)
  (s_0) edge [bend left=10] node {} (c21);
  \path[->]
  (x31)  edge  node {} (x32)
  (nx21) edge  node {} (nx22)
  (x21)  edge  node {} (x22)
  (nx11) edge  node {} (nx12)
  (x11)  edge  node {} (x12)
  (nx31) edge  node {} (nx32)
  (x41)  edge  node {} (x42)
  (nx41) edge  node {} (nx42)
  (c11)  edge  node {} (c12)
  (c21)  edge  node {} (c22);
  \path[->]
  (x32)  edge  node {} (x33)
  (nx22) edge  node {} (nx23)
  (x22)  edge  node {} (x23)
  (nx12) edge  node {} (nx13)
  (x12)  edge  node {} (x13)
  (nx32) edge  node {} (nx33)
  (x42)  edge  node {} (x43)
  (nx42) edge  node {} (nx43)
  (c12)  edge  node {} (c13)
  (c22)  edge  node {} (c23);
  \path[->]
  (x33)  edge [loop below] node {} ( )
  (nx23) edge [loop below] node {} ( )
  (x23)  edge [loop below] node {} ( )
  (nx13)  edge [loop below] node {} ( )
  (x13)  edge [loop below] node {} ( )
  (nx33) edge [loop below] node {} ( )
  (x43) edge [loop below] node {} ( )
  (nx43)  edge [loop below] node {} ( )
  (c13)  edge [loop below] node {} ( )
  (c23)  edge [loop below] node {} ( );
\end{tikzpicture}
}
\caption{The tree-structure for the Horn-formula $(\neg x_1\lor\neg\top\lor x_2)\land(\neg x_1\lor\neg x_2\lor\bot)$. The two rightmost branches represent the two clauses. The rest of the branches (from left to right) represent positive and negative values for $x_1, x_2$, $\top$ (i.e., $x_3$), and $\bot$  (i.e., $x_4$).}
\label{proofs:fppcomplete:example}
\end{minipage}\hfill
\begin{minipage}{.47\textwidth}
    \begin{center}
\resizebox{1\linewidth}{!}{
    \begin{tikzpicture}[shorten >=1pt,node distance=2cm,on grid,auto]
      \tikzstyle{every state}=[fill={rgb:black,0;white,10},minimum size=40pt,inner sep=0pt]
    
      \node[state,initial, ] (s_0)  {$\{\}$};
      \node[state, align=center, node distance=3cm] (x21) [below of=s_0] {$\{\}$};
      \node[state, node distance=1.75cm] (nx11) [left of=x21] {$\{\textit{neg}\}$};
      \node[state, node distance=1.5cm] (x11) [left of=nx11] {$\{\textit{pos}\}$};
      \node[state, node distance=1.75cm] (q21) [left of=x11] {$\{\}$};
      \node[state, node distance=1.5cm] (q11) [left of=q21] {$\{q, v\}$};
      \node[state, node distance=1.5cm] (nx21) [right of=x21] {$\{\}$};
      \node[state, align=center, node distance=1.75cm] (no11) [right of=nx21] {$\{\textit{eneg},f\}$};
      \node[state, align=center, node distance=1.5cm] (o11) [right of=no11] {$\{\textit{epos},f\}$};
      \node[state, node distance=1.5cm] (o21) [right of=o11] {$\{f\}$};
    
      \node[state] (x22) [below of=x21] {$\{\textit{pos}\}$};
      \node[state] (nx12) [below of=nx11] {$\{\}$};
      \node[state] (x12) [below of=x11] {$\{\}$};
      \node[state] (q22) [below of=q21] {$\{q, v\}$};
      \node[state] (q12) [below of=q11] {$\{\}$};
      \node[state] (nx22) [below of=nx21] {$\{\textit{neg}\}$};
      \node[state] (no12) [below of=no11] {$\{\textit{eneg'}\}$};
      \node[state] (o12) [below of=o11] {$\{\}$};
      \node[state] (o22) [below of=o21] {$\{\textit{epos}\}$};
    
      \node[state] (x23) [below of=x22] {$\{\}$};
      \node[state] (nx13) [below of=nx12] {$\{\}$};
      \node[state] (x13) [below of=x12] {$\{\}$};
      \node[state] (q23) [below of=q22] {$\{\}$};
      \node[state] (q13) [below of=q12] {$\{\}$};
      \node[state] (nx23) [below of=nx22] {$\{\}$};
      \node[state] (no13) [below of=no12] {$\{\textit{neg}\}$};
      \node[state] (o13) [below of=o12] {$\{\textit{pos}\}$};
      \node[state] (o23) [below of=o22] {$\{\textit{pos}\}$};

      \path[->]
      (s_0) edge  node {} (x21)
      (s_0) edge  node {} (nx11)
      (s_0) edge  node {} (x11)
      (s_0) edge  node {} (q21)
      (s_0) edge [bend right=10] node {} (q11)
      (s_0) edge  node {} (nx21)
      (s_0) edge  node {} (no11)
      (s_0) edge  node {} (o11)
      (s_0) edge [bend left=10] node {} (o21);
    
      \path[->]
      (x21)  edge node {} (x22)
      (nx11) edge node {} (nx12)
      (x11)  edge node {} (x12)
      (q21)  edge node {} (q22)
      (q11)  edge node {} (q12)
      (nx21) edge node {} (nx22)
      (no11) edge node {} (no12)
      (o11)  edge node {} (o12)
      (o21)  edge node {} (o22);
    
      \path[->]
      (x22)  edge node {} (x23)
      (nx12) edge node {} (nx13)
      (x12)  edge node {} (x13)
      (q22)  edge node {} (q23)
      (q12)  edge node {} (q13)
      (nx22) edge node {} (nx23)
      (no12) edge node {} (no13)
      (o12)  edge node {} (o13)
      (o22)  edge node {} (o23);
    
      \path[->]
      (x23)  edge [loop below] node {} ( )
      (nx13) edge [loop below] node {} ( )
      (x13)  edge [loop below] node {} ( )
      (q23)  edge [loop below] node {} ( )
      (q13)  edge [loop below] node {} ( )
      (nx23) edge [loop below] node {} ( )
      (no13) edge [loop below] node {} ( )
      (o13)  edge [loop below] node {} ( )
      (o23)  edge [loop below] node {} ( );
    \end{tikzpicture}
}
\end{center}
\caption{Kripke structure built for the QBF formula $\forall x_1.\exists x_2. x_1\lor x_2$}
\label{proofs:hierarchycomplete:example}
\end{minipage}
\end{figure}

\end{proof}

	\section{EXP-completeness of \mcfpfpdag}\label{proofs:fpexpcomplete}
In this section, we show that \mcfpfpdag is EXP-complete.
The full proofs for \Cref{proofs:fpexpcomplete:inexp} and Thm.~\ref{proofs:fpexpcomplete:complete} can be found in App.~\ref{app:fpexpcomplete}.
\begin{observation}
\label{proofs:dagremark}
An acyclic Kripke structure with $n$ states has at most $2^n$ traces.
\end{observation}

\begin{lemma}\label{proofs:fpexpcomplete:inexp}
\mcfpfpdag is in EXP. 
\end{lemma}
\begin{proof}[Sketch]
Let $K$ be an acyclic Kripke structure. 
We can unroll $K$  to a tree-structure~$T$, in which the initial state branches to the $2^n$ traces of $K$ (cf.~Obs.~\ref{proofs:dagremark}): the  initial state of $T$ has $2^n$ transitions, and each branch is a linear branch corresponding to one trace of $K$.
We can then apply the algorithm for \mcfpfptree which is polynomial in the size of $T$, and exponential in the size of $K$. \qed
\end{proof}

\begin{theorem}\label{proofs:fpexpcomplete:complete}
\mcfpfpdag is EXP-complete.
\end{theorem}

\begin{proof}[Sketch]
We provide a polynomial time reduction from the Succinct Circuit Value problem~\cite{CompuComplexity}.
We only provide intuition here. For the full definition of the Succinct Circuit Value problem, and for the full proof, 
see App.~\ref{app:fullproof:fpexpcomplete}.

In the Succinct Circuit Value problem we are given a succinct Boolean representation $\mathcal{C}_S$ of a Boolean circuit $\mathcal{C}$.
The circuit $\mathcal{C}$ has no input gate and a single output gate, and the decision problem is to
decide, given the succinct representation  $\mathcal{C}_S$, whether the output gate of  $\mathcal{C}$ is evaluated to true.
The inputs to the circuit~$\mathcal{C}_S$ are two binary numbers $i$ and $j$, and its outputs are two numbers $q$ and $r$,
where: $i$ and $q$ are indices of gates in $\mathcal{C}$ such that $q$ is the $j$-th neighbor of gate $i$; $r$ represents which kind of gate $i$ is.
This results in an encoding where every input-output pair of $\mathcal{C}_S$ represents exactly one edge in $\mathcal{C}$.

Now, given a succinct circuit $\mathcal{C}_S$, 
we build a Fixpoint Hyper$^2$LTL$_{fp}$ formula $\varphi$ (that does not depend on~$\mathcal C$ or $\mathcal{C}_S$) and a Kripke structure $K$ such that $K\models\varphi$ if and only if $\mathcal{C}$ evaluates to true.
The reduction consists of three phases, A, B, and C, constructing subformulas that define sets $X, Y$, and $Z$.~The first two phases collect the output of $\mathcal{C}_S$ for every input, and the last phase solves the Circuit Value problem on $\mathcal{C}$.
At each phase, we describe both subformulas and substructures, and explain how they correspond. However, the construction of the formulas and the construction of the structures are independent.

Phase A starts by evaluating $\mathcal{C}_S$ under all possible inputs.
For that, our fixpoint formula defines a set $X$ containing traces that represent the values for all edges in $\mathcal{C}_S$ for all possible inputs. 
For each edge  $(u, v)$ in $\mathcal{C}_S$ we add a substructure in $K$, in which 
for every possible input $e$ to $\mathcal{C}_S$ there are two traces encoding~$e$ and $(u, v)$, one is labeled $\textit{pos}$ and one is labeled $\textit{neg}$.
Both traces additionally encode all relevant information about $(u, v)$ such as the gates that it connects and their kinds.
The formula $\varphi$ then defines a set $X$ such that 
a trace labeled with \textit{pos}, encoding input $e$ and representing edge $(u, v)$ of $\mathcal{C}_S$ is contained in $X$ if $(u, v)$ carries positive value when $e$ is given as input to $\mathcal{C}_S$.
Similarly, a trace labeled with \textit{neg},  encoding input $e$ and representing edge $(u, v)$ of $\mathcal{C}_S$ is contained in $X$ if $(u, v)$ carries negative value when $e$ is given as input to $\mathcal{C}_S$.
Thus, $X$ captures the complete state of $\mathcal{C}_S$ under each possible input.

In phase B, we build a set $Y$ that contains, for each input, a trace encoding the input and the corresponding output of $\mathcal{C}_S$.
To do so, we add a substructure to $K$ containing branches encoding arbitrary values with \textit{inp} and \textit{outp}, and can be labeled with \textit{incomplete} at some point.
A trace encoding $e$ with \textit{inp}, $o$ with \textit{outp} and that visits a state labeled with \textit{incomplete} at the $k+1$-st state, encodes that for input $e$, the first $k$ output bits of $\mathcal{C}_S$ are the first $k$ bits of~$o$.
The constructed Fixpoint Hyper$^2$LTL$_{fp}$ formula incrementally adds traces whose \textit{incomplete} label appears later:
Let~$\pi$ be a trace encoding $e$ with \textit{inp}, $o$ with \textit{outp} and visiting a state labeled with \textit{incomplete} at step $k+1$.
We add $\pi$ to $Y$ if there exists a different trace $\pi'\in Y$ that visits a state labeled with \textit{incomplete} at the $k$-th step and encoding the same $o$ in the first $k-1$ steps and the same $e$ as $\pi$.
Additionally, the $k$-th output bit on $\pi$ has to be the $k$-th output bit of $\mathcal{C}_S$ on input~$e$.
Therefore, we additionally require that there is a trace in $X$ that indicates that the $k$-th output gate of $\mathcal{C}_S$ for input $e$ has the same value as the $k$-th bit encoded with $o$ on~$\pi$.
Complete input-output pairs of $\mathcal{C}_S$ can be found on all traces in $Y$ where the \textit{incomplete} mark is pushed the furthest.

In phase C, which is almost analogous to phase A, 
we are left to evaluate the circuit $\mathcal{C}$.
We add another substructure to the structure~$K$, annotating traces from $Y$ with \textit{pos} and \textit{neg}.
The fixpoint formula
defines a set $Z$ representing the state of $\mathcal{C}$.
It adds a trace of~$Y$ representing an edge $(u, v)$ of $\mathcal{C}$ and labeled with \textit{pos} if the edge carries a positive value in $\mathcal{C}$, and the trace labeled with \textit{neg} otherwise.
As $\mathcal{C}$ has no inputs, we do not encode an input to $\mathcal{C}$ on the traces.

The final, second-order quantifier free part of the constructed Fixpoint Hyper$^2$LTL$_{fp}$ formula, is satisfied if and only if there exists a trace in $Z$ indicating that the edge leaving the output gate of $\mathcal{C}$ carries a positive value.\qed
\end{proof}

	\section{The Complexity of \mcfulltree}
We now turn to study the complexity of the model checking problem for the full logic Hyper$^2$LTL, starting with the analysis of tree-shaped models. 

\begin{theorem}\label{proofs:hierarchycomplete:complete}
 The problem of  \textsc{MC[$\Sigma_k$-Hyper$^2$LTL, tree]} is $\Sigma^p_{k+1}$-complete, and the problem of \textsc{MC[$\Pi_k$-Hyper$^2$LTL, tree]}
   $\Pi^p_{k+1}$-complete.
\end{theorem}

We prove Thm.~\ref{proofs:hierarchycomplete:complete} by proving a sequence of claims: we show the $\Sigma^p_{k+1}$ (resp. $\Pi^p_{k+1}$) completeness by first showing containment in the respective class.
Then we prove hardness in these classes by reducing from the QBF problem. \Cref{proofs:hierarchycomplete:lowerboundgen} below reasons about a general trace-set $T$. We will use it also in~\cref{proofs:exphierarchy}, to prove containment in $\Sigma^{EXP}_{k+1}$ and $\Pi^{EXP}_{k+1}$ . 

\begin{lemma}\label{proofs:hierarchycomplete:lowerboundgen}
For assignments $\Delta, \Pi$, a set of traces $T$, 
and a formula $\varphi$,
we have that $\Pi, \Delta \models_T \varphi$ can be decided in $\Sigma^p_{k+1}$ 
in the size of $T$ if 
$\varphi$ is in $\Sigma_k$-Hyper$^2$LTL,
and it can be decided in $\Pi^p_{k+1}$ in the size of $T$ if $\varphi$ is in $\Pi_k$-Hyper$^2$LTL.
\end{lemma}

\begin{observation}\label{proofs:hatvarphi}
For each formula $\varphi$ we can define a corresponding formula $\hat\varphi$ where each quantifier is converted into its dual and the innermost quantifier free formula is negated. The formula $\hat\varphi$ can be constructed from $\varphi$ in polynomial time, and it has the same number of (second-order) quantifier alternations as $\varphi$.  The outermost second-order quantifier in $\hat\varphi$ is the dual to the outermost second-order quantifier in $\varphi$: If $\Pi, \Delta \models\hat\varphi$ can be decided in some complexity class C, then $\Pi, \Delta \models \varphi$ can be decided in co-C and vice versa. 
\end{observation}

\begin{proof}[of \Cref{proofs:hierarchycomplete:lowerboundgen}]
By induction on the number $k$ of quantifier alternations:\\
\textbf{Base Case $k=0$:}
$\varphi$ contains only either existential or universal second-order quantifiers.
We first prove, by structural induction over $\varphi$, that if $\varphi$ only contains existential second-order quantifiers, then $\Pi, \Delta \models \varphi$ can be decided in $\Sigma^p_{1}$.
Our structural induction hypothesis is that for all $\Pi$ and $\Delta$, the model checking problem $\Pi, \Delta\models\varphi'$ can be decided in $\Sigma^p_{1}$. 
\begin{itemize}[nosep] 
    \item Base Case: For a quantifier-free formula $\varphi$, model checking can be decided in polynomial time (Cor.~\ref{cor:noquantifiers}).
    \item If $\varphi = \forall\pi\in X.\varphi'$ or $\varphi = \exists\pi\in X.\varphi'$: By induction hypothesis, for any $t$ we can decide in $\Sigma^p_{1}$ whether $\Pi[\pi\mapsto t], \Delta\models \varphi'$ holds.
    The number of traces in $\Delta(X)$ is at most linear in $|T|$.
    Evaluating $\Pi[\pi\mapsto t], \Delta\models \varphi'$ for all possible values of $t$ is therefore still in~$\Sigma^p_{1}$.
    \item If $\varphi = \exists X.\varphi'$: We need to show that $\Pi, \Delta \models \varphi$ can be decided in $\Sigma^p_{1}$.
    By induction hypothesis we know that, if $A$ is given, then $\Pi, \Delta[X\mapsto A] \models \varphi'$ can be decided in $\Sigma^p_{1}$.
    If we first non-deterministically guess $A$ by guessing for each of the linearly many traces whether it is in $A$ and then evaluate $\Pi, \Delta[X\mapsto A] \models \varphi'$ we can decide $\Pi, \Delta \models \varphi$ in $\Sigma^p_{1}$.
\end{itemize}
This concludes the analysis for formulas only containing existential second-order quantifiers.
If $\varphi$ only contains universal second-order quantifiers, we can use the dual formula $\hat\varphi$ which only contains existential second-order quantifiers,
and $\Pi, \Delta \models\hat\varphi$ can be decided in $\Sigma^p_{1}$. By Obs.~\ref{proofs:hatvarphi} we have that $\Pi, \Delta \models\varphi$ can be decided in $\Pi^p_{1}$. This concludes the base case. 

\noindent\textbf{Induction Step:} We need to show that if $\varphi$ contains $k+1$ second-order quantifier alternations, then $\Pi, \Delta \models\varphi$ can be decided in $\Sigma^p_{k+2}$ (resp. $\Pi^p_{k+2}$) in the size of~$T$.
By induction hypothesis we know that for any $\Delta, \Pi$ and $\varphi'$ with $k$ second-order quantifier alternations $\Pi, \Delta \models\varphi'$ can be decided in $\Sigma^p_{k+1}$ (resp. $\Pi^p_{k+1}$) in the size of $T$.
Again, we first prove the claim for a formula $\varphi$ whose outermost second-order quantifier is existential, by structural induction over~$\varphi$.
We have two induction hypotheses here:
First, the model checking problem for Hyper$^2$LTL with $k$ second-order quantifier alternations is in $\Sigma^p_{k+1}$ (resp. $\Pi^p_{k+1}$).
Second, for all $\Pi$ and $\Delta$ the model checking problem $\Pi, \Delta\models\varphi'$ can be decided in $\Sigma^p_{k+2}$.
\begin{itemize}[nosep]
\item  $\varphi$ is quantifier-free: it can be decided in polynomial time in $|T|$ (Cor.~\ref{cor:noquantifiers}).
\item If $\varphi = \forall\pi\in X.\varphi'$ or $\varphi = \exists\pi\in X.\varphi'$: By induction hypothesis, for any $t$ we can decide in $\Sigma^p_{k+2}$ whether $\Pi[\pi\mapsto t], \Delta\models \varphi'$ holds.
    The number of traces in $\Delta(X)$ is at most the size of $T$.
    Evaluating $\Pi[\pi\mapsto t], \Delta\models \varphi'$ for all possible values of $t$ is therefore still in $\Sigma^p_{k+2}$.
    \item If $\varphi = \exists X.\varphi'$: We need to show that $\Pi, \Delta \models \varphi$ can be decided in $\Sigma^p_{k+2}$ in the size of $T$.
    By induction hypothesis, if $\varphi'$ still has $k+1$ quantifier alternations then its outermost second-order quantifier is existential and $\Pi, \Delta[X\mapsto A] \models \varphi'$ can be decided in $\Sigma^p_{k+2}$ for any $A$.
    By the first induction hypothesis it follows that if $\varphi'$ has $k$ quantifier alternations then its outermost second-order quantifier is universal and $\Pi, \Delta[X\mapsto A] \models \varphi'$ can be decided in $\Pi^p_{k+1}$.
    For both cases holds the following:
    If we first non-deterministically guess $A$ by guessing for each of the linearly many traces whether it is in $A$ and then evaluate $\Pi, \Delta[X\mapsto A] \models \varphi'$ we can decide $\Pi, \Delta \models \varphi$ in $\Sigma^p_{k+2}$.
\end{itemize}
This concludes the analysis for formulas only containing existential second-order quantifiers.
If the outermost second-order quantifier of $\varphi$ is universal, then, by Obs.~\ref{proofs:hatvarphi}, the outermost second-order quantifier of $\hat\varphi$ is existential and $\Pi, \Delta \models\hat\varphi$ can be decided in $\Sigma^p_{k+2}$.
Therefore $\varphi$ can be decided in $\Pi^p_{k+2}$. \qed
\end{proof}

By Obs.~\ref{proofs:treeremark}, for a tree-shaped structure $K$,  
 the size of $\textit{Traces}(K)$ is polynomial in the size of $K$. 
 Therefore, from \Cref{proofs:hierarchycomplete:lowerboundgen} we have the following. 
 
\begin{corollary}\label{proofs:hierarchycomplete:lowerbound}
\textsc{MC[$\Sigma_k$-Hyper$^2$LTL,tree]}
is in $\Sigma^p_{k+1}$ in the size of $K$,
and \textsc{MC[$\Pi_k$-Hyper$^2$LTL,tree]}
 is in $\Pi^p_{k+1}$.
\end{corollary}

\begin{corollary}\label{proofs:hierarchycomplete:pspace}
\mcfulltree is in PSPACE.
\end{corollary}

\subsection{Lower bound}\label{QBF}

For the hardness result, 
We reduce from the Quantified Boolean Formula Problem (QBF), which we define as in \cite{computersintractability}.
\begin{definition}[Quantified Boolean Formula~\cite{computersintractability}]
	Let $k, m_1, \dots, m_{k+1}\in\mathbb N$, and let $y$ be
	$ y = 
		\mathbb \exists x_{1, 1}.\dots\mathbb \exists x_{1, m_1}. \mathbb \forall x_{2, 1}.\dots\mathbb \forall x_{2, m_2}. \dots \mathbb Q x_{k+1, 1}.\dots\mathbb Q x_{k+1, m_{k+1}} E
		$ or
        $y =\mathbb \forall x_{1, 1}.\dots\mathbb \forall x_{1, m_1}. \allowbreak \mathbb \exists x_{2, 1}.\dots\mathbb \exists x_{2, m_2}. \dots \allowbreak\mathbb Q x_{k+1, 1}.\dots\mathbb Q x_{k+1, m_{k+1}} E$
        where $E$ is an arbitrary Boolean formula over the variables $x_{1,1},\dots, x_{k+1, m_{k+1}}$, and
        where the last quantifier $\mathbb Q$ is the same at the first (i.e., existential for the first option and universal for the second) iff $k$ is even.
In both cases we say that $y$ has $k$ quantifier alternations.
	The QBF problem is then to decide whether $y$ is valid.
\end{definition}
The QBF problem 
with $k$
quantifier alternations 
and outermost existential quantifier 
is $\Sigma^p_{k+1}$-complete,
and it is $\Pi^p_{k+1}$-complete for formulas with outermost universal quantification~\cite{computersintractability}.
Thus, the QBF problem with $k$ quantifier alternations, behaves similarly to \mcfulltree with $k$ second-order quantifier alternations (cf. Cor.~\ref{proofs:hierarchycomplete:lowerbound}).
We use this intuition in the following where we reduce QBF with $k$ alternations to \mcfulltree with $k$ second-order quantifier alternations, thus providing a lower bound and proving Thm.~\ref{proofs:hierarchycomplete:complete}. 

\begin{proof}[of \Cref{proofs:hierarchycomplete:complete}]
We provide a polynomial time reduction from the QBF problem. 
Given a QBF formula $y$ with $k$ quantifier alternations, we build a tree-shaped Kripke structure $K$ in polynomial time,
and a Hyper$^2$LTL formula $\varphi$ which only depends on $k$ and on the first quantifier of $y$, such that $K\models\varphi$ iff $y$ is valid.
The $AP$ set of $K$ is $\{\textit{q,v, pos, neg, eneg, eneg', epos, epos', f}\}$.
The structure~$K$ consists of an initial state that is connected to several branches, which do not branch further.
Every branch has length $l$, where $l$ is the sum of the number of variables and Boolean operators in~$y$.
Note that $k\le l$.

We associate
every variable $x_{i,j}$ and sub-formula $e$ in $y$ with indices $N(x_{i,j})$ and $N(e)$, respectively. We index the variables and sub-expressions subsequently, such that the index of every variable is smaller than the index of every subexpression.
We then construct $K$ as follows. 
\begin{itemize}[nosep]
\item For every variable $x_{i, j}$ $K$ contains two branches:
One is marked with $\textit{pos}$ on its $N(x_{i,j})$-th state, and the other is marked with $\textit{neg}$ on its $N(x_{i,j})$-th state.
\item For sub-expressions $E_1\otimes E_2$, where $\otimes\in \{\wedge, \vee \}$,
$K$ contains three~branches:
\begin{itemize}[nosep]
	\item If $\otimes = \land$:  the first branch is labeled with $\textit{pos}$ on its $N(\!E_1\!\otimes\!E_2\!)$-th state, with $\textit{epos}$ on~its $N(E_1)$-th state, and with $\textit{epos'}$ on its $N(E_2)$-th state.
	The second branch is labeled with $\textit{neg}$ on its $N(\!E_1\!\otimes\!E_2\!)$-th state and with $\textit{eneg}$ on its $N(E_1)$-th state.
	The third branch is labeled with $\textit{neg}$ on its $N(\!E_1\!\otimes\!E_2\!)$-th state and with $\textit{eneg}$ on~its~$N(E_2)$-th~state. 
	\item 
    If $\otimes = \lor$: the first branch
    is labeled with $\textit{pos}$ on its $N(\!E_1\!\otimes\!E_2\!)$-th state and with $\textit{epos}$ on its $N(E_1)$-th state.
	The second branch is labeled with $\textit{pos}$ on its $N(\!E_1\!\otimes\!E_2\!)$-th state and with $\textit{epos}$ on its $N(E_2)$-th state.
 	The third branch 
    is labeled with $\textit{neg}$ on its $N(\!E_1\!\otimes\!E_2\!)$-th state, with $\textit{eneg}$ on its $N(E_1)$-th state, and with $\textit{eneg'}$~on~its~$N(E_2)$~th~state.
\end{itemize}
\item For a negated sub-expression $E = \neg E_1$ in $y$, $K$ has two branches:
One labeled with $\textit{pos}$ on its $N(E)$-th state and $\textit{eneg}$ on its $N(E_1)$-th state.
The second is labeled with $\textit{neg}$ on its $N(E)$-th state and $\textit{epos}$ on its $N(E_1)$-th state.
\item $K$ contains $k+1$ additional branches:
The $i$-th branch is labeled at the $i$-th position with atomic proposition $q$.
The $j$-th state of the $i$-th branch is labeled with $v$ if and only if there exists an $a\le m_i$ such that $N(x_{i, a}) = j$.
Therefore, the state is labeled if and only if there exists a variable in the $i$-th quantifier sequence that is associated with $j$.
\end{itemize}
Thus, each branch for a sub-formula of $y$ indicates its value with \textit{pos} or \textit{neg} if its sub-formulas have the value indicated by \textit{epos, epos', eneg} and \textit{eneg'}.
All branches of the outermost Boolean operator of $y$ are labeled with $f$ in the first state.
See Fig.~\ref{proofs:hierarchycomplete:example} for an example.
%for the construction in Fig.~\ref{proofs:hierarchycomplete:example}. 
We now describe the formula $\varphi$, which has $k+2$ second-order quantifications.
We divide $\varphi$ into sub-formulas, each contains exactly one second-order quantification:
$\varphi = \varphi_1 \oplus_1 \varphi_2 \oplus_2\dots \varphi_{k+1}\oplus_{k+1}\varphi_{k+2}$ where 
connective $\oplus_i$ is $\land$ if $\varphi_i$ contains an existential second-order quantifier and $\rightarrow$ otherwise.
For all $1 \le i \le k+1$ formula $\varphi_i$ is of the following form:
\begin{align*}
&\mathbb Q X_i.(\forall\pi\in \mathfrak G.\forall\pi'\in \mathfrak G.\exists!\pi''\in X_i.
\quad\LTLnext^i q_{\pi'}\land\LTLglobally(\neg\textit{epos}_\pi\land \neg\textit{eneg}_\pi\land\neg q_\pi)\\&\land\LTLfinally (v_{\pi'}\land(\textit{pos}_{\pi}\lor\textit{neg}_{\pi}))\quad\rightarrow\LTLglobally((\textit{pos}_\pi\lor\textit{neg}_\pi)\leftrightarrow (\textit{pos}_{\pi''}\lor\textit{neg}_{\pi''})))
\end{align*}
If, in $y$,
variables $x_{i, 1}$ to $x_{i, m_i}$ are quantified universally, then $\mathbb Q := \forall$.
If they are quantified existentially, then $\mathbb Q := \exists$.

An instantiation of $X_i$ satisfies $\varphi_i$ if it contains for each variable $x_{i, 1}$ to $x_{i, m_i}$ one trace assigning it positive or negative value.
For each pair of traces $\pi$ and $\pi'$, the set $X_i$ has to contain exactly one trace $\pi''$.
The sub-formula $\LTLnext^i q_{\pi'}$ restricts instantiations for $\pi'$ such that it can only be the trace that is labeled with $v$ at all positions corresponding to the relevant variables.
Instantiations for $\pi$ are restricted such that $\pi$ has to represent some variable $x_{i, 1}$ to $x_{i, m_i}$.
All instantiations for $\pi''$ therefore also represent the positive or negative assignment to some variable $x_{i, 1}$ to $x_{i, m_i}$.

The formula $\varphi_{k+2}$ should not introduce another quantifier alternation.
Therefore, its quantifier depends on the second-order quantifier in $\varphi_{k+1}$.
\begin{align*}
&\varphi_{k+2} =~  \mathbb Q Z. (\forall\pi_1\in X_1.\dots\forall\pi_{k+1}\in X_{k+1}. \pi_1\triangleright Z\land\dots\land\pi_{k+1}\triangleright Z) & (1)\\
&\land(\exists \pi\in Z. \LTLfinally f_\pi) & (2)\\
&\land(\forall \pi\in Z.\exists\pi'\in Z.\exists\pi''\in Z. \LTLglobally((\textit{epos}_\pi\leftrightarrow\textit{pos}_{\pi'})\land (\textit{eneg}_\pi\leftrightarrow\textit{neg}_{\pi'}) & (3)\\
&\land(\textit{epos'}_\pi\leftrightarrow\textit{pos}_{\pi''})\land (\textit{eneg'}_\pi\leftrightarrow\textit{neg}_{\pi''}))) & (4)\\
&\land(\forall\pi\in Z. \LTLglobally(\neg\textit{eneg}_\pi\land\neg\textit{epos}_\pi) \rightarrow\bigvee_{i=1}^{k+1}\pi\triangleright X_i) \oplus\exists\pi\in Z.\LTLfinally f_\pi\land\LTLfinally\textit{pos}_\pi & (5)
\end{align*}
If $\varphi_{k+1}$ is universally quantified, then $\mathbb Q = \forall$ and $\oplus = \rightarrow$.
If $\varphi_{k+1}$ is existentially quantified, then $\mathbb Q = \exists$ and $\oplus = \land$.
There is only one unique instantiation for $Z$ that satisfies $\varphi_{k+2}$.
The instantiation is a superset of $X_1, \ldots, X_{k+1}$ and contains a trace representing the whole quantifier free part of $y$ (line 2).
Similar to the proof of Thm.~\ref{proofs:fpexpcomplete:complete}, the instantiation may only contain a trace if the subformulas of the represented formula have the correct value (lines 3-4).
The instantiation thus evaluates $E$, and it remains to see whether there is a~trace~$\pi$ indicating that $E$ has a positive value (line 5, right).
While $\varphi$ is not given in formal Hyper$^2$LTL syntax it can be easily transformed into it. See App.~\ref{proofs:hierarchycomplete} for full proof.
\qed
\end{proof}

	\section{The Complexity of \mcfulldag}\label{proofs:exphierarchy}
We show that the problem of 
\textsc{MC[$\Sigma_k$-Hyper$^2$LTL, acyclic]}
is in $\Sigma^{EXP}_{k+1}$,
and \textsc{MC[$\Pi_k$-Hyper$^2$LTL, acyclic]}
is in $\Pi^{EXP}_{k+1}$.
From this, it follows that Hyper$^2$LTL model checking on acyclic models is decidable and is in EXPSPACE.
We then use reduction from acceptance problems for alternating Turing machines, 
and show that \mcfulldag is complete in these classes (Thm.~\ref{proofs:exphierarchy:complete}).
Containment follows directly from \Cref{proofs:hierarchycomplete:lowerboundgen} and Obs.~\ref{proofs:dagremark}. Therefore, we have:

\begin{corollary}\label{proofs:exphierarchy:lowerbound} %This is the uppterbound right?
The problem of \textsc{MC[$\Sigma_k$-Hyper$^2$LTL, acyclic]}
is in $\Sigma^{EXP}_{k+1}$, 
and the problem of 
\textsc{MC[$\Pi_k$-Hyper$^2$LTL, acyclic]}
is in $\Pi^{EXP}_{k+1}$.% for formulas with outermost second-order universal quantifier
\end{corollary}

\begin{corollary}\label{proofs:exphierarchy:expspace}
\mcfulldag is in EXPSPACE.
\end{corollary}

\begin{theorem}\label{proofs:exphierarchy:complete}
\textsc{MC[$\Sigma_k$-Hyper$^2$LTL, acyclic]}
is $\Sigma^{EXP}_{k+1}$-complete, and 
\textsc{MC[$\Pi_k$-Hyper$^2$LTL, acyclic]}
 is $\Pi^{EXP}_{k+1}$-complete.
\end{theorem}

\begin{proof}[Sketch]
We reduce (using a polynomial time reduction) from the corresponding acceptance problem of alternating Turing machines: 
given an alternating Turing machine $M$ and a $m\in\mathbb{N}$ (encoded in binary), decide whether $M$ accepts after $m$ steps.
We assume that $M$ alternates at most~$k$ times between universal and existential states (we call each such alternation \emph{alternation block}).
Additionally, its initial state is universal if and only if we prove completeness for~$\Pi_{k+1}^{EXP}$. See App.~\ref{app:exphierarchy} for a formal proof.

Given an alternating Turing machine $M$ and $m\in\mathbb{N}$ in binary, 
we build a structure $K$ and a formula $\varphi$ such that $K\models\varphi$ iff~$M$ accepts in at most $m$ steps.
$\varphi$ only depends on $k$ and whether the initial state of $M$ is universal or existential.

We encode $M$ into $K$ using the  
following substructures:
A substructure representing the traces for a counter, encoding all possible numbers (up to $m$); One substructure per transition in $M$, using encodings for states and tape letters; 
A substructure containing traces for the tape entries; A substructure containing traces representing that all tape cells are filled with blanks at time stamp 0; and A substructure containing exactly one trace encoding the number $m$. 
We use atomic propositions to encode all of the above, and to differentiate between the different substructures. 
All traces in $K$ have length $l+1$ which we define as the number of bits needed to encode $2\cdot m$.
Thus, $K$ encodes all valid transitions of $M$, some initial configuration, and additional traces such that every possible state and tape content can be represented as a set of traces.

The formula $\varphi$ collects valid computation fragments (i.e. a fragment that is taken within the same alternation block) into sets. It
quantifies over sets $X_1,\ldots, X_{k+1}$ simulating an execution of $M$.
The instantiations of a set $X_i$ are restricted such that each instantiation of $X_i$ is a valid execution branch through alternation block $i$ of $M$.
Each taken transition in the simulated execution is represented by a trace encoding the transition and a trace encoding the written letter to the tape.
Additionally, a computation fragment must be connectable to the last step of the previous fragment and the first step of the next fragment.
Further, it is only allowed to contain either transition starting in universal states or in existential but it may not alternate. We describe all of these using the formula $\varphi$ as we can freely quantify over the different sets. Then,
the answer to the model checking problem is equal to the question whether there is a computation step in the last fragment which visits the accepting state.

The alternations of $M$ are encoded in the quantification over the sets collecting computation fragments; 
If the computation fragment is universal then $\varphi$ accepts if all sets that represent a valid computation fragment lead to an accepting state.
If the computation fragment is existential then $\varphi$ accepts if there exists a set that represents a valid computation to an accepting state.\qed
\end{proof}

	\section{Complexity in the Size of the Model and Formula}\label{proofs:pspacecomplete}
In this section, we show that if the size of the formula is not assumed to be constant, then \mcfulltree as well as \mcfpfptree are PSPACE-complete in the combined size of the Kripke structure and formula.

\begin{lemma}\label{proofs1:trees:inpspace}
\mcfulltree is in PSPACE in the combined size of the structure and formula.
\end{lemma} 
\begin{proof}
Let $\varphi$ be a Hyper$^2$LTL formula of size $m$ and with $q$ quantifiers, and let $K$ be a tree-shaped Kripke structure with $n$ states.
We evaluate $\varphi$ recursively from the outermost quantifier to the innermost one.
To evaluate a quantifier~we iterate over every possible instantiation and recursively evaluate its inner formula.
We track instantiations of trace or trace-set quantifiers by marking the corresponding leaf or leaves respectively. 
When all traces are fixed, the inner, quantifier-free formula, can be decided in polynomial space~(Cor.~\ref{lemma:noquantifiers}).
To every time-point, we track at most $q$ quantifier instantiations.
Tracking a set of traces needs $\mathcal O(n\cdot\log(n))$ space, as there are at most $n$ traces in $K$ and tracking a trace needs $\mathcal O(\log n)$ space.
The space required to find the next set of traces or the next trace is polynomial in $n$ and constant in $m$.
Thus, we need at most $\mathcal O(q\cdot n\cdot\log(n))$ space - polynomial in the combined size of structure and formula. \qed
\end{proof}

\begin{theorem}\label{proofs:pspacecomplete:fpfp}
\mcfpfptree is PSPACE-complete in the combined input of the Kripke structure and formula.
\end{theorem}

\begin{theorem}\label{proofs:pspacecomplete:full}
\mcfulltree is PSPACE-complete in the combined input of the Kripke structure and formula.
\end{theorem}

For both theorems, 
containment in PSPACE follows from \Cref{proofs1:trees:inpspace}.
Since HyperLTL model checking on trees is PSPACE-complete in the combined input consisting of structure and formula~\cite{MonitorHLTL}, and since HyperLTL is subsumed by \fpfp and thus by Hyper$^2$LTL, 
Theorems~\ref{proofs:pspacecomplete:fpfp} and~\ref{proofs:pspacecomplete:full} follow.

	\section{Conclusion}\label{conclusion}
In this work, we analyzed the complexity of model checking Hyper$^2$LTL and \fpfp on finite structures. 
This problem is particularly relevant for monitoring purposes, in which tree-shaped and acyclic models are used as data structures to maintain the (repeatedly growing) set of traces seen so far.   
We showed that the model checking complexity for  \fpfp is polynomial in the number of traces in the model.
It follows from our analysis, that model checking  \fpfp over finite structures is not much more complex than HyperLTL model checking, despite the fact that \fpfp is much more expressive.\footnote{On tree-shaped models, HyperLTL model checking is L-complete~\cite{MonitorHLTL} and \fpfp model checking is P-complete. On acyclic models, model checking for HyperLTL is PSPACE-complete~\cite{MonitorHLTL} and for \fpfp is EXP-complete.} 
Contrary to Fixpoint Hyper$^2$LTL$_{fp}$, unrestricted Hyper$^2$LTL reasons about all possible sets of traces in the system, significantly increasing time complexity, while preserving space complexity. 
Our complexity results validate the efficiency of the \fpfp fragment, motivating its use in finite-trace settings, such as monitoring. 

The main contribution of this work is establishing lower bounds for the respective model checking problems. In particular, in the monitoring setting, the model (current set of traces) is changing repeatedly, and therefore we focused in complexity analysis in the size of the model.
As future work, we intend to complete the analysis of \Cref{proofs:pspacecomplete} to consider also the size of the formula. 

\bibliographystyle{splncs04}
\bibliography{main}
\newpage

\appendix
\section{Appendix: Full Proofs} \label{app}

\subsection{Full Proofs for \Cref{proofs:fppcomplete}: 
 \mcfpfptree is P-complete}\label{app:fppcomplete}

We show that the reduction described in \Cref{sec:hornsat}, from the Horn saisfiability problem, is correct.

\begin{lemma}\label{proofs:fppcomplete:correct}
Let $K$ and $\varphi$ be the structure and formula described in the proof of \Cref{proofs:fppcomplete:complete}, and $h$ be the respective Horn formula. 
$K\models\varphi$ iff $h$ is satisfiable.
\end{lemma}

\begin{proof}
We start by proving that the Horn formula is satisfiable if $K\models\varphi$.
If the answer to the model checking problem is positive, then we can construct a satisfying assignment from the set $A$.
If the trace that represents the positive form of a literal is in $A$, then we assign this literal to true, if the trace representing a negative form of a literal is in $A$, then we assign this literal to false.
If there is neither the positive nor the negative trace for a literal in $A$, we assign this literal to false and if the positive as well as the negative trace is contained in $A$, then the answer to the model checking problem would not be positive.

We prove by way of contradiction that this assignment is satisfying.
Otherwise, there is a clause $c$, which is not satisfied. We distinguish two cases:
\begin{itemize}[nosep]
\item If there are one or no literals in $c$ for which a trace representing positive assignment or a trace representing negative assignment is contained in $A$, then there are at least two literals in $c$ assigned to false, which satisfies $c$.
\item If there are two or three literals in $c$ for which the positive or negative trace is contained in $A$, then $A$ does contain traces $\alpha$ and $\alpha'$ for two of the literals in $c$, which do not satisfy $c$.
$\varphi_{fp}$ then enforces that the trace satisfying the remaining literal in $c$ is an element of $A$. This leads to an assignment that satisfies $c$ if only one of the traces for the literal is contained in $A$ or it leads to a resulting set $A$ which contains the positive as well as the negative trace for the remaining literal, which would then contradict the assumption that the answer to the model checking problem is positive.
\end{itemize}

It remains to show that $K\not\models\varphi$ implies that the given Horn formula is unsatisfiable.
We first prove that if a trace representing the assignment of one literal to true is in $A$, then there is no assignment satisfying the Horn formula, where this literal is assigned to false and vice versa.
More intuitively, we only add assignment traces whose assignment must be necessarily correct.
We do natural induction over the number of fixpoint iterations.

\begin{description}[nosep]
\item[Base Case:] $\varphi'_{fp}$ enforces that at least the traces assigning $\top$ to true and $\bot$ to false are part of $A$.
Per definition of the Horn-satisfiability problem, there is no valid assignment that assigns $\top$ to false or $\bot$ to true.
Because $A$ is empty in the beginning, $\varphi_{fp}$ does not add any traces in the first fixpoint iteration.

\item[Induction Case:] We consider a trace $\beta$ that is added to $A$ in the current iteration.
By induction hypothesis we know that the claim holds for each trace that is already in $A$ from previous iterations.
For trace $\beta$ to be added to $A$, there have to exist three traces:
$\pi\in \mathfrak G$ representing a horn clause $c$, and $\alpha, \alpha'\in A$ representing two assignments to literals.
Instantiating $\varphi_{fp}$ with these four traces must satisfies it.
The traces satisfy $\varphi_{fp}$ exactly iff $\alpha$ and $\alpha'$ represent an assignment for two of the literals in $c$ in such a way, that $c$ can only be satisfied if the third literal is assigned according to $\beta$.
By the induction hypothesis and the fact that $\alpha$ and $\alpha'$ were added to $A$ in previous iterations, we know that there does not exist a valid assignment satisfying the Horn formula, which conflicts with $\alpha$ or $\alpha'$. That means that there is also no satisfying assignment conflicting $\beta$.
\end{description}

The answer to the model checking problem is only negative if there exist two traces in $A$ which represent the assignment of the same literal to true and to false.
We can now conclude that in this case, there is no satisfying assignment to the Horn formula which does not assign true as well as false to the same literal, which would result in an invalid assignment.
\end{proof}

\begin{lemma}\label{proofs:fppcomplete:inl}
The given reduction needs space that is logarithmic in the size of $h$.
\end{lemma}
\begin{proof}
To show that this reduction only needs space which is logarithmic to the size of $h$, let $k$ be the highest number of a literal in the Horn formula and let $n$ be the number of clauses in the Horn formula.
The Kripke structure has $n+4+2k$ many branches with a depth of $\lceil\log_2(k+3)\rceil$.
These branches can be built sequentially by only remembering in which branch we are and at what depth.
This needs only $\lceil\log_2(n+4+2k)\rceil$ or $\lceil\log_2(\lceil\log_2(k+3)\rceil)\rceil$ space respectively.
With the first counter we can compute the corresponding literal or find the corresponding clause in the input and with the second counter we decide if the binary representation of the relevant literal(s) has a 1 at this position and if the current node has to be a leaf.
\end{proof}

\subsection{Full Proofs for \Cref{proofs:fpexpcomplete}:  \mcfpfpdag is EXP-complete}\label{app:fpexpcomplete}

\paragraph*{\Cref{proofs:fpexpcomplete:inexp} (restated)}
\mcfpfpdag is in EXP.
\begin{proof}
We reduce \mcfpfpdag to \mcfpfptree by unrolling the acyclic model to a tree.
An acyclic model can not have more than $2^{p(n)}$ traces, where $p(n)$ is some polynomial in the number of states of the model (Observation~\ref{proofs:dagremark}).
We can iterate over all traces in a manner similar to DFS.

Let $K$ be the given acyclic Kripke structure with $n$ states and $t \le2^{p(n)}$ traces and let $\varphi$ be the given \fpfp formula.
We construct a tree-shaped Kripke structure $T$ which contains all traces in $K$.
We iterate over all traces in $K$ and add for each trace in $K$ a new branch with equally labeled states to $T$.
This can be done in exponential time because $t$ is exponential in $n$ and adding a trace to $T$ does not take longer than the length of the trace which is bound by $n$.
If $|T|$ is the number of states of $T$, then $|T|\le n\cdot t \le n\cdot 2^{p(n)}$.

Because the sets of traces in $K$ and $T$ are equal $K\models\varphi$ holds exactly if $T\models\varphi$ holds.
Therefore we can use the algorithm from \Cref{proofs:fppcomplete:inp} whose runtime is polynomial in the size of the given structure.
Let $p'$ be the polynomial upper-bounding the runtime of this algorithm.
If we run this algorithm on $T$ its runtime can be expressed as $p'(|T|)\le p'(2^{p(n)})$ which implies the existence of a polynomial $p''$ such that $p'(2^{p(n)})\le2^{p''(n)}$.

Since the construction of $T$ can be done in exponential time and $2^{p''(n)}$ is an upper-bound for the runtime of the described algorithm, \mcfpfpdag can be solved in exponential time.
\end{proof}

\subsubsection{EXP-hardness of \mcfpfpdag}\label{app:exphard}

To prove that \mcfpfpdag is EXP-hard (\Cref{proofs:fpexpcomplete:complete}), we give a polynomial-time many-one reduction from Succinct Circuit Value.
We use very similar definitions as~\cite{CompuComplexity}.

\begin{definition}[Boolean Circuit]
A Boolean circuit is a directed acyclic Graph $\mathcal{C} = (V, E)$.
The nodes are called gates and every gate $u$ is annotated with a sort $\textit{sort}(u)$ which is TRUE, FALSE, AND, OR, NOT or INPUT.
If a gate is of sort TRUE, FALSE or INPUT, it has an indegree of 0.
Gates of sort NOT have indegree 1 and gates of sort AND or OR have an indegree of 2.
The gates with outdegree 0 are called output gates.
We index input gates and output gates.

The semantics of Boolean circuits for some input are defined by a function $\mathcal T$ that maps gates to Boolean values.
An input to a circuit is a bitstring (binary number) that assigns one of its bits to each input gate.
The value of $\mathcal T(u)$ for $u\in V$ is defined inductively.
\begin{itemize}[nosep]
\item If $u$ is an input gate, then $\mathcal T(u)$ is the value of the corresponding input bit.
\item If $u$ is of sort TRUE or FALSE, then $\mathcal T(u) = \top$ or $\mathcal T(u) = \bot$ respectively.
\item If $u$ is of sort NOT, then there exists exactly one gate $v$ such that $(v, u)\in E$.
	$\mathcal T(u)$ is then $\neg \mathcal T(v)$.
\item If $u$ is of sort OR, then there are exactly two gates $v, v'$ such that $(v, u), (v', u)\in E$.
	$\mathcal T(u)$ is defined as $\mathcal T(v)\lor \mathcal T(v')$.
\item If $u$ is of sort AND, then there are exactly two gates $v, v'$ such that $(v, u), (v', u)\in E$.
	$\mathcal T(u)$ is defined as $\mathcal T(v)\land \mathcal T(v')$.
\end{itemize}

We say an edge $(u, v)$ carries a positive value if $\mathcal T(u) = \top$ and we say that $(u, v)$ carries a negative value otherwise.
The output of a Boolean circuit is the bitstring (binary number) of the values assigned to the output gates by $\mathcal T$.
\end{definition}

\begin{definition}[Succinct Representation]
We say that a Boolean circuit $\mathcal{C}_S$ \emph{succinctly represents} a Boolean circuit $\mathcal{C}$ if $\mathcal{C}_S$ represents $\mathcal{C}$ as follows:
The gates of $\mathcal{C}$ are indexed by $1-m$.
The neighbors of a gate $u$ are all gates $v$ such that $(u, v)\in E$ or $(v, u)\in E$.
The neighbors of gate $u$ are also numbered such that the first up to two neighbors are the predecessors of $u$ and the rest are its successors.

Let $k$ be the smallest natural number satisfying $2^k - 1 \ge m$.
This allows us to encode every number between $0$ and $m$ in $k$ bits.
$\mathcal{C}_S$ has $2k$ input gates and $k + 3$ output~gates.

Let the input to $\mathcal{C}_S$ be of the form $i*j$, where $*$ is the concatenation of two bitstrings, and where $i$ and $j$ both consist of $k$ bits.
Let $q*r$ be the corresponding output where $q$ consists of $k$ bits and $r$ consists of three bits.
Let $u\in V$ be the gate with number $i$ and let $v\in V$ be the gate with number $q$.
The encoding is such that the $j$-th neighbor of $u$ is gate $v$ and $r$ encodes the kind of $u$.
If $u$ has no $j$-th neighbor, then the neighbor is some fictitious gate with the number 0.

For a gate $u$, $\textit{sort}(u)$ is encoded in $r$ as follows:
$\textit{TRUE} \rightarrow 001, \textit{FALSE}\rightarrow 010, \textit{AND}\rightarrow 011, \textit{OR}\rightarrow 100, \textit{NOT}\rightarrow 101$.
\end{definition}

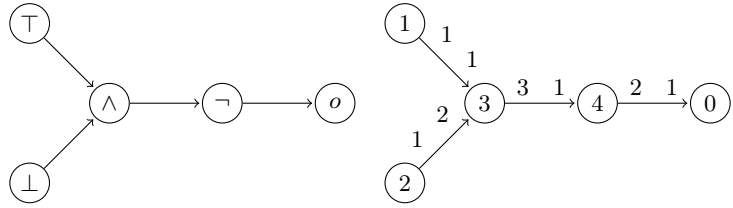
\begin{figure}
\begin{center}
\begin{subfigure}[c]{0.4\textwidth}
\begin{tikzpicture}[shorten >=1pt,node distance=1.5cm,on grid,auto]
  \tikzstyle{every state}=[fill={rgb:black,0;white,10},minimum size=15pt,inner sep=0pt]

  \node[state] (s0)  {$\top$};
  \node[state] (s2) [below right of=s0] {$\land$};
  \node[state] (s1) [below left of=s2] {$\bot$};
  \node[state] (s3) [right of=s2] {$\neg$};
  \node[state] (s4) [right of=s3] {$o$};

  \path[->]
  (s0) edge  node {} (s2)
  (s1) edge  node {} (s2)
  (s2) edge  node {} (s3)
  (s3) edge  node {} (s4);

\end{tikzpicture}
\end{subfigure}
\begin{subfigure}[c]{0.4\textwidth}
\begin{tikzpicture}[shorten >=1pt,node distance=1.5cm,on grid,auto]
  \tikzstyle{every state}=[fill={rgb:black,0;white,10},minimum size=15pt,inner sep=0pt]

  \node[state] (s0)  {$1$};
  \node[state] (s2) [below right of=s0] {$3$};
  \node[state] (s1) [below left of=s2] {$2$};
  \node[state] (s3) [right of=s2] {$4$};
  \node[state] (s4) [right of=s3] {$0$};

  \path[->]
  (s0) edge  node[near start] {1} node[near end] {1} (s2)
  (s1) edge  node[near start] {1} node[near end] {2} (s2)
  (s2) edge  node[near start] {3} node[near end] {1} (s3)
  (s3) edge  node[near start] {2} node[near end] {1} (s4);

\end{tikzpicture}
\end{subfigure}
\end{center}
\caption{Circuit $C$ with numbers for gates and neighbors.}
\label{proofs:fpexpcomplete:CS:circuit}
\end{figure}

\begin{definition}[Succinct Circuit Value Problem]
Given a Boolean circuit $\mathcal{C}_S$ with no input gates and which succinctly represents a Boolean circuit $\mathcal{C}$, decide whether $\mathcal T$ is true for the only output gate of $\mathcal{C}$.
\end{definition}

\begin{example}
Consider the circuit $\mathcal C$ in \Cref{proofs:fpexpcomplete:CS:circuit}.
It has no input gates and one output gate and is therefore suitable for the succinct encoding.
We fix a numbering for the nodes as indicated on the right.
Note that the output gate is required to have number $0$.
The numbering of the neighbors of each node are indicated on the outgoing edges.
The definition of a succinct representation requires us to give all neighbors connected with incoming edges smaller numbers than neighbors connected with outgoing edges.

A succinct representation $C_S$ of $C$ would be a circuit with $3 + 3$ inputs.
This is enough to encode the number of any gate (a number between 0 and 4) in the first three inputs and to encode the number of a neighbor in the next three inputs.
Because each node is connected to each other node at most once, a node can't have more than 4 neighbors in our example.
$C_S$ would have $3 + 3$ output gates.
The first three output gates encode the number of a gate and the next three output gates encode its kind.

For an input $i * j$ the succinct representation of $C$ will output the $j$-th neighbor of $i$ and its kind.
In particular, $C_S$ behaves as follows in this example:
\begin{itemize}
\item For input $1, 1$ it outputs $3, \textit{AND}$.
\item For input $2, 1$ it outputs $3, \textit{AND}$.
\item For input $3, 1$ it outputs $1, \textit{TRUE}$.
\item For input $3, 2$ it outputs $2, \textit{FALSE}$.
\item For input $3, 3$ it outputs $4, \textit{NOT}$.
\item For input $4, 1$ it outputs $3, \textit{AND}$.
\end{itemize}
The numbers for gates and neighbors are encoded in binary and the kind for each gate is encoded as given by the definition.
Importantly, $C_S$ will output $0$ in its first three bits for input $4, 2$ to indicate that the second neighbor of gate four is the output gate.
\end{example}

\paragraph{ Proof of \Cref{proofs:fpexpcomplete:complete}}\label{app:fullproof:fpexpcomplete}

We denote the set of all gates in $\mathcal{C}_S$ and $\mathcal{C}$ as $V(\mathcal{C}_S)$ and $V(\mathcal{C})$ respectively, and we denote the set of all edges in $\mathcal{C}_S$ and $\mathcal{C}$ as $E(\mathcal{C}_S)$ and $E(\mathcal{C})$ respectively.
Additionally, we denote the number of input gates of $\mathcal{C}_S$ as $|e|$ and the number of gates it has as $n$.

Note that every input-output pair of $\mathcal{C}_S$ represents an edge in $\mathcal{C}$.

To reduce Succinct Circuit Value to \mcfpfpdag we build an acyclic Kripke structure $K$ from the given circuit $\mathcal{C}_S$ and a fixed \fpfp formula $\varphi$ such that $K\models\varphi$ iff the Succinct Circuit Value Problem is true for $\mathcal{C}_S$.
The reduction consists of three phases.
The first two phases collect the output of $\mathcal{C}_S$ for every input and the last phase solves the Circuit Value problem on $\mathcal{C}$.

Intuitively, in phase A some set $X$ is built such that $X$ contains traces that represent the values for all edges in $\mathcal{C}_S$ for all possible inputs.
Therefore all traces have a binary number representing the input of $\mathcal{C}_S$ and two gate IDs encoded on them; one for the gate $u$ they are coming from and one for the gate $v$ they are going to.
Additionally, the trace encodes $\textit{sort}(v)$ and whether the edge carries positive or negative value for the given input.
For every edge $(u, v) \in E(\mathcal{C}_S)$ and every input $e$, $X$ contains exactly one trace encoding $e$, $(u, v)$ and $\textit{sort}(v)$.
$X$ contains the trace encoding $e$, $(u, v)$, $\textit{sort}(v)$ and positive value if edge $(u, v)$ carries positive value under input $e$.
If the edge carries negative value under input $e$, then $X$ contains the trace encoding $e$, $(u, v)$, $\textit{sort}(v)$ and negative value.

In phase B, we build a set $Y$ containing for each input one trace encoding the input and the corresponding output of $\mathcal{C}_S$.
For this, every trace in phase B encodes two binary numbers; one representing an input to $\mathcal{C}_S$ and one representing the corresponding output of $\mathcal{C}_S$.
We build $Y$ from set $Y'$, which contains traces where the last states don't encode output.
We call these traces with a partial output \emph{incomplete}.
The fixpoint iteration for $Y'$ adds, with every iteration traces that have one bit more output than the previous iteration.
In the last step of the iteration, the output is long enough to encode a complete input-output pair.
We can compute set $Y$ by taking all traces from $Y'$ that are not incomplete.

The succinct encoding is defined such that every edge in $\mathcal{C}$ is defined as two input-output pairs of $\mathcal C_S$.
This means that at this point all information about $\mathcal C$ is contained in set $Y$ and phase C can evaluate circuit $\mathcal C$.
All traces used in phase C are traces from $Y$ on which additionally appear the atomic propositions \textit{pos} and \textit{neg} to represent positive or negative value.
We construct set $Z$ such that for each trace in $Y$ the positive trace is included in $Z$ iff the represented edge carries positive value in $\mathcal{C}$.
If the represented edge is negative in $\mathcal{C}$, then we add the negative trace to $Z$.
At the last step, we check whether $Z$ contains a trace representing that the output edge carries a positive value.
The answer to the Succinct Circuit Value Problem is true iff this trace is in $Z$.

We now describe the reduction formally.
Given a circuit $\mathcal{C}_S$, we build a Kripke structure $K$ from $\mathcal{C}_S$.
The formula $\varphi$ is independent of $\mathcal{C}_S$.

\paragraph*{Phase A}\label{proofs:fpexpcomplete:phasea}

We construct the part of $K$ representing all edges in $\mathcal{C}_S$.
For this, we first modify $\mathcal{C}_S$.
First, we eliminate all binary gates that have two inputs coming from the same gate.
It follows, that there are no two edges $(u, v), (s, t)\in E(\mathcal{C}_S)$ with $u=s$ and $v = t$.
Then, we replace the TRUE and FALSE gates in $\mathcal{C}_S$ by circuits representing $p\land\neg p$ and  $p\lor\neg p$ respectively, where $p$ is some input gate.
This shortens the construction in phase C.
Further, we add to every gate producing an output of $\mathcal{C}_S$ an output edge going to a special gate of sort OUTPUT which we call \emph{output gate}.
Additionally, we give every gate $u$ in $\mathcal{C}_S$ a unique ID$>0$ denoted by $ID(u)$.
The IDs are distributed such that for every $(u, v)\in E(\mathcal{C}_S)$ holds $ID(u) < ID(v)$.
Finally, we calculate $l\in\mathbb N$ such that traces of length $l$ are long enough to encode one ID of a gate in $\mathcal{C}_S$ as well as inputs and outputs of $\mathcal{C}_S$. E.g. $l = \max(\lceil\log_2(n)\rceil, |e|, o)$ if $\mathcal{C}_S$ has $o$ output gates.

In the case where $|e| < l$ we add $l - |e|$ \emph{dummy input gates} to $\mathcal{C}_S$ which are not connected to any other gate.
Similarly, if $o < l$ we add $l - o$ \emph{dummy output gates} to $\mathcal{C}_S$ which always output false.
$ID(u)$ is undefined if $u$ is a dummy input gate or dummy output gate.
We call every input where the inputs to the dummy input gates are false \emph{valid}.
This ensures that inputs and outputs occupy the whole trace when encoding inputs or outputs of $\mathcal C_S$ on a trace of $K$.

For every edge in $\mathcal{C}_S$, we build two substructures connected to the initial state in $K$.
These substructures consist of $l$ layers with two states per layer.
The first structure contains only traces representing positive value on the edge and the second structure contains only traces representing negative value on the edge.
Both structure are able to encode arbitrary inputs to $\mathcal C_S$ on their traces.

Let $a_{u, v, i, j, \textit{pos}}$ be the $j$-th state in the $i$-th layer of one of the substructures representing the edge $(u, v)\in E(\mathcal{C}_S)$ and let $a_{u, v, i, j, \textit{neg}}$ be the same state in the other substructure for $(u, v)$.
Note that $i\in\{1,\dots, l\}$ and $j\in\{1, 2\}$.
In every substructure the two states of one layer are fully connected to the states of the next layer and the states of the last layer have self-loops.
There are no other transitions in these substructures (see \Cref{proofs:fpexpcomplete:phasea:example:one}).

\begin{figure}[t]

\begin{subfigure}{0.3\textwidth}
\resizebox{.6\linewidth}{!}{
\begin{tikzpicture}[shorten >=1pt,node distance=2cm,on grid,auto]
  \tikzstyle{every state}=[fill={rgb:black,0;white,10},minimum size=30pt,inner sep=0pt]

  \node[state,initial, ] (s_0)  {};
  \node[state] (a11) [below of=s_0] {$\{a,\\ \textit{not}\}$};
  \node[state, align=center] (a12) [left of=a11] {$\{a, \textit{inp}$\\$\textit{not}\}$};
  \node[state, align=center] (a21) [below of=a11] {$\{\textit{fromid},$\\$\textit{toid}\}$};
  \node[state, align=center] (a22) [below of=a12] {$\{\textit{fromid},$\\$\textit{toid},\textit{inp}\}$};
  \node[state] (a31) [below of=a21] {$\{\textit{toid}\}$};
  \node[state] (a32) [below of=a22] {$\{$\textit{inp},\\ \textit{toid}$\}$};

  \path[->]
  (s_0) edge  node {} (a11)
  (s_0) edge  node {} (a12)
  (a11) edge  node {} (a21)
  (a11) edge  node {} (a22)
  (a12) edge  node {} (a21)
  (a12) edge  node {} (a22)
  (a21) edge  node {} (a31)
  (a21) edge  node {} (a32)
  (a22) edge  node {} (a31)
  (a22) edge  node {} (a32)
  (a31) edge  [loop below] node {} ( )
  (a32) edge  [loop below] node {} ( );
\end{tikzpicture}
}
\caption{The positive substructure for an edge between gates $u$ and $v$ with $\textit{ID(u)}\!= \!2$,$ \textit{ID(v)} \!= \!6$ and $v$ being a NOT gate. Every state is additionally labeled with $\textit{pos}$.}
\label{proofs:fpexpcomplete:phasea:example:one}
\end{subfigure}
\hfill
\begin{subfigure}{0.3\textwidth}
\resizebox{.6\linewidth}{!}{
\begin{tikzpicture}[shorten >=1pt,node distance=2cm,on grid,auto]
  \tikzstyle{every state}=[fill={rgb:black,0;white,10},minimum size=30pt,inner sep=0pt]

  \node[state,initial, ] (s_0)  {};
  \node[state, align=center] (a11) [below of=s_0] {$\{a,\textit{or},$\\$\textit{input}\}$};
  \node[state, align=center] (a12) [left of=a11] {$\{a, \textit{inp}$\\$\textit{or, input}\}$};
  \node[state, align=center] (a22) [below of=a12] {$\{\textit{fromid},$\\$\textit{inp}\}$};
  \node[state] (a31) [below of=a21] {$\{\textit{toid}\}$};
  \node[state, align=center] (a32) [below of=a22] {$\{\textit{inp},$\\$\textit{toid}\}$};

  \path[->]
  (s_0) edge  node {} (a11)
  (s_0) edge  node {} (a12)
  (a11) edge  node {} (a22)
  (a12) edge  node {} (a22)
  (a22) edge  node {} (a31)
  (a22) edge  node {} (a32)
  (a31) edge  [loop below] node {} ( )
  (a32) edge  [loop below] node {} ( );
\end{tikzpicture}
}
\caption{The positive substructure for an edge between gates $u$ and $v$ with $\textit{ID(u)}\! = \!2$,$\textit{ID(v)} \!= \!4$ and $v$ being an OR gate. $u$ takes the second input bit. Every state is additionally labeled with $\textit{pos}$.}
\label{proofs:fpexpcomplete:phasea:example:two}
\end{subfigure}
\hfill
\begin{subfigure}{0.3\textwidth}
\resizebox{.6\linewidth}{!}{
\begin{tikzpicture}[shorten >=1pt,node distance=2cm,on grid,auto]
  \tikzstyle{every state}=[fill={rgb:black,0;white,10},minimum size=30pt,inner sep=0pt]
  \node[state,initial, ] (s_0)  {};
  \node[state, align=center] (a11) [below of=s_0] {$\{a,$\\$\textit{output}\}$};
  \node[state, align=center] (a12) [left of=a11] {$\{a, \textit{inp},$\\$\textit{output}\}$};
  \node[state, align=center] (a21) [below of=a11] {$\{\textit{fromid}\}$};
  \node[state, align=center] (a22) [below of=a12] {$\{\textit{inp},$\\$\textit{fromid}\}$};
  \node[state, align=center] (a31) [below of=a21] {$\{\textit{toid}$\\$\textit{bit}\}$};
  \node[state, align=center] (a32) [below of=a22] {$\{\textit{inp},\textit{bit}$\\$\textit{toid}\}$};

  \path[->]
  (s_0) edge  node {} (a11)
  (s_0) edge  node {} (a12)
  (a11) edge  node {} (a21)
  (a11) edge  node {} (a22)
  (a12) edge  node {} (a21)
  (a12) edge  node {} (a22)
  (a21) edge  node {} (a31)
  (a21) edge  node {} (a32)
  (a22) edge  node {} (a31)
  (a22) edge  node {} (a32)
  (a31) edge  [loop below] node {} ( )
  (a32) edge  [loop below] node {} ( );
\end{tikzpicture}
}
\caption{The negative substructure for an edge between gates $u$ and $v$ with $\textit{ID(u)} \!= \!2$,$ \textit{ID(v)} \!=\! 4$. Gate $v$ is the third output bit. Every state is additionally labeled with $\textit{neg}$.}
\label{proofs:fpexpcomplete:phasea:example:three}
\end{subfigure}
\caption{Substructures for phase A with $l=|e| = 3$.}
\label{proofs:fpexpcomplete:phasea:example}
\end{figure}
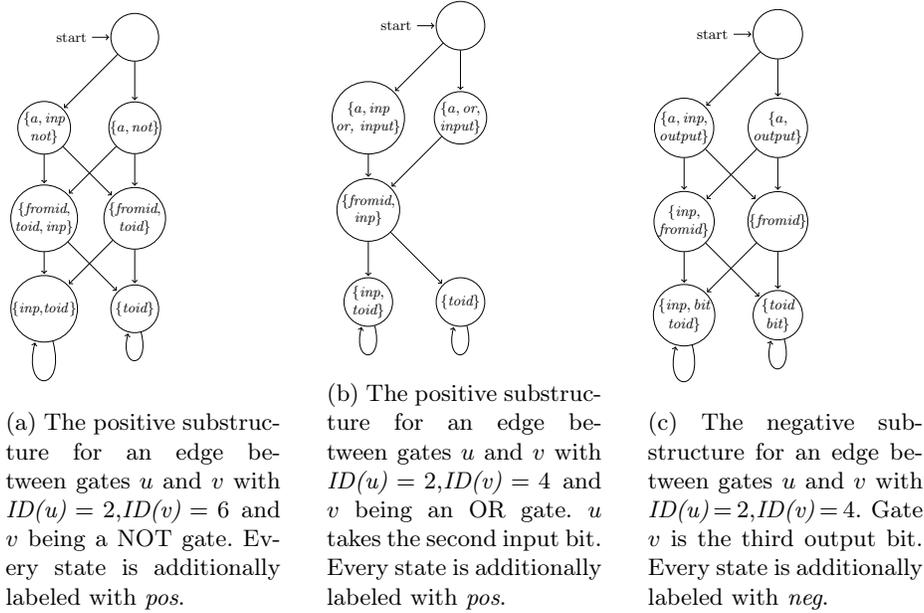

All states $a_{u, v, i, j,\textit{pos}}$ are labeled with $\textit{pos}$.
All states $a_{u, v, i, j,\textit{neg}}$ are labeled with $\textit{neg}$.
Additionally states $a_{u, v, 1, j, \textit{pos}}$ and $a_{u, v, 1, j, \textit{neg}}$ (the states in the first layer) are labeled with atomic proposition $a$ to represent that they are part of phase A.
We label the states $a_{u, v, i, j, \textit{pos}}$ and $a_{u, v, i, j, \textit{neg}}$ with the atomic proposition $\textit{fromid}$ iff the $i$-th bit of $ID(u)$ is 1 and equally with atomic proposition $\textit{toid}$ for $ID(v)$.
Therefore $u$ is encoded with \textit{fromid} and $v$ is encoded with \textit{toid} on every trace in every structure.

We use atomic propositions $\textit{output},\textit{and},\textit{or},\textit{not}$ to encode $\textit{sort}(v)$.
States $a_{u, v, 1, j, \textit{pos}}$ and $a_{u, v, 1, j, \textit{neg}}$ are labeled with one of these propositions.

States $a_{u, v, i, 2, \textit{pos}}$ and $a_{u, v, i, 2, \textit{neg}}$ with $i\le |e|$ are labeled with $\textit{inp}$.
This gives us traces encoding arbitrary bitstrings with $\textit{inp}$ which will be later used to encode inputs to $\mathcal C_S$.
An example of a completely labeled structure is given in \Cref{proofs:fpexpcomplete:phasea:example:one}.

If $u$ is the input gate which takes the $m$-th input, then states $a_{u, v, 1, j,\textit{pos}}$ and $a_{u, v, 1, j,\textit{neg}}$ are labeled with $\textit{input}$.
Additionally, there are no states $a_{u, v, m, 1, \textit{pos}}$ and $a_{u, v, m, 2, \textit{neg}}$.
By that we remove the states that allowed us to encode certain inputs to $\mathcal C_S$.
The positive structure for $(u, v)$ now does not contain any trace where the $m$-th input bit is negative and the negative structure does not contain any trace where the $m$-th input bit is positive.

For the edges going to the $m$-th output gate in $\mathcal{C}_S$, the states $a_{u, v, m, j, \textit{pos}}$ and $a_{u, v, m, j, \textit{pos}}$ are labeled with atomic proposition $\textit{bit}$ to mark the represented output bit.
Example substructures for edges leaving input gates and entering output gates are given in \Cref{proofs:fpexpcomplete:phasea:example:two} and \Cref{proofs:fpexpcomplete:phasea:example:three}.

For every dummy output gate responsible for the $m$-th output of $\mathcal{C}_S$ we build one substructure, which looks exactly like a negative substructure for a normal output gate but there are no gate IDs encoded on this substructure.
We label these substructures with atomic proposition $\textit{dummy}$ in the first states.

To summarize which traces are now part of $K$ and labeled with $a$ somewhere:
For each edge $(u, v)\in E(\mathcal{C}_S)$ and each valid input $e$ we have one trace labeled with $\textit{pos}$ and one trace labeled with $\textit{neg}$.
Both of these traces have the IDs of $u$ and $v$ as well as $e$ and the type of gate $v$ encoded on them.
Positive traces representing an edge leaving the $m$-th input gate only exist if the $m$-th bit of the encoded $e$ is a 1.
The corresponding negative traces only exist if the $m$-th bit of the encoded $e$ is a 0.
For each edge going to the $m$-th output gate the traces representing that edge are labeled with $\textit{bit}$ in their $m+1$-st state.
For every valid input $e$ and every dummy output gate, responsible of the $m$-th output, there are traces encoding $e$, labeled with $\textit{bit}$ at the $m+1$-st step and labeled with $\textit{neg}$.
These are the only traces labeled with atomic proposition $\textit{dummy}$.

The set $X$ should contain a trace representing edge $(u, v)\in E(\mathcal{C}_S)$ labeled with $\textit{pos}$ and input~$e$ iff the value of edge $(u, v)$ is positive for input $e$.
Similarly, the corresponding $\textit{neg}$ trace should be in $X$ iff the value of edge $(u, v)$ is negative for input $e$.
The fixpoint formula  $\varphi$ describing $X$ has different parts for each sort of gate, 
and is defined as follows: $\varphi = (X, \curlyvee, \varphi_{\textit{input}}\land \varphi_{\textit{and}}\land \varphi_{\textit{or}}\land \varphi_{\textit{not}}).\varphi_B$. 

Since AND and OR gates behave similarly, $\varphi_{\textit{and}}$  adds negative traces for OR gates and $\varphi_{\textit{or}}$  adds negative traces for AND gates.
This allows us to describe the formula more succinctly.
We start by describing $\varphi_{\textit{input}}$
which adds all traces representing edges leaving input gates to $X$.
\begin{equation*}
\varphi_{\textit{input}}=\forall\pi\in \mathfrak G.\LTLfinally a_\pi\land\LTLfinally \textit{input}_\pi \rightarrow \pi\triangleright X
\end{equation*}
In $\varphi_{\textit{and}}$, defined below,
we quantify over three traces: $\pi_2$ and $\pi_3$ are two different traces already in $X$, and $\pi_1$ is the trace we add if the step formula is satisfied.
To make sure that we only consider traces that are relevant to phase A, we check whether all three traces are labeled with $a$.
Additionally, all three traces should encode the same input to $\mathcal{C}_S$, and $\pi_2$ and $\pi_3$ should point to the gate where $\pi_1$ is coming from.
We add $\pi_1$ to set $X$ if, additionally, all three traces represent positive values on the edges and the gate $\pi_2$ is going to is an AND gate.
Alternatively, we add $\pi_1$ to $X$ if all three traces represent negative values on their edges and $\pi_2$ is going to an OR gate.
\begin{align*}
&\varphi_{\textit{and}}=(\forall\pi_1\in\mathfrak G.\forall\pi_2,\pi_3\in X.\LTLfinally(\pi_2\neq_{AP}\pi_3) ~\land\LTLfinally a_{\pi_1}\land\LTLfinally a_{\pi_2}\land \LTLfinally a_{\pi_3}\\
&\land\LTLglobally((\textit{inp}_{\pi_1}\leftrightarrow \textit{inp}_{\pi_2}) \land (\textit{inp}_{\pi_2}\leftrightarrow \textit{inp}_{\pi_3}))\land\LTLglobally((\textit{toid}_{\pi_2}\leftrightarrow\textit{toid}_{\pi_3}) \land(\textit{toid}_{\pi_3}\leftrightarrow \textit{fromid}_{\pi_1}))\\
&\land((\LTLfinally\textit{and}_{\pi_2}\land\LTLfinally\textit{pos}_{\pi_1}\land\LTLfinally\textit{pos}_{\pi_2}\land\LTLfinally\textit{pos}_{\pi_3})\lor(\LTLfinally\textit{or}_{\pi_2}\land\LTLfinally\textit{neg}_{\pi_1}\land\LTLfinally\textit{neg}_{\pi_2}\land\LTLfinally\textit{neg}_{\pi_3}))\to\pi_1\triangleright X)
\end{align*}

The formula $\varphi_{\textit{or}}$, defined below, adds positive traces for outgoing edges for OR gates with at least one positive input, and negative edges for AND gates with at least one negative input.
First, we check again whether both traces are from the substructures from phase A, have the same input encoded, and whether $\pi_2$ is an input to the gate for which $\pi_1$ is an output.
Then we add $\pi_1$ to $X$ if it represents a positive output to an OR gate with one positive input or if $\pi_1$ represents a negative output for an AND gate with one negative input.
\begin{align*}
&\varphi_{\textit{or}}=(\forall\pi_1\in\mathfrak G.\forall\pi_2\in X.~\LTLfinally a_{\pi_1}\land\LTLfinally a_{\pi_2}\land\LTLglobally(\textit{inp}_{\pi_1}\leftrightarrow \textit{inp}_{\pi_2})\land\LTLglobally(\textit{toid}_{\pi_2}\leftrightarrow \textit{fromid}_{\pi_1})\\
&\land((\LTLfinally\textit{or}_{\pi_2}\land\LTLfinally\textit{pos}_{\pi_1}\land\LTLfinally\textit{pos}_{\pi_2})\lor(\LTLfinally\textit{and}_{\pi_2}\land\LTLfinally\textit{neg}_{\pi_1}\land\LTLfinally\textit{neg}_{\pi_2}))\to\pi_1\triangleright X)
\end{align*}
Similarly, for NOT gates we add $\pi_1$ if it represents the opposite value of $\pi_2$:
\begin{align*}
&\varphi_{\textit{not}}=(\forall\pi_1\in\mathfrak G.\forall\pi_2\in X.~\LTLfinally a_{\pi_1}\land\LTLfinally a_{\pi_2}\land\LTLglobally(\textit{inp}_{\pi_1}\leftrightarrow \textit{inp}_{\pi_2})\land\LTLglobally(\textit{toid}_{\pi_2}\leftrightarrow \textit{fromid}_{\pi_1})\\
&\land((\LTLfinally\textit{not}_{\pi_2}\land\LTLfinally\textit{pos}_{\pi_1}\land\LTLfinally\textit{neg}_{\pi_2})\lor(\LTLfinally\textit{not}_{\pi_2}\land\LTLfinally\textit{neg}_{\pi_1}\land\LTLfinally\textit{pos}_{\pi_2}))\to\pi_1\triangleright X)
\end{align*}

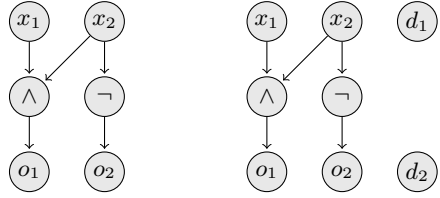
\begin{figure}[t]
\begin{center}
\begin{subfigure}[c]{0.25\textwidth}
\begin{tikzpicture}[shorten >=1pt,node distance=1cm,on grid,auto]
  \tikzstyle{every state}=[fill={rgb:black,1;white,10},minimum size=15pt,inner sep=0pt]

  \node[state] (s0)  {$x_1$};
  \node[state] (s1) [right of=s0] {$x_2$};
  \node[state] (s3) [below of=s0] {$\land$};
  \node[state] (s4) [below of=s1] {$\neg$};
  \node[state] (s5) [below of=s3] {$o_1$};
  \node[state] (s6) [below of=s4] {$o_2$};

  \path[->]
  (s0) edge  node {} (s3)
  (s1) edge  node {} (s3)
  (s1) edge  node {} (s4)
  (s3) edge  node {} (s5)
  (s4) edge  node {} (s6);

\end{tikzpicture}
\end{subfigure}
\begin{subfigure}[c]{0.25\textwidth}
\begin{tikzpicture}[shorten >=1pt,node distance=1cm,on grid,auto]
  \tikzstyle{every state}=[fill={rgb:black,1;white,10},minimum size=15pt,inner sep=0pt]

  \node[state] (s0)  {$x_1$};
  \node[state] (s1) [right of=s0] {$x_2$};
  \node[state] (s2) [right of=s1] {$d_1$};
  \node[state] (s3) [below of=s0] {$\land$};
  \node[state] (s4) [below of=s1] {$\neg$};
  \node[state] (s5) [below of=s3] {$o_1$};
  \node[state] (s6) [below of=s4] {$o_2$};
  \node[state] (s7) [right of=s6] {$d_2$};

  \path[->]
  (s0) edge  node {} (s3)
  (s1) edge  node {} (s3)
  (s1) edge  node {} (s4)
  (s3) edge  node {} (s5)
  (s4) edge  node {} (s6);

\end{tikzpicture}
\end{subfigure}
\end{center}
\caption{The circuit considered in \Cref{proofs:fpexpcomplete:phasea:example:circuittext}}
\label{proofs:fpexpcomplete:phasea:example:circuit}
\end{figure}

\begin{example}\label{proofs:fpexpcomplete:phasea:example:circuittext}
Consider the circuit $\mathcal{C}_S$ of \Cref{proofs:fpexpcomplete:phasea:example:circuit} (left).
A valid distribution of IDs is from the top-left to the bottom-right with IDs one to six.
There are six states, so we set $l = 3$.
The resulting circuit after adding dummy gates $d_1$ and $d_2$ is depicted in \Cref{proofs:fpexpcomplete:phasea:example:circuit} on the right.
We denote the states in this example by the symbol they are labeled with such that we have the states $x_1, x_2,\land,\neg,o_1, o_2, d_1$ and $d_2$ and $ID(x_1) = 1, ID(x_2) = 2, ID(\land) = 3, ID(\neg) = 4, ID(o_1) = 5, ID(o_2) = 6$.

The circuit
$\mathcal{C}_S$ has three input gates, three output gates and four valid inputs.
The construction of $K$ and $\varphi_{\textit{input}}$ enforces that $X$ contains the following traces:
\begin{itemize}[nosep]
\item $t_1, t_2, t_3$: Negative traces encoding input $000$ and representing edges $(x_1, \land)$, $(x_2, \land)$ and $(x_2,\neg)$.
\item $t_4$: A negative trace encoding input $010$ and representing edge $(x_1,\land)$.
\item $t_5, t_6$: Positive traces encoding $010$ and representing edges $(x_2,\land)$ and $(x_2, \neg)$.
\item $t_7$: A positive trace encoding input $100$ and representing edge $(x_1,\land)$.
\item $t_8, t_9$: Negative traces encoding $100$ and representing edges $(x_2,\land)$ and $(x_2, \neg)$.
\item $t_{10}, t_{11}, t_{12}$: Positive traces encoding $110$ and representing edges $(x_1, \land), (x_2, \land),(x_2,\neg)$.
\end{itemize}
Once traces $t_1-t_{12}$ are in $X$, the formula $\varphi_{\textit{and}}$ adds more traces to $X$:
\begin{itemize}[nosep]
\item $t_{13}$: A positive trace encoding input $110$ and edge $(\land, o_1)$ by instantiating $\pi_2$ with $t_{10}$ and $\pi_3$ with $t_{11}$.
\end{itemize}
Similarly, $\varphi_{\textit{or}}$ is only satisfied if the following traces are also in $X$:
\begin{itemize}[nosep]
\item $t_{14}$: A negative trace encoding input $000$ and edge $(\land, o_1)$ by instantiating $\pi_2$ with $t_1$ or $t_2$.
\item $t_{15}$: A negative trace encoding  $010$ and  $(\land, o_1)$ by instantiating $\pi_2$ with $t_4$.
\item $t_{16}$: A negative trace encoding  $100$ and  $(\land, o_1)$ by instantiating $\pi_2$ with $t_8$.
\end{itemize}
Finally, $\varphi_{\textit{not}}$ requires the following traces in $X$ to be satisfied:
\begin{itemize}[nosep]
\item $t_{17}$: A positive trace encoding  $000$ and  $(\neg, o_2)$ by instantiating $\pi_2$ with $t_3$.
\item $t_{18}$: A negative trace encoding  $010$ and  $(\neg, o_2)$ by instantiating $\pi_2$ with $t_6$.
\item $t_{19}$: A positive trace encoding  $100$ and  $(\neg, o_2)$ by instantiating $\pi_2$ with $t_9$.
\item $t_{20}$: A negative trace encoding  $110$ and  $(\neg, o_2)$ by instantiating $\pi_2$ with $t_{12}$.
\end{itemize}
By the construction of the substructures, traces $t_{13}, t_{14}, t_{15}$, $t_{16}$ are labeled with $\textit{bit}$ in the second step, and traces $t_{17}, t_{18}, t_{19}$, $t_{20}$ are labeled with $\textit{bit}$ in the third step. Traces
$t_{13}$ to $t_{16}$ represent the value of the first output bit for every input, and traces $t_{17}$ to $t_{20}$ represent the value of the second output bit for every input.
Since the input and output of~$\mathcal{C}_S$ is entirely represented by these traces and the traces labeled with $\textit{dummy}$, they are the only traces relevant for phase~B.
\end{example}

\begin{lemma}\label{proofs:fpexpcomplete:phasea:poly}
The construction in phase A takes polynomial time in the size of~$\mathcal{C}_S$.
\end{lemma}
\begin{proof}
Since $\mathcal{C}_S$ is explicitly given as input we pre-process $\mathcal{C}_S$ in polynomial time in the size of the input.
The IDs can be distributed by calculating a topological order on the gates of $\mathcal{C}_S$ and we add at most $l$ dummy gates.
The substructures can be built by iterating over every edge and building both relevant substructures, as the form of the substructure only depends on the two gates the edge connects.
Last, $l$ and the number of edges in $\mathcal{C}_S$ are polynomial in the size of the input, so the structure has only polynomially many states and transitions.
\end{proof}

\begin{lemma}\label{proofs:fpexpcomplete:phasea:correct}
For all gates $u, v \in V(\mathcal{C}_S)$, all inputs $e$ to $\mathcal{C}_S$, and all atomic propositions $d\in\{\textit{pos},\textit{neg}\}$, it holds that a trace $\pi$ encoding $ID(u)$ with \textit{fromid}, $ID(v)$ with \textit{toid}, and $e$ with \textit{inp}, and containing $d$ is an element of $X$ iff all these three conditions hold: (1) $(u, v)\in E(\mathcal{C}_S)$; (2) $e$ is a valid input; and (3) for input $e$ this edge carries a positive value if $d = \textit{pos}$ and a negative value if $d = \textit{neg}$.
\end{lemma}
\begin{proof}
All formulas adding traces to $X$ only add traces that are labeled with $a$ and therefore come from the substructures described in phase A.
We say that a trace $\pi$ \emph{represents edge} $(u, v)$ if $\pi$ has $ID(u)$ encoded with $\textit{fromid}$ and $ID(v)$ encoded with $\textit{toid}$.
Further, we say that $\pi$ \emph{represents a positive value} if it is labeled with \textit{pos} and it \emph{represents a negative value} if it is labeled with \textit{neg}.

Because these substructures only contain traces encoding edges of $\mathcal{C}_S$ there can not be a trace representing an edge that does not exist in $\mathcal{C}_S$.

Additionally, only the first $|e|$ layers are labeled with $\textit{inp}$ such that $e$ can not have positive inputs for dummy input gates.
Thus, every encoded input is valid.

Further, we observe that $\textit{sort}(u)$ can be found as follows:
Let $\pi$ represent edge $(u, v)\in E(\mathcal{C}_S)$ and $\pi_2$ represent edge $(w, u)\in E(\mathcal{C}_S)$.
Then $\pi_2$ is labeled with the atomic proposition representing $\textit{sort}(u)$.

We show, by induction over  $ID(u)$, that $\pi$ is in $X$ iff $\pi$ represents the correct value.
We distinguish between the different types of $u$.

\begin{description}[nosep]
\item[$u$ is an input gate:] By construction of the substructures only traces where $u$ is an input gate are labeled with $\textit{input}$.
We show that the process of adding all traces on which atomic proposition $\textit{input}$ appears, adds $\pi$ to $X$ iff $\pi$ represents the correct value for the edge $(u, v)$.
Let $m\in\mathbb N$ be such that $u$ takes the $m$-th input bit of $\mathcal{C}_S$.
If the represented value is positive, then $\pi$ should be added iff the $m$-th bit of the encoded input $e$ is 1.
The only substructure containing traces that represent $u, v$ and $\textit{pos}$ the state $a_{u, v, m, 1, \textit{pos}}$ is missing and $\pi$ has to visit $a_{u, v, m, 2, \textit{pos}}$.
Because $a_{u, v, m, 2, \textit{pos}}$ is labeled with \textit{inp} and it is the $m$-th state encoding $e$, the $m$-th bit of $e$ has to be 1.
The case where the represented value is negative is analogous.
\item[$u$ is an AND gate:] If $u$ is an AND gate, then there are gates $s, s'\in V(\mathcal{C}_S)$ with $s\neq s'$ such that $(s, u), (s', u)\in E(\mathcal{C}_S)$.
By the distribution of the IDs, we know $ID(s), ID(s')~<~ID(u)$.
Thus by induction hypothesis, we know that for every valid input $e$ there are traces $\pi_2, \pi_3$ representing $(s, u)$ and $(s', u)$ in $X$, that represent the correct value of the edges for input~$e$.
If the represented value on $\pi$ positive then it should be included in $X$ iff $\pi_2$ and $\pi_3$ also represent a positive value.
The only formula that adds traces starting in an AND gate ($\LTLfinally \textit{and}_{\pi_2}$) and representing a positive value ($\LTLfinally\textit{pos}_\pi$) is $\varphi_{\textit{and}}$, which adds $\pi$ iff $\pi, \pi_2$ and $\pi_3$ represent positive value.
If the represented value is negative then $\pi$ should be included in $X$ iff $\pi_2$ or $\pi_3$ is represent negative value. The formula
$\varphi_{\textit{or}}$ is the only formula adding traces that represent negative value and $u$ is an AND gate.
$\varphi_{\textit{or}}$ adds $\pi$ if it can find a negative trace for the same input and entering $u$.
This is the case iff $\pi_2$ or $\pi_3$ is~negative.
\item[$u$ is an OR gate:] It holds that $p\lor q\equiv \neg(\neg p\land \neg q)$.
Every formula adding traces leaving AND gates, namely $\varphi_\textit{and}$ and $\varphi_\textit{or}$, also add traces leaving OR gates but with swapped $\textit{pos}$ and $\textit{neg}$ on $\pi,\pi_2,\pi_3$.
Therefore, this case is analogous to the case where $u$ is an AND gate.
\item[$u$ is a NOT gate:] If $u$ is a NOT gate, then there exists exactly one gate $s \in V(\mathcal{C}_S)$ such that $(s, u)\in E(\mathcal{C}_S)$.
We know that $ID(s) < ID(u)$ by the distribution of the IDs.
We know by induction hypothesis that for every valid input $e$ there is a trace $\pi_2\in X$ that represents $(s, u)$.
If the represented value is positive, then $\pi$ should be added iff $\pi_2$ represents a negative value and if the represented value is negative, then $\pi$ should be added iff $\pi_2$ is negative.
The formula $\varphi_{\textit{not}}$ does exactly this.
\end{description}
This concludes the construction and correctness of phase A. 
\end{proof}

\paragraph*{Phase B}
In phase B we collect one trace per input-output pair of $\mathcal{C}_S$ in a set $Y$.
In the set $X$ of phase A, the output of $\mathcal{C}_S$ for every input is distributed over several traces (one trace per output gate). Here, we would like to collect them on one trace using a binary representation.

\begin{figure}[t]
\begin{subfigure}{0.47\textwidth}
\resizebox{.9\linewidth}{!}{
\begin{tikzpicture}[shorten >=1pt,node distance=2cm,on grid,auto]
  \tikzstyle{every state}=[fill={rgb:black,1;white,10},minimum size=30pt,inner sep=0pt]

  \node[state,initial, ] (s_0)  {};
  \node[state, align=center, node distance=2.5cm] (a11) [below of=s_0] {$\{b, \textit{inp},$\\$\textit{outp}\}$};
  \node[state] (a12) [right of=a11] {$\{b, \textit{inp}\}$};
  \node[state] (a13) [right of=a12] {$\{b, \textit{outp}\}$};
  \node[state] (a14) [right of=a13] {$\{b\}$};
  \node[state, node distance=2.5cm, fill={rgb:white,10}] (b12) [left of=a11] {$\{b\}$};
  \node[state, fill={rgb:white,10}] (b11) [left of=b12] {$\{b, \textit{inp}\}$};

  \node[state, align=center, node distance=4cm] (a21) [below of=a11] {$\{\textit{inp},$\\$\textit{outp}\}$};
  \node[state] (a22) [right of=a21] {$\{\textit{inp}\}$};
  \node[state] (a23) [right of=a22] {$\{\textit{outp}\}$};
  \node[state] (a24) [right of=a23] {$\{\}$};
  \node[state, node distance=2.5cm, fill={rgb:white,10}] (b22) [left of=a21] {$\{\}$};
  \node[state, fill={rgb:white,10}] (b21) [left of=b22] {$\{\textit{inp}\}$};

  \node[state, align=center, node distance=4cm] (a31) [below of=a21] {$\{\textit{outp},$\\$\textit{inp, end}\}$};
  \node[state, align=center] (a32) [right of=a31] {$\{\textit{inp}$\\$\textit{end}\}$};
  \node[state, align=center] (a33) [right of=a32] {$\{\textit{outp}$\\$\textit{end}\}$};
  \node[state] (a34) [right of=a33] {$\{\textit{end}\}$};
  \node[state, node distance=2.5cm, fill={rgb:white,10}] (b32) [left of=a31] {$\{\}$};
  \node[state, fill={rgb:white,10}] (b31) [left of=b32] {$\{\textit{inp}\}$};

  \path[->]
  (s_0) edge  node {} (a11)
  (s_0) edge  node {} (a12)
  (s_0) edge  node {} (a13)
  (s_0) edge  node {} (a14)
  (s_0) edge  node {} (b12)
  (s_0) edge  node {} (b11);

  \path[->]
  (a11) edge  node {} (a21)
  (a11) edge  node {} (a22)
  (a11) edge  node {} (a23)
  (a11) edge  node {} (a24)
  (a11) edge  node {} (b22)
  (a11) edge  node {} (b21)
  (a12) edge  node {} (a21)
  (a12) edge  node {} (a22)
  (a12) edge  node {} (a23)
  (a12) edge  node {} (a24)
  (a12) edge  node {} (b22)
  (a12) edge  node {} (b21)
  (a13) edge  node {} (a21)
  (a13) edge  node {} (a22)
  (a13) edge  node {} (a23)
  (a13) edge  node {} (a24)
  (a13) edge  node {} (b22)
  (a13) edge  node {} (b21)
  (a14) edge  node {} (a21)
  (a14) edge  node {} (a22)
  (a14) edge  node {} (a23)
  (a14) edge  node {} (a24)
  (a14) edge  node {} (b22)
  (a14) edge  node {} (b21)
  (b11) edge  node {} (b22)
  (b11) edge  node {} (b21)
  (b12) edge  node {} (b22)
  (b12) edge  node {} (b21);

  \path[->]
  (a21) edge  node {} (a31)
  (a21) edge  node {} (a32)
  (a21) edge  node {} (a33)
  (a21) edge  node {} (a34)
  (a21) edge  node {} (b32)
  (a21) edge  node {} (b31)
  (a22) edge  node {} (a31)
  (a22) edge  node {} (a32)
  (a22) edge  node {} (a33)
  (a22) edge  node {} (a34)
  (a22) edge  node {} (b32)
  (a22) edge  node {} (b31)
  (a23) edge  node {} (a31)
  (a23) edge  node {} (a32)
  (a23) edge  node {} (a33)
  (a23) edge  node {} (a34)
  (a23) edge  node {} (b32)
  (a23) edge  node {} (b31)
  (a24) edge  node {} (a31)
  (a24) edge  node {} (a32)
  (a24) edge  node {} (a33)
  (a24) edge  node {} (a34)
  (a24) edge  node {} (b32)
  (a24) edge  node {} (b31)
  (b21) edge  node {} (b32)
  (b21) edge  node {} (b31)
  (b22) edge  node {} (b32)
  (b22) edge  node {} (b31);

  \path[->]
  (a31) edge  [loop below] node {} ( )
  (a32) edge  [loop below] node {} ( )
  (a33) edge  [loop below] node {} ( )
  (a34) edge  [loop below] node {} ( )
  (b31) edge  [loop below] node {} ( )
  (b32) edge  [loop below] node {} ( );
\end{tikzpicture}
}
\caption{Substructures for phase B with $l=|e| = 3$. White states are additionally labeled with $\textit{incomplete}$.}
\label{proofs:fpexpcomplete:phaseb:example}
\end{subfigure}
\hfill
\begin{subfigure}{0.47\textwidth}
\resizebox{.9\linewidth}{!}{
\begin{tikzpicture}[shorten >=1pt,node distance=2.5cm,on grid,auto]
  \tikzstyle{every state}=[fill={rgb:white,10},minimum size=17pt,inner sep=0pt]
  \node[state,initial, ] (s_0)  {};
  \node[state,] (c14) [below left=3cm and 1cm of s_0] {$\{c\}$};
  \node[state, align=center, node distance=4cm] (c15) [below right=3cm and 1cm of s_0] {$\{\textit{inp, c,}$\\$\textit{outp}\}$};
  \node[state, node distance=1.5cm] (c13) [left of=c14] {$\{\textit{outp, c}\}$};
  \node[state, node distance=1.5cm] (c12) [left of=c13] {$\{\textit{inp, c}\}$};
  \node[state, align=center, node distance=1.5cm] (c11) [left of=c12] {$\{\textit{inp},c,$\\$\textit{outp}\}$};
  \node[state, node distance=1.5cm] (c16) [right of=c15] {$\{\textit{inp, c}\}$};
  \node[state, node distance=1.5cm] (c17) [right of=c16] {$\{\textit{outp, c}\}$};
  \node[state, node distance=1.5cm] (c18) [right of=c17] {$\{c\}$};

  \node[state, align=center] (c21) [below of=c11] {$\{\textit{inp},$\\$\textit{outp}\}$};
  \node[state,] (c22) [below of=c12] {$\{\textit{inp}\}$};
  \node[state,] (c23) [below of=c13] {$\{\textit{outp}\}$};
  \node[state,] (c24) [below of=c14] {$\{\}$};
  \node[state, align=center] (c25) [below of=c15] {$\{\textit{inp},$\\$\textit{outp}\}$};
  \node[state,] (c26) [below of=c16] {$\{\textit{inp}\}$};
  \node[state,] (c27) [below of=c17] {$\{\textit{outp}\}$};
  \node[state,] (c28) [below of=c18] {$\{\}$};

  \node[state, fill={rgb:black,1;white,10}, align=center] (c31) [below of=c21] {$\{\textit{inp},$\\$\textit{outp}\}$};
  \node[state, fill={rgb:black,1;white,10},] (c32) [below of=c22] {$\{\textit{inp}\}$};
  \node[state, fill={rgb:black,1;white,10},] (c33) [below of=c23] {$\{\textit{outp}\}$};
  \node[state, fill={rgb:black,1;white,10},] (c34) [below of=c24] {$\{\}$};
  \node[state, fill={rgb:black,1;white,10}, align=center] (c35) [below of=c25] {$\{\textit{inp},$\\$\textit{outp}\}$};
  \node[state, fill={rgb:black,1;white,10},] (c36) [below of=c26] {$\{\textit{inp}\}$};
  \node[state, fill={rgb:black,1;white,10},] (c37) [below of=c27] {$\{\textit{outp}\}$};
  \node[state, fill={rgb:black,1;white,10},] (c38) [below of=c28] {$\{\}$};

  \node[state, align=center] (c41) [below of=c31] {$\{\textit{inp},$\\$\textit{outp}\}$};
  \node[state,] (c42) [below of=c32] {$\{\textit{inp}\}$};
  \node[state,] (c43) [below of=c33] {$\{\textit{outp}\}$};
  \node[state,] (c44) [below of=c34] {$\{\}$};
  \node[state, align=center] (c45) [below of=c35] {$\{\textit{inp},$\\$\textit{outp}\}$};
  \node[state,] (c46) [below of=c36] {$\{\textit{inp}\}$};
  \node[state,] (c47) [below of=c37] {$\{\textit{outp}\}$};
  \node[state,] (c48) [below of=c38] {$\{\}$};

  \node[state, align=center] (c51) [below of=c41] {$\{\textit{inp},$\\$\textit{outp}\}$};
  \node[state,] (c52) [below of=c42] {$\{\textit{inp}\}$};
  \node[state,] (c53) [below of=c43] {$\{\textit{outp}\}$};
  \node[state,] (c54) [below of=c44] {$\{\}$};
  \node[state, align=center] (c55) [below of=c45] {$\{\textit{inp},$\\$\textit{outp}\}$};
  \node[state,] (c56) [below of=c46] {$\{\textit{inp}\}$};
  \node[state,] (c57) [below of=c47] {$\{\textit{outp}\}$};
  \node[state,] (c58) [below of=c48] {$\{\}$};

  \path[->]
  (s_0) edge [bend right=10] node {} (c11)
  (s_0) edge  node {} (c12)
  (s_0) edge  node {} (c13)
  (s_0) edge  node {} (c14)
  (s_0) edge  node {} (c15)
  (s_0) edge  node {} (c16)
  (s_0) edge  node {} (c17)
  (s_0) edge [bend left=10] node {} (c18);

  \path[->]
  (c11) edge  node {} (c21)
  (c11) edge  node {} (c22)
  (c11) edge  node {} (c23)
  (c11) edge  node {} (c24)
  (c12) edge  node {} (c21)
  (c12) edge  node {} (c22)
  (c12) edge  node {} (c23)
  (c12) edge  node {} (c24)
  (c13) edge  node {} (c21)
  (c13) edge  node {} (c22)
  (c13) edge  node {} (c23)
  (c13) edge  node {} (c24)
  (c14) edge  node {} (c21)
  (c14) edge  node {} (c22)
  (c14) edge  node {} (c23)
  (c14) edge  node {} (c24);

  \path[->]
  (c21) edge  node {} (c31)
  (c21) edge  node {} (c32)
  (c21) edge  node {} (c33)
  (c21) edge  node {} (c34)
  (c22) edge  node {} (c31)
  (c22) edge  node {} (c32)
  (c22) edge  node {} (c33)
  (c22) edge  node {} (c34)
  (c23) edge  node {} (c31)
  (c23) edge  node {} (c32)
  (c23) edge  node {} (c33)
  (c23) edge  node {} (c34)
  (c24) edge  node {} (c31)
  (c24) edge  node {} (c32)
  (c24) edge  node {} (c33)
  (c24) edge  node {} (c34);

  \path[->]
  (c31) edge  node {} (c41)
  (c31) edge  node {} (c42)
  (c31) edge  node {} (c43)
  (c31) edge  node {} (c44)
  (c32) edge  node {} (c41)
  (c32) edge  node {} (c42)
  (c32) edge  node {} (c43)
  (c32) edge  node {} (c44)
  (c33) edge  node {} (c41)
  (c33) edge  node {} (c42)
  (c33) edge  node {} (c43)
  (c33) edge  node {} (c44)
  (c34) edge  node {} (c41)
  (c34) edge  node {} (c42)
  (c34) edge  node {} (c43)
  (c34) edge  node {} (c44);

  \path[->]
  (c41) edge  node {} (c51)
  (c41) edge  node {} (c52)
  (c41) edge  node {} (c53)
  (c41) edge  node {} (c54)
  (c42) edge  node {} (c51)
  (c42) edge  node {} (c52)
  (c42) edge  node {} (c53)
  (c42) edge  node {} (c54)
  (c43) edge  node {} (c51)
  (c43) edge  node {} (c52)
  (c43) edge  node {} (c53)
  (c43) edge  node {} (c54)
  (c44) edge  node {} (c51)
  (c44) edge  node {} (c52)
  (c44) edge  node {} (c53)
  (c44) edge  node {} (c54);

  \path[->]
  (c51) edge [loop below] node {} ( )
  (c52) edge [loop below] node {} ( )
  (c53) edge [loop below] node {} ( )
  (c54) edge [loop below] node {} ( );

  \path[->]
  (c15) edge  node {} (c25)
  (c15) edge  node {} (c26)
  (c15) edge  node {} (c27)
  (c15) edge  node {} (c28)
  (c16) edge  node {} (c25)
  (c16) edge  node {} (c26)
  (c16) edge  node {} (c27)
  (c16) edge  node {} (c28)
  (c17) edge  node {} (c25)
  (c17) edge  node {} (c26)
  (c17) edge  node {} (c27)
  (c17) edge  node {} (c28)
  (c18) edge  node {} (c25)
  (c18) edge  node {} (c26)
  (c18) edge  node {} (c27)
  (c18) edge  node {} (c28);

  \path[->]
  (c25) edge  node {} (c35)
  (c25) edge  node {} (c36)
  (c25) edge  node {} (c37)
  (c25) edge  node {} (c38)
  (c26) edge  node {} (c35)
  (c26) edge  node {} (c36)
  (c26) edge  node {} (c37)
  (c26) edge  node {} (c38)
  (c27) edge  node {} (c35)
  (c27) edge  node {} (c36)
  (c27) edge  node {} (c37)
  (c27) edge  node {} (c38)
  (c28) edge  node {} (c35)
  (c28) edge  node {} (c36)
  (c28) edge  node {} (c37)
  (c28) edge  node {} (c38);

  \path[->]
  (c35) edge  node {} (c45)
  (c35) edge  node {} (c46)
  (c35) edge  node {} (c47)
  (c35) edge  node {} (c48)
  (c36) edge  node {} (c45)
  (c36) edge  node {} (c46)
  (c36) edge  node {} (c47)
  (c36) edge  node {} (c48)
  (c37) edge  node {} (c45)
  (c37) edge  node {} (c46)
  (c37) edge  node {} (c47)
  (c37) edge  node {} (c48)
  (c38) edge  node {} (c45)
  (c38) edge  node {} (c46)
  (c38) edge  node {} (c47)
  (c38) edge  node {} (c48);

  \path[->]
  (c45) edge  node {} (c55)
  (c45) edge  node {} (c56)
  (c45) edge  node {} (c57)
  (c45) edge  node {} (c58)
  (c46) edge  node {} (c55)
  (c46) edge  node {} (c56)
  (c46) edge  node {} (c57)
  (c46) edge  node {} (c58)
  (c47) edge  node {} (c55)
  (c47) edge  node {} (c56)
  (c47) edge  node {} (c57)
  (c47) edge  node {} (c58)
  (c48) edge  node {} (c55)
  (c48) edge  node {} (c56)
  (c48) edge  node {} (c57)
  (c48) edge  node {} (c58);

  \path[->]
  (c55) edge [loop below] node {} ( )
  (c56) edge [loop below] node {} ( )
  (c57) edge [loop below] node {} ( )
  (c58) edge [loop below] node {} ( );

\end{tikzpicture}
}
\caption{The substructure built for phase C with $l=5$ and two bits to encode a gate in $C$. The states on the left are additionally labeled with $\textit{pos}$ and the states on the right are labeled with $\textit{neg}$. Gray states are labeled with \textit{kstart} and \textit{nstart}.}
\label{proofs:fpexpcomplete:phasec:example}
\end{subfigure}

\end{figure}

We build two substructures for this process.
Similar to phase A, both structures consist of $l$ layers.
One structure has four states per layer, named $b_{i, j}$ with $i\in\{1,...,l\}, j\in\{1, 2, 3, 4\}$.
The other structure has two states per layer, which are named $b'_{i, j}$ with $i\in\{1,...,l\}, j\in\{1, 2\}$.

States $b_{1, j}$ and $b'_{1, j}$ are labeled with $b$ and there are transitions from the initial state of $K$ to $b_{1, j}, b'_{1, j}$. %the initial state of $K$ is connected to them.
States $b_{i, 1}, b_{i, 2}$ and $b'_{i, 1}$ with $i\le |e|$ are labeled with atomic proposition $\textit{inp}$ and states $b_{i, 1}$ and $b_{i, 3}$ are labeled with $\textit{outp}$.
Additionally, all states $b'_{i, j}$ are labeled with atomic proposition $\textit{incomplete}$ and states $b_{l, j}$ are labeled with $\textit{end}$.
An example of a substructure for phase B is given in \Cref{proofs:fpexpcomplete:phaseb:example}.

As in phase A, every state is connected to all states in the next layer of their substructure and the states of the last layer of both structures have self-loops.
In contrast to phase A, all states $b_{i, j}$ with $i \in\{1, ..., l-1\}$ have a transition to states $b'_{i+1, 1}$ and $b'_{i+1, 2}$.

To summarize all traces in the substructure we build in phase B:
All traces have a number of length $l$ with proposition $\textit{inp}$ encoded on them.
Additionally, they have numbers of different lengths encoded with $\textit{outp}$ until they reach a state labeled with $\textit{incomplete}$.
We call these traces \emph{incomplete}.
Traces whose $\textit{outp}$ number consisting of $l$ bits never visit a state labeled with $\textit{incomplete}$ but visit a state labeled with $\textit{end}$.
We iteratively compute all traces encoding complete input-output pairs of~$\mathcal C_S$.
Each trace encodes an input with \textit{inp} from the beginning but the number of encoded output bits with \textit{outp} will increase.
To denote that the encoded output ends, the trace changes into the substructure that has its states labeled with \textit{incomplete}.

The formula $\varphi_B$ first computes the intermediate set $Y'$ which contains auxiliary traces labeled with \textit{incomplete} as well as the traces we are interested in.
The set $Y$ then consists only of traces from $Y'$ which do not visit a state labeled with \textit{incomplete}.
Therefore, $\varphi_B$ is defined as follows. 
\begin{align*}
\varphi_B = (Y', \curlyvee, \varphi_B'\land\varphi_B'').(Y,\curlyvee, \forall\pi\in Y'. \LTLglobally\neg\textit{incomplete}\to\pi\triangleright Y).\varphi_C
\end{align*}
The set $Y'$ should contain for every input $e$ and every prefix $e'$ of the corresponding output, one trace $\pi$, which has $e$ encoded with $\textit{inp}$ and $e'$ encoded with $\textit{outp}$.
For this, we first add all traces where $e'$ has length 0: 
\begin{align*}
\varphi_B' = \forall \pi\in\mathfrak G. \LTLfinally b\land\LTLnext\textit{incomplete}\to\pi\triangleright Y'
\end{align*}
$\varphi_B'$ adds to $Y'$ all traces which immediately visit a state labeled with $\textit{incomplete}$.
\begin{align*}
&\varphi_B'' = \forall\pi_1\in\mathfrak G.\forall\pi_2\in Y'.\forall\pi_3\in \mathfrak G.\forall\pi_4\in X.\LTLfinally b_{\pi_1}\land(\pi_3 =_{\textit{AP}}\pi_4\lor\LTLfinally \textit{dummy}_{\pi_3})\\
&\land \LTLglobally((\textit{inp}_{\pi_1}\leftrightarrow\textit{inp}_{\pi_2})\land(\textit{inp}_{\pi_2}\leftrightarrow\textit{inp}_{\pi_3}))\land (\textit{outp}_{\pi_1}\leftrightarrow\textit{outp}_{\pi_2})~\LTLuntil~\textit{incomplete}_{\pi_2}\\
&\land \LTLfinally(\neg\textit{incomplete}_{\pi_2}\land\LTLnext\textit{incomplete}_{\pi_2}\land\LTLnext(\neg\textit{incomplete}_{\pi_1}\land
(\textit{end}_{\pi_1}\lor\LTLnext\textit{incomplete}_{\pi_1})\\
&\quad\quad\quad\land\textit{bit}_{\pi_3}\land(\textit{outp}_{\pi_1}\leftrightarrow\LTLfinally\textit{pos}_{\pi_3})))
\to\pi_1\triangleright Y'
\end{align*}
The formula $\varphi''_B$ adds to $Y'$ traces which have one more output bit encoded on them than the traces added in the previous iteration.
This means they visit a state labeled with \textit{incomplete} one step later.
We quantify over traces $\pi_1, \pi_2, \pi_3, \pi_4$:
trace $\pi_2$ is the trace that we want to continue by one output bit, trace $\pi_1$ is the trace that we might want to add and trace $\pi_3$ is the trace that indicates the output of $\mathcal C_S$ at the bit by which $\pi_2$ should be continued.
Trace $\pi_3$ is either a trace representing a dummy output or $\pi_3$ is equal to $\pi_4$ and therefore $\pi_3\in X$.
This is the only purpose of $\pi_4$.

We check whether all three traces have the same input encoded on them and whether $\pi_1$ and $\pi_2$ have the same output encoded until $\pi_2$ enters the $\textit{incomplete}$ states.
Then, we identify the position $m$ one before $\pi_2$ enters a state labeled with \textit{incomplete}.
As $\pi_1$ should encode one bit of output more than $\pi_2$ it should visit a state labeled \textit{incomplete} for the first time one after $\pi_1$ does at position $m+1$.
Alternatively, the encoded output has reached its final length $l$ in which case $\pi_1$ can also visit a state labeled with \textit{end}.
To ensure that $\pi_3$ is a trace responsible for output bit $m$, we also require that it is marked with \textit{bit} at $m+1$.
Finally, $\pi_1$ is labeled with \textit{outp} at position $m+1$ iff $\pi_3$ represents positive value.

\begin{example}
We follow \Cref{proofs:fpexpcomplete:phasea:example:circuittext}.
To demonstrate how phase B works we focus on traces that specify the behavior of $\mathcal{C}_S$ for the input $100$.
The output of $\mathcal{C}_S$ for input $100$ is $010$.
Therefore $Y'$ should contain the following trace:
$
\{\} \{\textit{inp}\}\{\textit{outp}\}\{\textit{end}\}\cdots
$. The formula~$\varphi'_B$ is only satisfied if every trace that immediately visits an incomplete state is contained in~$Y'$.
This also leads to the fact that $Y'$ contains a trace $t'_1 = \{\}\{\textit{inp, incomplete}\}\{\textit{incomplete}\}\{\textit{incomplete}\}\cdots$.
Once $t'_1$ is an element of $Y'$ we can instantiate $\varphi''_B$ as follows:
\begin{align*}
\pi_1, t'_2 &=~ \{\}\{\textit{inp}\}\{\textit{incomplete}\}\{\textit{incomplete}\}\cdots\\
\pi_2, t'_1 &=~ \{\}\{\textit{inp, incomplete}\}\{\textit{incomplete}\}\{\textit{incomplete}\}\cdots\\
\pi_3,\pi_4, t_{16} &=~ \{\}\{\textit{neg, bit, toid, inp, output}\}\{\textit{neg, fromid}\}\{\textit{neg, fromid, toid}\}\cdots
\end{align*}
This is the only instantiation, where $\pi_3 =_{\textit{AP}}\pi_4$, all four traces encode input $100$, $\pi_1$ and $\pi_2$ agree on the output until $\pi_2$ visits an incomplete state (which it does immediately), and $\pi_1$ continues $\pi_2$ according to $\pi_3$.
That means that at the point where $\pi_2$ visits a state not labeled with $\textit{incomplete}$ and visits states labeled with $\textit{incomplete}$ in the next steps (which is for this instantiation of $\pi_2$ directly at the beginning of the trace) hold the following conditions:
\begin{itemize}[nosep]
\item $\pi_1$ does not visit an incomplete state in the next step. This holds since the second state of $t'_2$ is only labeled with $\textit{inp}$.
    \item $\pi_1$ visits an incomplete state in two steps, or a state labeled with $\textit{end}$ in the next step. This holds since the third state of $t'_2$ is labeled with $\textit{incomplete}$.
\item $\pi_3$ is labeled with $\textit{bit}$ in the next step. This holds since $t_{16}$ represents an edge to the first output gate and is therefore labeled with $\textit{bit}$ in its second state.
\item $\pi_1$ is labeled with $\textit{outp}$ in the next state iff $\pi_3$ visits a state labeled with $\textit{pos}$. This holds since $t'_2$ is not labeled with $\textit{outp}$, and $t_{16}$ is a negative trace.
\end{itemize}
The next fixpoint iteration works very similarly.
\begin{align*}
\pi_1, t'_3 &= \{\}\{\textit{inp}\}&&\{\textit{outp}\}\{\textit{incomplete}\}\cdots\\
\pi_2, t'_2 &= \{\}\{\textit{inp}\}&&\{\textit{incomplete}\}\{\textit{incomplete}\}\cdots\\
\pi_3, \pi_4, t_{19} &= \{\}\{\textit{pos, inp, fromid, output}\}&&\{\textit{pos, toid, fromid, bit}\}\{\textit{pos}\}\cdots
\end{align*}
For the last iteration, two things change.
First, there is no trace in $X$ for which $\textit{bit}$ holds in the third state after the initial state.
Therefore $\pi_3$ is instantiated with a trace for a dummy output gate.
Second, since $l=3$ the trace $t'_4$ which is added now to $Y'$ is not incomplete but ends in a state labeled with $\textit{end}$.
The instantiation now looks as follows:
\begin{align*}
\pi_1 &= t'_4 &=~& \{\}&&\{\textit{inp}\}&&\{\textit{outp}\}&&\{\textit{end}\}\cdots\\
\pi_2 &= t'_3 &=~& \{\}&&\{\textit{inp}\}&&\{\textit{outp}\}&&\{\textit{incomplete}\}\cdots\\
\pi_3 &=  &~& \{\}&&\{\textit{neg, inp, output, dummy}\}&&\{\textit{neg}\}&&\{\textit{neg, bit}\}\cdots
\end{align*}
We can pick an arbitrary trace for $\pi_4$.

Trace $t'_4$ is not labeled with $\textit{incomplete}$ anywhere, and therefore it is included in $Y$.
With~$t'_4$ there is a trace in $Y$ that for input $100$ encodes the correct corresponding output $010$.
\end{example}

\begin{lemma}
The construction of phase B takes polynomial time in $|\mathcal{C}_S|$.
\end{lemma}
\begin{proof}
The structure built in phase B only depends on $l$.
Because $l$ is linear in the size of~$\mathcal{C}_S$ and the substructure consists of $(2 + 4)\cdot l$ states it can be built in polynomial time in the size of $\mathcal{C}_S$.
\end{proof}

\begin{lemma}
Let $\pi$ be a trace labeled with $b$ and encoding number $e$ with $\textit{inp}$.
Let $o$ be the number of length $m$ encoded on $\pi$ with $\textit{outp}$ before $\pi$ visits a state labeled with $\textit{incomplete}$.
Then it holds that $\pi\in Y'$ iff $o$ is a prefix of the output of $\mathcal{C}_S$ if $e$ was given as input.
\end{lemma}
\begin{proof}
First, we observe that every trace that is in $Y'$ has to be from the substructures described in phase B because $\varphi_B'$ as well as $\varphi_B''$ are only satisfied if a trace is marked with $b$ before adding it to $Y'$.
We prove the remaining claim by induction over $m$.

\emph{Case $m = 0$:} If $o$ consists of zero bits, it is a prefix for every bit sequence.
Then, every trace encoding an input and zero bits of an output should be in $Y'$.
The set $Y'$ contains all these traces because otherwise, it would not satisfy~$\varphi'_B$.

\emph{Case $0 < m \le l$:} We know by induction hypothesis that $Y'$ contains a trace $\pi_2$ encoding $m-1$ output bits of $e$ with \textit{inp}.
Additionally, we know from \Cref{proofs:fpexpcomplete:phasea:correct} and the construction of the substructures for phase A that for input $e$ there exists a trace $\pi_3$ where the $m+1$-th step is labeled with $\textit{bit}$ and $e$ is encoded with \textit{inp}.
The trace $\pi_3$ is either in $X$ or labeled with $\textit{dummy}$.
It is labeled with $\textit{dummy}$ iff the $m$-th output gate is a dummy gate.
By \Cref{proofs:fpexpcomplete:phasea:correct}, trace $\pi_3$ correctly represents the value of the $m$-th output bit for input $e$ by visiting states labeled with \textit{pos} (or \textit{neg}, respectively).

If $\pi_3$ represents a positive value, then the $m$-th bit of the output is a 1.
In this case, $\pi$ should be in $Y'$ iff it continues the input and output encoded on $\pi_2$ and visits a state labeled with $\textit{outp}$ where $\pi_2$ first visits a state labeled with $\textit{incomplete}$.
Similarly, if $\pi_3$ is negative, then $\pi$ should be added if it continues $\pi_2$ by a state that is not labeled with $\textit{outp}$.

If the $m$-th output is a dummy output, then $K$ does not contain any traces encoding positive value for the $m$-th output.
The formula $\varphi_B''$ specifies exactly this behavior.
Trace $\pi$ is added to $Y'$ iff we find $\pi_2\in Y'$ and $\pi_3\in \mathfrak G$ such that $\pi$ is labeled with~$b$ somewhere, all three traces encode the same input, and $\pi$ and $\pi_2$ agree on the output until $\pi_2$ visits an $\textit{incomplete}$ state.

The formula $\neg\textit{incomplete}_{\pi_2}\land\LTLnext\textit{incomplete}_{\pi_2}$ is only satisfied at the $m$-th position because $\pi_2$ encodes $m-1$ bits of output.
In the $m+1$-th step several conditions must be true.

First, $\pi$ must not be labeled with $\textit{incomplete}$, and has to be labeled $\textit{incomplete}$ in the next step, or $\pi$ ends immediately.
This is true iff $\pi$ encodes exactly one output bit more than $\pi_2$, which means it encodes $m-1+1 = m$ bits of output.

Second, $\pi_3$ has to be labeled with $\textit{bit}$ at the $m+1$-th position which is per construction of the substructures in phase A only the case if it represents the output gate responsible for the $m$-th output bit.
Finally, $\pi$ has to be labeled with $\textit{outp}$ at position $m+1$ iff $\pi_3$ is labeled with $\textit{pos}$ somewhere (after the $m+1$-th position).
If $\pi$ visits $\textit{outp}$ in the $m+1$-th position, then the $m$-th bit of the output encoded with \textit{outp} is true, otherwise it is false.
Since all traces in $X$ which represent a positive value only visit states labeled with \textit{pos}, we have that $\LTLfinally\textit{pos}_{\pi_3}$ is satisfied at this position iff $\pi_3$ represents positive value.
Therefore, $\pi$ is part of $Y'$ iff the $m$-th encoded output bit on $\pi$ has the correct value.

This concludes phase B. 
\end{proof}

\begin{corollary}
The set $Y$ contains, for every input $e$ of $\mathcal{C}_S$ and its corresponding output $o$, exactly one trace encoding $e$ with $\textit{inp}$ and $o$ with $\textit{outp}$.
\end{corollary}

\paragraph*{Phase C}

This phase is similar to phase A,
with the exception that $\mathcal{C}$ is encoded in $\mathcal{C}_S$ such that every edge $(u, v) \in E(\mathcal{C})$ is represented by two input-output pairs --
one pair that identifies $u$ as a predecessor of $v$, and one pair that identifies $v$ as a successor of $u$.

To explicitly annotate the traces from $Y$ with some more atomic propositions we build two substructures in the resulting structure $K$.
Both structures consist of $l$ layers with four states each.
We call the $j$-th state in the $i$-th layer of one structure $c_{i, j}$, and the $j$-th state of the $i$-th layer of the other structure $c'_{i, j}$.
There are transitions from the initial state of $K$ to states $c_{1, j}$ and $c'_{1, j}$.
Additionally, every layer is fully connected to the states of the next layer in their structure and the states in the last layer have self-loops.

States $c_{i, 1}, c_{i, 2}, c'_{i, 1}, c'_{i, 2}$ are labeled with $\textit{inp}$, and states $c_{i, 1}, c_{i, 3}, c'_{i, 1}, c'_{i, 3}$ are labeled with $\textit{outp}$.
Additionally, every state $c_{i, j}$ is labeled with $\textit{pos}$ and every state $c'_{i, j}$ is labeled with $\textit{neg}$.
As in the other phases, we mark $c_{1, j}$ and $c'_{1, j}$ with atomic proposition $c$.
Since we are interested in the semantics of the input-output pairs, we mark the bit where the encoding of the neighbor starts with atomic proposition $\textit{nstart}$ and the start of the encoded sort with $\textit{kstart}$.
See example in \Cref{proofs:fpexpcomplete:phasec:example}.
Because the encoding of the kind and the encoding of neighbor both start after a gate ID was encoded, $\textit{nstart}$ and $\textit{kstart}$ are always at the same position.
We use both atomic propositions to emphasize whether we are looking for the encoded kind or the encoded neighbor on a trace.

For each pair of numbers of length $l$ we have two traces.
One trace labeled with \textit{pos} and one trace labeled with \textit{neg}.
Both encode one number with \textit{inp} and the other number with \textit{outp}.
The placement of \textit{inp} and \textit{outp} allows us to encode all input-output pairs of $\mathcal C_S$ resp. edges of $\mathcal C$.
The value carried by the represented edge of $\mathcal C$ can then be encoded by choosing the trace labeled with \textit{pos} or the trace labeled with \textit{neg}.
The traces are marked with atomic propositions $\textit{nstart}$ and $\textit{kstart}$ where the encoding of the neighbor and the sort of gate starts.

Similar to phase A, we maintain a set $Z$ that collects the positive trace for an edge in $\mathcal{C}$ if the value on this edge is positive and the negative trace otherwise.
Traces representing input-output pairs of $\mathcal{C}_S$ which identify some gate $v$ as a predecessor of some gate $u$, are not collected in set $Z$ and all formulas are constructed such that they ignore these traces.

The formula $\varphi_C$ is then similar to the one from phase A, but for every gate, its sort is binary encoded after $\textit{kstart}$ and both directions of an edge are represented as a trace.
For a given trace we can find out which direction of an edge it represents by looking at the sort of the gate and the encoded neighbor number.
As in phase A, the formula $\varphi_C$ distincts over the sort of the gates:
\begin{align*}
\varphi_C = (Z, \curlyvee, \varphi_{\textit{cconst}}\land\varphi_{\textit{cand}}\land \varphi_{\textit{cor}}\land \varphi_{\textit{cnot}}).\varphi_D
\end{align*}
Adding the correct traces for TRUE and FALSE gates in $\mathcal{C}$ is relatively simple, and is done via the formula $\varphi_{\textit{cconst}}$ defined below.
Since TRUE and FALSE gates do not have any meaningful input edges we can simply add all traces representing positive edges involving a TRUE gate and all traces representing a negative edge and a FALSE gate.
The formula $\varphi_{\textit{cconst}}$ quantifies over $\pi$ from the substructure for phase C and $\pi'$ from $Y$.
It first checks whether $\pi$ and $\pi'$ have the same input-output pair encoded to make sure that there is actually an edge in $\mathcal{C}$ represented by $\pi$.
Next, $\pi$ is added to $Z$ if the encoded sort is TRUE and $\pi$ represents a positive edge or if the gate is a FALSE gate and $\pi$ represents a negative edge.
\begin{align*}
\varphi_{\textit{cconst}} = &\forall\pi\in\mathfrak G.\forall\pi'\in Y.\LTLfinally c_\pi\land\LTLglobally (\textit{inp}_\pi\leftrightarrow\textit{inp}_{\pi'}) \land\LTLglobally (\textit{outp}_\pi\leftrightarrow\textit{outp}_{\pi'})\\
&\quad\land(\LTLfinally(\textit{kstart}\land\neg\textit{outp}_\pi \land \LTLnext\neg\textit{outp}_\pi\land\LTLnext\LTLnext\textit{outp}_\pi\land\LTLfinally\textit{pos}_\pi)\\
&\quad\quad\lor \LTLfinally(\textit{kstart}\land\neg\textit{outp}_\pi \land \LTLnext\textit{outp}_\pi\land\LTLnext\LTLnext\neg\textit{outp}_\pi\land\LTLfinally\textit{neg}_\pi))\to \pi \triangleright Z
\end{align*}
The remaining formulas work analogously to phase A. In particular, the formula $\varphi_{\textit{cand}}$ also adds negative traces for OR gates and $\varphi_{\textit{cor}}$ also adds negative traces for AND gates.
A formula of the form $(\textit{outp}_{\pi_2}\leftrightarrow\textit{inp}_{\pi_1})\LTLuntil\textit{nstart}_{\pi_1}$ is satisfied iff $\pi_2$ represents an edge to the gate that $\pi_1$ starts from.
We can check for the bit string that represents the kind of some gate by using a formula of the form $\LTLfinally(\textit{kstart}_{\pi_1} \land\textit{outp}_{\pi_1}\land\LTLnext\textit{outp}_{\pi_1}\land\LTLnext\LTLnext\textit{outp}_{\pi_1}))$ with the corresponding negations.
The same holds for the encoded neighbour number where we use $\LTLfinally(\textit{nstart}_{\pi_1} \land\textit{inp}_{\pi_1}\land\LTLnext\textit{inp}_{\pi_1}\land\LTLnext\LTLnext\textit{inp}_{\pi_1}))$ with corresponding negations.
\begin{align*}
&\varphi_{\textit{cand}} = \forall\pi_1\in\mathfrak G.\forall\pi'_1\in Y.\forall\pi_2,\pi_3\in Z.\LTLfinally c_{\pi_1}\land\LTLglobally (\textit{inp}_{\pi_1}\leftrightarrow\textit{inp}_{\pi_1'}) \land\LTLglobally (\textit{outp}_{\pi_1}\leftrightarrow\textit{outp}_{\pi_1'})\\
&\land(\textit{outp}_{\pi_2}\leftrightarrow\textit{outp}_{\pi_3}\leftrightarrow\textit{inp}_{\pi_1})~\LTLuntil ~ \textit{nstart}_{\pi_1}\land\neg\LTLfinally(\textit{nstart}_{\pi_1}\land\textit{inp}_{\pi_1}\land\LTLnext\LTLglobally\neg\textit{inp}_{\pi_1})\\
&\land\neg\LTLfinally(\textit{nstart}_{\pi_1}\land\neg\textit{inp}_{\pi_1}\land\LTLnext\textit{inp}_{\pi_1}\land\LTLnext\LTLnext\LTLglobally\neg\textit{inp}_{\pi_1})\land(\LTLfinally(\textit{kstart}\land\neg\textit{outp}_{\pi_1}\land\LTLnext\textit{outp}_{\pi_1}\land\LTLnext\LTLnext\textit{outp}_{\pi_1})\\
&\land\LTLfinally\textit{pos}_{\pi_1}\land\LTLfinally\textit{pos}_{\pi_2}\land\LTLfinally\textit{pos}_{\pi_3}\lor \LTLfinally(\textit{kstart}\land\textit{outp}_{\pi_1}\land\LTLnext\neg\textit{outp}_{\pi_1}\land\LTLnext\LTLnext\neg\textit{outp}_{\pi_1})\\
&\quad\land\LTLfinally\textit{neg}_{\pi_1}\land\LTLfinally\textit{neg}_{\pi_2}\land\LTLfinally\textit{neg}_{\pi_3})\to \pi_1\triangleright Z
\end{align*}
\begin{align*}
&\varphi_{\textit{cor}} = \forall\pi_1\in\mathfrak G.\forall\pi'_1\in Y.\forall\pi_2\in Z.\LTLfinally c_{\pi_1}\land\LTLglobally (\textit{inp}_{\pi_1}\leftrightarrow\textit{inp}_{\pi_1'}) \land\LTLglobally (\textit{outp}_{\pi_1}\leftrightarrow\textit{outp}_{\pi_1'})\\
&\land(\textit{outp}_{\pi_2}\leftrightarrow\textit{inp}_{\pi_1})~\LTLuntil ~ \textit{nstart}_{\pi_1}\land\neg\LTLfinally(\textit{nstart}_{\pi_1}\land\textit{inp}_{\pi_1}\land\LTLnext\LTLglobally\neg\textit{inp}_{\pi_1})\\
&\land\neg\LTLfinally(\textit{nstart}_{\pi_1}\land\neg\textit{inp}_{\pi_1}\land\LTLnext\textit{inp}_{\pi_1}\land\LTLnext\LTLnext\LTLglobally\neg\textit{inp}_{\pi_1})\land(\LTLfinally(\textit{kstart}\land\neg\textit{outp}_{\pi_1}\land\LTLnext\textit{outp}_{\pi_1}\land\LTLnext\LTLnext\textit{outp}_{\pi_1})\\
&\land\LTLfinally\textit{neg}_{\pi_1}\land\LTLfinally\textit{neg}_{\pi_2}\lor \LTLfinally(\textit{kstart}\land\textit{outp}_{\pi_1}\land\LTLnext\neg\textit{outp}_{\pi_1}\land\LTLnext\LTLnext\neg\textit{outp}_{\pi_1})\land\LTLfinally\textit{pos}_{\pi_1}\land\LTLfinally\textit{pos}_{\pi_2})\to \pi_1\triangleright Z
\end{align*}
\begin{align*}
&\varphi_{\textit{cnot}} = \forall\pi_1\in\mathfrak G.\forall\pi'_1\in Y.\forall\pi_2\in Z.~\LTLfinally c_{\pi_1}\land\LTLglobally (\textit{inp}_{\pi_1}\leftrightarrow\textit{inp}_{\pi_1'}) \land\LTLglobally (\textit{outp}_{\pi_1}\leftrightarrow\textit{outp}_{\pi_1'})\\
&\land(\textit{outp}_{\pi_2}\leftrightarrow\textit{inp}_{\pi_1})~\LTLuntil ~ \textit{nstart}_{\pi_1}\land\neg\LTLfinally(\textit{nstart}_{\pi_1}\land\textit{inp}_{\pi_1}\land\LTLnext\LTLglobally\neg\textit{inp}_{\pi_1})\\
&\land\LTLfinally(\textit{kstart}\land\textit{outp}_{\pi_1}\land\LTLnext\neg\textit{outp}_{\pi_1}\land\LTLnext\LTLnext\textit{outp}_{\pi_1})\land(\LTLfinally\textit{neg}_{\pi_1}\land\LTLfinally\textit{pos}_{\pi_2}\lor \LTLfinally\textit{pos}_{\pi_1}\land\LTLfinally\textit{neg}_{\pi_2})\to \pi_1\triangleright Z
\end{align*}
\begin{lemma}
The substructure described for phase C can be built in polynomial time in the size of $\mathcal{C}_S$.
\end{lemma}
\begin{proof}
The size of the substructure depends linearly on $l$.
Which states should be labeled with $\textit{kstart}$ and $\textit{nstart}$ can be trivially calculated from the input.
\end{proof}

\begin{lemma}\label{proofs:fpexpcomplete:phasec:correct}
For every input $e$ and corresponding output $o$ of $\mathcal{C}_S$ it holds that $Z$ contains a trace $\pi$ encoding $e$ and $o$ representing edge $(u, v)$ of $\mathcal{C}$ if $v$ is a successor of $u$. The trace
$\pi$ is labeled with $\textit{pos}$ iff the represented edge of $\mathcal{C}$ carries a positive value, and is labeled with $\textit{neg}$ otherwise.
\end{lemma}
\begin{proof}
We observe that every formula adding a trace $\pi_1$ to $Z$ only adds $\pi_1$ if it is labeled with $c$ at some point and there is a trace $\pi_1'\in Y$ such that $\pi_1$ and $\pi_1'$ encode the same input and output.
Therefore every trace in $Z$ is from the substructure for phase C and every trace in $Z$ encodes a valid input-output pair of $\mathcal{C}_S$.
Further, every trace $\pi$ represents an edge in $\mathcal{C}$ which goes from a gate $u \in V(\mathcal{C})$ to a gate $v \in V(\mathcal{C})$.
The ID of $u$ is encoded on $\pi$ with $\textit{inp}$ until $\pi$ visits a state labeled with $\textit{nstart}$ and similarly the ID of $v$ is encoded on $\pi$ with $\textit{outp}$ until $\pi$ visits a state labeled $\textit{kstart}$.
Therefore we have the following:
Let $\pi_1$ represent $(u, v)\in E(\mathcal{C})$ and let $\pi_2$ represent $(s, t)\in E(\mathcal{C})$. Then the formula 
$(\textit{outp}_{\pi_2}\leftrightarrow\textit{inp}_{\pi_1})~\LTLuntil~\textit{nstart}_{\pi_1}$ is satisfied iff $t = u$.

Additionally, $Z$ should only contain traces representing the edge $(u, v)$ where $v$ is the successor of $u$ and not the traces representing the fact that $u$ is the predecessor of $v$.
By the definition of the Succinct Circuit Value Problem, we know that the first two neighbors of AND and OR gates are its predecessors and the first neighbor of a NOT gate is its predecessor.
Therefore the formulas $\varphi_{\textit{cand}}$ and $\varphi_{\textit{cor}}$ require that a trace $\pi_1$ is only added to $Z$ if it does not binary encode a one or a two after $\textit{nstart}$.
Similarly $\varphi_{\textit{cnot}}$ requires that $\pi_1$ is only added if it does not encode a one after $\textit{nstart}$.
This suffices for $Z$ to only contain traces with information about successors.

The only difference to phase A is that the traces added to $Z$ represent edges from $\mathcal{C}$ and not $\mathcal{C}_S$ and that every trace encoding edge $(u, v)$ contains $\textit{sort}(u)$ and not $\textit{sort}(v)$.
The formulas check for the $\textit{sort}(u)$ by looking at the state where $\textit{kstart}$ holds and at the next two states.
Therefore dummy outputs that were added to $\mathcal{C}_S$ in phase A do not have any effect here.
We need to prove correctness for every $\textit{sort}(u)$. All cases are 
analogous to the cases in phase A, except the case for which $u$ is a TRUE or FALSE gate.
\begin{description}[nosep]
\item[$u$ is a constant gate:] If $u$ is a gate for constant true, then all edges leaving $u$ carry a positive value.
Therefore a trace representing an edge leaving $u$ should be included in $Z$ iff it is labeled with $\textit{pos}$.
Similarly, if $u$ is a FALSE gate a trace representing an edge leaving $u$ should be in $Z$ iff it is labeled with $\textit{neg}$. The set 
$Z$ contains exactly the described traces as otherwise it would not satisfy $\varphi_{\textit{const}}$.
\end{description}
\end{proof}
After building the set $Z$ we can determine the output of $\mathcal{C}$ by looking for its output gate.
By definition of the Succinct Circuit Value Problem, the output gate has its first output connected to a gate with ID 0.
Therefore the following $\varphi_D$ is satisfied if there exists a trace $\pi\in Z$ which represents an edge that carries a positive value and is connected to a gate with ID zero.
Additionally, $\pi$ needs to represent the first output of a gate.
Therefore, the satisfaction of $\varphi_D$ depends on the sort and neighbor number encoded on $\pi$.
\begin{align*}
&\varphi_D = \exists\pi\in Z. \LTLfinally\textit{pos}_\pi\land\neg\textit{outp}_\pi~\LTLuntil~\textit{kstart}_\pi\land(\LTLfinally(\textit{nstart}_\pi\land\textit{inp}_\pi\land\LTLnext\textit{inp}_\pi\land\LTLnext\LTLnext\LTLglobally\neg\textit{inp}_\pi)\\
&\land\LTLfinally(\textit{kstart}_\pi\land\textit{outp}_{\pi}\land\LTLnext\neg\textit{outp}_{\pi}\land\LTLnext\LTLnext\neg\textit{outp}_{\pi})\lor\LTLfinally(\textit{nstart}_\pi\land\textit{inp}_\pi\land\LTLnext\textit{inp}_\pi\land\LTLnext\LTLnext\LTLglobally\neg\textit{inp}_\pi)\\
&\land\LTLfinally(\textit{kstart}_\pi\land\neg\textit{outp}_{\pi}\land\LTLnext\textit{outp}_{\pi}\land\LTLnext\LTLnext\textit{outp}_{\pi})\lor\LTLfinally(\textit{nstart}_\pi\land\neg\textit{inp}_\pi\land\LTLnext\textit{inp}_\pi\land\LTLnext\LTLnext\LTLglobally\neg\textit{inp}_\pi)\\
&\land\LTLfinally(\textit{kstart}_\pi\land\textit{outp}_{\pi}\land\LTLnext\neg\textit{outp}_{\pi}\land\LTLnext\LTLnext\textit{outp}_{\pi})\lor\LTLfinally(\textit{nstart}_\pi\land\textit{inp}_\pi\land\LTLnext\LTLglobally\neg\textit{inp}_\pi)\\
&\land\LTLfinally(\textit{kstart}_\pi\land\neg\textit{outp}_{\pi}\land\LTLnext\neg\textit{outp}_{\pi}\land\LTLnext\LTLnext\textit{outp}_{\pi})\lor\LTLfinally(\textit{nstart}_\pi\land\textit{inp}_\pi\land\LTLnext\LTLglobally\neg\textit{inp}_\pi)\\
&\land\LTLfinally(\textit{kstart}_\pi\land\neg\textit{outp}_{\pi}\land\LTLnext\textit{outp}_{\pi}\land\LTLnext\LTLnext\neg\textit{outp}_{\pi})
)
\end{align*}

The formula $\varphi_D$ asserts that there is a trace $\pi \in Z$ encoding that the first edge leaving some gate goes to a gate with ID 0 and carries positive value.
Let $(u, v)$ be the edge represented by $\pi$.
The ID of $v$ is encoded with \textit{outp} before a state labeled with \textit{kstart} is visited.
Each trace $\pi$ satisfying $\neg\textit{outp}_\pi~\LTLuntil~\textit{kstart}_\pi$ therefore goes to a gate with ID 0.
The rest of the formula is a disjunction over the kind of $u$ which is encoded with \textit{outp} after a state labeled with \textit{kstart} was visited.
Depending on the kind of $u$ the first edge leaving $u$ has a different neighbor number.
The number is encoded with \textit{inp} after a state labeled with \textit{nstart} was visited.
It is three for AND and OR gates, it is two for NOT gates and one for constant gates.

\begin{lemma}
$K$ is acyclic
\end{lemma}
\begin{proof}
For every phase we built substructures consisting of layers and every layer is only connected to the next layer.
The only layers that are not connected to the next layer are the last layers which have self-loops, that satisfy our definition of acyclic structures.
\end{proof}
This concludes the construction and correctness of phase C, and we are ready to prove \Cref{proofs:fpexpcomplete:complete}. 

\paragraph*{\Cref{proofs:fpexpcomplete:complete} (restated)}
\mcfpfpdag is EXP-complete.

\begin{proof}
The answer to the Succinct Circuit Value Problem is true iff there exists a gate $u\in V(\mathcal{C})$ with no other gate $v\in V(\mathcal{C})$ such that $(u, v)\in E(\mathcal{C})$, and such that the output of gate $u$ is true.
We know by \Cref{proofs:fpexpcomplete:phasec:correct} that $Z$ contains a trace labeled with $\textit{pos}$ and encoding the $ID(u)$ with $\textit{inp}$ for some gate $u$ iff the output of gate $u$ is true.
By the definition of the Succinct Circuit Value Problem follows that if the $i+1$-th neighbor of a gate with $i$ inputs has ID zero then the gate has no successor and is the output of $\mathcal{C}$.
Note that the number of inputs a gate has is determined by its sort.
Therefore the answer to the Succinct Circuit Value Problem is true iff $Z$ contains a trace $\pi$ that is labeled with $\textit{pos}$ and represents an edge going to a gate with ID zero which is the $i+1$-th neighbor of a gate with $i$ inputs.
By the EXP-completeness of the Succinct Circuit Value Problem~\cite{CompuComplexity} follows the EXP-hardness of \mcfpfpdag.
From the EXP-hardness of \mcfpfpdag and \Cref{proofs:fpexpcomplete:inexp} follows that \mcfpfpdag is EXP-complete.
\end{proof}

\subsection{The Complexity of \mcfulltree}\label{proofs:hierarchycomplete}
We provide an example and a full proof for the lower bound of \Cref{QBF}. 

\begin{example}
Consider the QBF formula $y =\forall x_1.\exists x_2. x_1\lor x_2$. The formula 
$y$ has one quantifier alternation, such that $k = 1$. We follow the reduction described in the proof of \Cref{proofs:hierarchycomplete:complete}, in \Cref{QBF}. 
Since $y$ is valid, the model checking problem should be satisfied.
The Kripke structure built by the reduction can be seen in \Cref{proofs:hierarchycomplete:example}.

The formula
$\varphi$ for $k = 1$ starting with a universal quantifier, quantifies over the sets $X_1, X_2$ and $Z$.
The set $X_1$ is universally quantified and $X_2$ as well as $Z$ are existentially quantified.
By existentially quantifying over $Z$, $\varphi$ only contains one quantifier alternation.

The two possible instantiations for $X_1$ are: 
$$\{ \{\}\{q, v\}\{\}\{\}\dots, \{\}\{\textit{pos}\}\{\}\{\}\dots\}\quad \text{or} \quad\{ \{\}\{q, v\}\{\}\{\}\dots, \{\}\{\textit{neg}\}\{\}\{\}\dots\}$$

The possible instantiations for $X_2$ are $$\{ \{\}\{\}\{q, v\}\{\}\{\}\dots, \{\}\{\}\{\textit{pos}\}\{\}\{\}\dots\}\quad \text{or} \quad\{ \{\}\{\}\{q, v\}\{\}\{\}\dots, \{\}\{\}\{\textit{neg}\}\{\}\{\}\dots\}$$ 

There exists a set $Z$ satisfying the condition in $\varphi_{k+2}$ (i.e., $\varphi_3$) if the trace $\{\}\{\}\{\textit{pos}\}\{\}\dots$ is an element of either $X_1$ or $X_2$.
Because for both instantiations of $X_1$ there exists an instantiation of $X_2$ which contains $\{\}\{\}\{\textit{pos}\}\{\}\cdots$, the model checking problem is satisfied.
\end{example}

\begin{lemma}\label{proofs:hierarchycomplete:correct}
Let $y$ be a QBF formula, and let $K$ be the Kripke structure and $\varphi$ be the Hyper$^2$LTL formula constructed in the reduction described in the proof of \Cref{proofs:hierarchycomplete:complete} in \Cref{QBF}. Then, 
$K\models\varphi$ if and only if $y$ is valid.
\end{lemma}
\begin{proof}
We say that a trace $\pi$ represents variable $x_{i, j}$ or subexpression $E'$ if $\pi$ visits a state labeled with \textit{pos} or \textit{neg} in the state number $N(x_{i, j}) + 1$ or $N(E')+1$, respectively.
We first argue that an instantiation of $X_i$ with $i \le k+1$ satisfying $\varphi_i$ represents an assignment to all variables $x_{i, j}$ where $j \le m_i$.
There exists only one trace $\pi'$ in $K$ satisfying $\LTLnext^i q_{\pi'}$.
This is the only instantiation for $\pi'$ which can satisfy $\varphi_i$.
It is marked with $v$ at every position which represents a variable $x_{i, j}$.
An instantiation now contains, for each variable, exactly one trace $\pi''$ that has $\textit{pos}$ or $\textit{neg}$ at the respective position.
All traces $\pi$ representing variables satisfy $\LTLglobally(\neg\textit{epos}_\pi\land \neg\textit{eneg}_\pi\land\neg q_\pi)$.
The traces representing variables $x_{i, j}$ with $j\le m_i$ additionally satisfy $\LTLfinally (v_{\pi'}\land(\textit{pos}_\pi\lor\textit{neg}_\pi))$.
The formula $\varphi_i$ is satisfied if and only if, for each trace $\pi$ representing variable $x_{i, j}$, there is exactly one trace $\pi''\in X$ that represents the same variable.

By the construction of $K$ and by the definition of $\varphi$, we quantify over these assignments exactly as $y$ does.
It remains to show that under fixed instantiations for $X_1,\dots, X_{k+1}$, the formula $\varphi_{k+2}$ is satisfied if and only if $E$ (the quantifier free formula of $y$) is satisfied under the assignment represented by $X_1, \dots, X_{k+1}$.

First, we note that some trace $\pi$ can only be part of $Z$ if there are traces $\pi', \pi''$ that satisfy the condition in lines 3 and 4 of $\varphi_{k+2}$.
By the construction of $K$ and $\varphi_{k+2}$, this is the case if and only if the subexpressions represented by $\pi'$ and $\pi''$ have the value represented by $\textit{epos}$ and $\textit{eneg}$.

Now we prove that this instantiation satisfies the QBF formula $y$.
We start by showing that every possible instantiation of $Z$ contains at least one trace per quantifier-free subexpression in $y$.
We do this by well-founded induction over the relation that relates every subexpression of $E$ with its one or two largest subexpressions.
In the case that the subexpression is a variable, it is not related to any another subexpression.
The base case here is $E$.
By construction of $K$, only traces representing $E$ are marked with~$f$.
Additionally, the formula $\varphi_{k+2}$ enforces that there is at least one trace in $Z$ that is marked with $f$.
Therefore, the set~$Z$ contains a trace for $E$.
The induction step follows from the fact that a trace $\pi$ can only be included in $Z$ if there are two traces $\pi',\pi''\in Z$ satisfying the condition in lines 3 and 4.
This is by construction of $K$ the case only if $\pi'$ and $\pi''$ represent the one or two largest subexpressions of the expression represented by $\pi$.
If $\pi$ represents a variable, $\pi'$ and $\pi''$ can be instantiated with a trace labeled with $q$ since traces labeled with $q$ are not labeled with $\textit{pos}$ or $\textit{neg}$.

Next, we prove by structural induction over $E$ that every instantiation of $Z$ satisfying $\varphi_{k+2}$ contains a trace $\pi$ labeled with $\textit{pos}$ at position $j$ if and only if $e$ with $N(e) = j$ evaluates to true under the assignment represented by $X_1,\ldots, X_{k+1}$.
Instantiations for $Z$ contain the $\textit{neg}$ trace otherwise.
\begin{description}[nosep]
\item[$x_{a, b}$] A variable is true iff it is assigned to true.
	Since $Z$ is a superset of $X_1\cup \dots\cup X_{k+1}$, $Z$ contains the trace representing the value the variable is assigned to.
    Each trace that represents a variable $x_{a, b}$ and a value that it is not assigned to, cannot be in $Z$, as it would not satisfy the condition in line 5 of $\varphi_{k+2}$.

\item[$E_1\land E_2$]
	The positive trace $\pi$ can be included in $Z$, if there exist traces $\pi', \pi''$ that are marked with $\textit{pos}$ where $\pi$ is marked with $\textit{epos}$ or $\textit{epos'}$.
	By the construction of $K$, we know that $\pi'$ and $\pi''$ represent $E_1$ and $E_2$.
	By the induction hypothesis we have that $\pi'$ and $\pi''$ are only in $Z$ if and only if $E_1$ and $E_2$ are true.
	In this case, $E_1\land E_2$ is also true and $\pi$ should be part of $Z$.

	If $\pi$ is one of the negative traces it can only be included in $Z$ if there exist $\pi', \pi''\in Z$ such that at least one of them is marked with $\textit{neg}$ where $\pi$ is marked with $\textit{eneg}$.
	We know that this trace represents either $E_1$ or $E_2$ by the construction of $K$.
	By induction hypothesis, we know that this trace only exist in $Z$ if either $E_1$ or $E_2$ is negative.
	In this case, $E_1\land E_2$ is false and $\pi$ should be included in $Z$.

	If only one of $E_1$ or $E_2$ is negative, we can always instantiate $\pi''$ with a trace that is neither marked with $\textit{pos}$ nor with $\textit{neg}$.

\item[$E_1 \lor E_2$]
	Analogous to $E_1 \land E_2$.

\item[$\neg E_1$]
	If $\pi$ is the negative trace, it is marked with $\textit{epos}$ at the position that represents $E_1$.
	The trace $\pi$ can only be part of $Z$ if there exist traces $\pi', \pi''\in Z$ such that $\pi'$ is marked with $\textit{pos}$ exactly where $\pi$ is marked with $\textit{epos}$.
	Additionally, $\pi''$ has to be marked with $\textit{pos}$ or $\textit{neg}$ exactly where $\pi$ is marked with $\textit{epos'}$ or $\textit{eneg'}$. The trace
	$\pi''$ can be instantiated with some trace labeled neither with $\textit{pos}$ nor $\textit{neg}$.
	By the induction hypothesis, trace $\pi'$ with the described properties only exists if $E_1$ is true.
	Therefore, the negative trace $\pi$ can only be included in $Z$ if $\neg E_1$ is false.

	If $\pi$ is the positive trace, the situation is analogous.
\end{description}

Only the traces representing $E$ are labeled with $f$.
Therefore, a trace $\pi$ labeled with $f$ as well as $\textit{pos}$ is contained in $Z$ if and only if $E$ is true under the assignment represented by $X_1,\dots, X_{k+1}$.
\end{proof}

\begin{lemma}\label{proofs:hierarchycomplete:poly}
The reduction of \Cref{proofs:hierarchycomplete:correct} takes polynomial time in the size of $y$.
\end{lemma}
\begin{proof}
The structure $K$ consists of branches of length $l$, which is polynomial in the size of $y$.
Additionally, $K$ contains only a constant amount of branches per operator and variable in $y$.
Therefore, $K$ can be built in polynomial time.

The formula $\varphi$ only depends on the number of quantifier alternations in $y$.
These are linear in the size of $y$.
\end{proof}

\paragraph*{\Cref{proofs:hierarchycomplete:complete} (restated)}
\mcfulltree is $\Sigma^p_{k+1}$-complete if the given formula has $k$ quantifier alternations and the outermost second-order quantifier is existential.
\mcfulltree is $\Pi^p_{k+1}$-complete if the given formula has $k$ quantifier alternations and the outermost second-order quantifier is universal.

\begin{proof}
The Quantified Boolean Formula Problem with $k$ quantifier alternations is $\Sigma^p_{k+1}$-complete if the first quantifier is existential and $\Pi^p_{k+1}$-complete otherwise~\cite{computersintractability}.
By the reduction (\Cref{proofs:hierarchycomplete:correct}, \Cref{proofs:hierarchycomplete:poly}) holds for every $k$, that \mcfulltree with $k$ quantifier alternations is $\Sigma^p_{k+1}$-hard if the formula starts with an existential quantifier and $\Pi^p_{k+1}$-hard otherwise.

By \Cref{proofs:hierarchycomplete:lowerbound} the completeness in the respective complexity classes follows.
\end{proof}

\subsection{The Complexity of \mcfulldag}\label{app:exphierarchy}
In this section we prove \Cref{proofs:exphierarchy:complete}, that is, we show that \mcfulldag with $k$ second-order quantifier alternations is $\Sigma^{EXP}_{k+1}$-hard if the outermost second-order quantification is existential and $\Pi^{EXP}_{k+1}$-hard otherwise.
We reduce from the universal problem in these classes.
That is, given an alternating Turing machine $M$ which does not alternate more than $k$ times, and given a binary number $m$ decide whether $M$ accepts after less than $m$ steps.
To do so, we transform every alternating Turing machine $M$ and natural number $m$ into a Kripke structure $K$ and a formula $\varphi$ such that $K\models \varphi$ if and only if $M$ accepts after less than $m$ steps.
The formula $\varphi$ may only depend on $k$ and whether the initial state of $M$ is universal or existential.

We require that $M$ has exactly one accepting and one rejecting state which are both sink states.
Additionally, $M$ does not have more than $k$ alternations between existential and universal states.
We unroll these alternations such that every computation of $M$ alternates exactly $k$ times between universal and existential states.
We call the states of $M$ which correspond to the same unrolled alternation of $M$ an \emph{alternation block}.
Every valid run of $M$ passes the same $k + 1$ alternation blocks in the same order.

For each existential state $s$ that has universal as well as existential states as predecessors, we add one state $s'$ which is also existentially quantified.
Between $s'$ and $s$ there is one transition that does not alter the head position or the tape.
All transitions that enter $s$ from a universal state are moved to be transitions entering $s'$.
We apply the dual construction for universal states that have existential states as well as universal states as predecessor.
This makes sure that all alternation blocks are separated by a set of states that are universal and only have existential predecessors or vice versa.
Consequently, each state has only existential or universal predecessors.
Finally, to encode states and tape symbols on traces, we number states and tape symbols.

To encode $M$ into $K$ , we need to represent the transition function of $M$, the tape, the current state of $M$, and a counter.
All traces in $K$ have length $l+1$ which we define as the number of bits needed to encode $2\cdot m$.
The structure $K$ encodes all valid transitions of $M$, some initial configuration and otherwise only provides enough traces such that every possible state and tape content can be represented as a set of traces.
Thus, $K$ consists of the following substructures:
\begin{itemize}[nosep]
\item A substructure representing the traces for a counter, encoding all possible numbers with $a$ and $b$.
All its states are labeled with \textit{counter}.
\item One substructure per transition in $M$, encoding states with $\textit{from, to}$; tape letters with $\textit{sig, sig'}$; an arbitrary time stamp with $t$. Some states are additionally labeled with $l, r, u, e, \textit{alt}$ and $\textit{start}$.
All states are also labeled with~\textit{step}.
\item A substructure containing traces for the tape entries, encoding the letter with \textit{sig}, an arbitrary time stamp with $t$ and the cell number with $c$.
All states are additionally labeled with \textit{tape}.
\item A substructure containing traces repesenting that all tape cells are filled with blanks at time stamp 0.
All these states are also labeled with \textit{tape} and \textit{init}.
\item A substructure containing exactly one trace encoding $m$ with atomic proposition \textit{m} and all its states labeled with \textit{const}.
\end{itemize}
$\varphi$ quantifies over two auxiliary sets $S, G$, and sets $X_1,\ldots, X_{k+1}$ containing an execution of $M$.
The instantiations of a set $X_i$ are restricted such that each instantiation of $X_i$ is a valid execution branch through alternation block $i$ of $M$.

\emph{Counter}
We build $K$ such that it has a substructure that contains traces encoding all binary numbers with $a$ and $b$.
For this we add states $c_{i, j}$ with $i\in \{1..l\}, j\in\{1..4\}$.
Every state $c_{i, j}$ has transitions to states $c_{i+1, j'}$ for all $j'$ except for the states $c_{l, j}$ which have self loops.
The initial state of $K$ is connected to states $c_{1, j}$.
States $c_{i, 1}$ and $c_{i, 2}$ are labeled with the proposition $a$ and states $c_{i, 1}$ and $c_{i, 3}$ are labeled with $b$.
All states of this structure are labeled with \textit{counter}.

\emph{Transition Relation}
The transition relation of a Turing machine maps a given state and symbol on the tape to the successor state, new symbol and head movement.
For every transition in the transition relation we add one substructure to $K$ which encodes two states, two symbols and the head movement.
All traces in these substructures visit a state labeled with \textit{step}.
The current state is encoded with atomic proposition $\textit{from}$ and the current symbol is encoded with atomic proposition $\textit{sig}$.
The next state is encoded with $\textit{to}$ and the new tape symbol is encoded with $\textit{sig'}$.

For the head movement we use atomic propositions $l$ and $r$.
If the head should move to the right (in the direction of higher numbered tape cells), then the traces are marked with $r$.
If the head should move to the left, then it is marked with $l$.
If the head should not move, then neither $l$ nor $r$ are present on the trace.

The substructure also encodes relevant information about the state the encoded transition is leaving.
All substructures representing transitions leaving a universal state are marked with atomic proposition $u$ and if the respective transition leaves an existential state, they are marked with $e$.
If, additionally, the state is the initial state then it is marked with atomic proposition \textit{start}.
Similarly, if the transition enters the accepting state, then it is marked with \textit{acc} and if the entered state is the first state of an alternation block it is marked with~\textit{alt}.

Every substructure is built such that it can encode arbitrary numbers with atomic propositions $t$ and $h$.
The proposition $t$ represents timestamps and the proposition $h$ represents the corresponding head position.
We call a transition of $M$ together with a concrete timestamp and a head position a \emph{step} of $M$.

\emph{Tape}
The tape is represented as a sequence of write accesses to particular tape cells.
Every such access is represented by a trace encoding the number of the tape cell with atomic proposition $c$, the time stamp with atomic proposition $t$ and the written letter with \textit{sig}.
The substructure providing all traces for the tape therefore contains traces encoding arbitrary numbers with $c, t$ and \textit{sig}.
To identify traces representing tape accesses, they are also marked with \textit{tape}.

\emph{Constants}
There is one important constant:
$m$, the number of maximal steps.
$m$ is also the starting position of the head on the tape.
To encode these, we add one trace to $K$ labeled with \textit{const} and encoding $m$ with atomic proposition $m$.

\emph{Initial Configuration}
Initially the tape is filled with the special blank letter.
Because these letters will not be added by taking a transition, we will construct a substructure containing them.
All traces in this structure are labeled with \textit{tape} as well as \textit{init}.
They encode arbitrary numbers with $c$ but only encode the blank symbol with \textit{sig} and are not labeled with $t$ (to represent time stamp 0).

\emph{Formula}
We construct a corresponding formula $\varphi$ that is satisfied by $K$ if and only if the Turing machine $M$ accepts within $m$ steps.
The formula may only depend on $k$ and whether the initial state of $M$ is existential or universal.

We start by setting up the counter.
The counter implements the clock of $M$ and is needed to decide which content of a tape cell is the newest.
Therefore, we quantify over two sets of traces -- $S$ (in $\varphi_S$) and $G$ (in $\varphi_G$).
The set $S$ contains all traces that encode two binary numbers with $a$ and $b$ such that the number encoded with $b$ is the successor of the number encoded of $a$. The set 
$G$ is the transitive closure of $S$ and therefore represents the greater than relation on numbers of length $l$.

We use the LTL formula $(a\land \neg b)\LTLuntil(\neg a\land b\land \LTLnext\LTLsquare(a\leftrightarrow b))$ to identify all traces where the number encoded with $b$ is the successor of the number encoded with~$a$.
\begin{align*}
\varphi_S =&\mathbb Q S.\forall \pi\in \mathfrak G.\pi\triangleright S\leftrightarrow(\LTLfinally \textit{counter}_\pi\land\LTLfinally b_\pi\land
((a_\pi\land \neg b_\pi)\LTLuntil(\neg a_\pi\land b_\pi\land \LTLnext\LTLsquare(a_\pi\leftrightarrow b_\pi)))) \\
\varphi_G =& \mathbb Q G.\forall\pi\in \mathfrak G.\pi\triangleright G\\
&\leftrightarrow
(\pi\triangleright S\lor
\exists\pi'\pi''\in G.\LTLsquare((b_{\pi'}\leftrightarrow a_{\pi''})\land(a_{\pi}\leftrightarrow a_{\pi'})\land(b_{\pi}\leftrightarrow b_{\pi''})))\land\LTLfinally\textit{counter}_\pi
\end{align*}
The quantifier $\mathbb Q$ in $\varphi_S$ and $\varphi_G$ as well as the connective between the two formulas depend on the quantification of the initial state of $M$.
Therefore, they do not introduce another quantifier alternation.

The overall structure of $\varphi$ is as follows:
\begin{align*}
\varphi =
\mathbb Q_1 S. \varphi_S \oplus_1 \mathbb Q_1 G.\varphi_G\oplus_1
\mathbb Q_1 X_1.\varphi_1\oplus_1\mathbb Q_2 X_2.\varphi_2\oplus_2\dots\mathbb Q_{k+1} X_{k+1}.\varphi_{k+1}\oplus_{k+1} \varphi_Y
\end{align*}
If the initial state of $M$ is existential, then $\mathbb Q_i$ is an existential quantifier if and only if~$i$ is odd.
Otherwise, $\mathbb Q_i$ is an existential quantifier if and only if $i$ is even.
If $\mathbb Q_i \neq \exists$ then $\mathbb Q_i = \forall$.
Further, $\oplus_i$ is $\land$ if $\mathbb Q_i$ is existential and $\rightarrow$ otherwise.

The overall formula $\varphi$ collects traces representing steps of $M$ in alternation block $i$ into set $X_i$ and quantifies over it universally if the $i$-th alternation block contains only universal states, and existentially otherwise.

We start by defining a subformula $\textit{validS}(\pi, \pi')$ which is true if and only if $\pi'$ represents a step that is a valid successor of the step represented with $\pi$.
\begin{align*}
&\textit{validS}(\pi, \pi') := \LTLfinally \textit{step}_\pi \land \LTLfinally\textit{step}_{\pi'} \land\LTLsquare(\textit{to}_\pi\leftrightarrow\textit{from}_{\pi'}) & (1)\\
&\land\big [(\neg\LTLfinally r_\pi\land\neg\LTLfinally l_\pi\rightarrow\LTLsquare(h_\pi\leftrightarrow h_{\pi'})) \land(\LTLfinally r_\pi\rightarrow\exists\rho\in S.\LTLsquare((a_\rho\leftrightarrow h_\pi)\land(b_\rho\leftrightarrow h_{\pi'}))) &(2)\\
&\land(\LTLfinally l_\pi\rightarrow\exists\rho\in S.\LTLsquare((b_\rho\leftrightarrow h_\pi)\land(a_\rho\leftrightarrow h_{\pi'})))\big ]\quad&(3)\\
&\land\big [\exists\rho\in S.\LTLsquare((t_\pi\leftrightarrow a_\rho)\land(t_{\pi'}\leftrightarrow b_\rho))\big ] \quad&(4)\\
&\land\big [\bigvee_{j \le i}\exists\tau\in X_j.\LTLfinally \textit{tape}_\tau \land
\LTLsquare((\textit{sig}_{\pi'}\leftrightarrow \textit{sig}_\tau)\land(h_{\pi'}\leftrightarrow c_\tau)) \quad&(5)\\
& \land\exists\gamma\in G.\LTLsquare((t_\tau\leftrightarrow a_\gamma)\land(t_{\pi'}\leftrightarrow b_\gamma))\land\neg\big[ \bigvee_{j'\le i}\exists\tau'\in X_{j'}.\LTLfinally\textit{tape}_{\tau'}\land
\LTLsquare(c_\tau\leftrightarrow c_{\tau'}) \quad&(6)\\
&\land\exists\gamma,\gamma'\in G.\LTLsquare((t_\tau\leftrightarrow a_\gamma)\land(t_{\tau'}\leftrightarrow b_\gamma)\land(t_{\tau'}\leftrightarrow a_{\gamma'})\land(t_{\pi'}\leftrightarrow b_{\gamma'}))\big ]\big ] &(7)\\
&\land\big [\exists\tau\in X_i.\LTLfinally\textit{tape}_\tau\land\LTLsquare((\textit{sig'}_{\pi}\leftrightarrow \textit{sig}_\tau)\land(h_{\pi}\leftrightarrow c_\tau)\land(t_{\pi'}\leftrightarrow t_\tau))\big ]&(8)
\end{align*}
A step $s'$ is a valid successor of step $s$ iff the following conditions are met:
\begin{itemize}[nosep]
	\item $s$ ends in the state that $s'$ starts from (rightmost conjunction of line 1).
	\item The head has moved accordingly to $s$ (lines 2-3).
	\item The timestamp on $s'$ is the successor of the timestamp of $s$ (line 4).
	\item The tape has the required letter at the current step and head position (lines~5-7).
	\item The letter written by $s$ is on the tape at the old head position but at the current step (line 8).
\end{itemize}

To check in \textit{validS} whether or not there is a given letter $\sigma$ at some given head position $h$ and timestamp $t$ does not only require us to check whether there is a trace that indicates that~$\sigma$ is written at $h$ (line 7), but we also have to make sure that there is a suitable timestamp $t'$ encoded.
Such a $t'$ must be before $t$ (line 8) and there must be no trace encoding that there was a(nother) letter written to the tape position between $t'$ and $t$ (line 9-10).

We use the formula $\textit{validS}$ above to ensure that all traces in $X_i$ are a valid pass of alternation block $i$, using the following formula $\varphi_i^{\text{con}}$ for some $i>1$:
\begin{align*}
&\varphi_i^{\text{con}} = (\exists\pi_1,\pi_2\in X_i. (\exists!\pi_1'\in X_{i-1}. \textit{validS}(\pi_1', \pi_1)\land\textit{alt}_{\pi_1'} \land\neg(\exists\pi_1''\in X_{i-1}. \textit{validS}(\pi_1', \pi_1''))\land\textit{alt}_{\pi_2})\\
&\land\forall\pi\in X_i. \pi\neq\pi_1 \land \pi\neq\pi_2\rightarrow (\exists!\pi'\in X_i.\textit{validS}(\pi,\pi')) \land \exists!\pi'\in X_i.\textit{validS}(\pi',\pi))
\end{align*}
The formula $\varphi_i^{\text{con}}$ is satisfied if the transitions in $K$ are connectable to a path, which means that for every step represented as trace $\pi$ there is exactly one predecessor step and one successor step (bottom line).
The only situation in which there is no valid predecessor (successor) to a step is when the step is the first (last) step of the computation fragment.
If the step is the first step of a computation ($\pi_1$) then it should be connectable to the computation fragment that is included in $X_{i-1}$.
Further, if the step is the last step of the computation ($\pi_2$) in $X_i$ then it should end in the next alternation block of $M$ ($\textit{alt}_{\pi_2}$).
In the case where $i = 1$, we replace this with the following condition:
\setcounter{equation}{0}
\begin{align*}
&\varphi^{\text{con}}_1 = (\exists\pi_1,\pi_2\in X_1. (\textit{start}_{\pi_1}\land \LTLnext t_{\pi_1}\land\LTLnext\LTLnext\LTLglobally \neg t_{\pi_1}\\
&\land(\exists \pi'\in \mathfrak G.\LTLfinally\textit{const}_{\pi'}\land \LTLglobally(h_{\pi_1}\leftrightarrow m_{\pi'}))\land\textit{alt}_{\pi_2})\land\forall\pi\in X_1. \pi\neq\pi_1 \land \pi\neq\pi_2\\
&\rightarrow (\exists!\pi'\in X_1.\textit{validS}(\pi,\pi')) \land \exists!\pi'\in X_1.\textit{validS}(\pi',\pi))\land\forall \pi\in\mathfrak G.\LTLfinally\textit{init}\rightarrow \pi\triangleright X_1
\end{align*}
$\varphi^{\text{con}}_1$ specifies that there exists exactly one valid start step at time 1 (first line), which is a transition leaving the initial state of $M$ and has the head at position $m$ (second line).
There are two more conditions to enforce for $X_1, \dots, X_{k+1}$:
\begin{itemize}[nosep]
\item Each universally quantified set only contains transitions leaving universal states and equally for existentially quantified sets.
\item Every trace representing a write to the tape has a corresponding step.
\end{itemize}
The first condition is rather easy.
We identify all relevant traces and make sure that they are marked correctly with $u$ or $e$.
We define $\varphi^q_i$ as $\forall\pi\in X_i. \LTLfinally \textit{step}_\pi\rightarrow e_\pi$ if $i$ is odd and~$M$ starts in an existential state or if $i$ is even and $M$ starts in a universal state.
We define $\varphi^q_i = \forall\pi\in X_i. \LTLfinally \textit{step}_\pi\rightarrow u_\pi$ in all other cases.

Finally, $\varphi^t_i$ is required such that the sets do not contain more traces than they should.
While for transitions this is covered by the requirement that they are connectable to a path it is not covered for the traces representing the tape.
Therefore, we do this here:
\begin{align*}
\varphi^t_i =
\forall\pi\in X_i.\LTLfinally\textit{tape}_\pi\rightarrow \exists \pi'\in X_i.\LTLfinally\textit{step}_{\pi'}\land
\LTLglobally((\textit{sig}'_{\pi'}\leftrightarrow \textit{sig}_\pi)
\land(h_{\pi'}\leftrightarrow c_\pi)
\land(t_\pi\leftrightarrow t_{\pi'}))
\end{align*}
The complete condition $\varphi_i$ that a set $X_i$ has to satisfy is the conjunction of the previously defined formulas: $\varphi_i = \varphi^{\text{con}}_i\land\varphi^q_i\land\varphi^t_i$.

The last condition that has to be true on the computation branch is that it is accepting in less than $m$ steps.
The formula $\varphi_Y$ is therefore satisfied if set $X_{k+1}$ contains a step entering the accepting state of $M$ with a time stamp less than $m$.
Such a step is marked with $acc$ by construction.
\begin{align*}
\varphi_Y = \exists\pi\in X_{k+1}.\LTLfinally\textit{step}_\pi\land \LTLfinally\textit{acc}_\pi\land \exists\gamma\in G.\exists\pi'\in\mathfrak G.\LTLfinally\textit{const}_{\pi'}\land\LTLglobally((t_\pi\leftrightarrow a_\gamma)\land(m_{\pi'}\leftrightarrow b_\gamma))
\end{align*}

We now turn to prove the correctness of the construction of $K$ and definition of $\varphi$, in a sequence of lemmas, that show the correctness of each of the subformulas we have defined. 

\begin{lemma}
The set $S$ contains a trace $\pi$ iff $\pi$ satisfies the following conditions: (1) $\pi$ encodes two $l$-bit numbers $A, B$ with $a$ and $b$; and (2) $A+1 = B$.
\end{lemma}
\begin{proof}
The first condition is trivial.
The second condition holds by the definition of binary increment.
Starting from the LSB, flip every one to a zero ($a\land\neg b$) until there is a zero which then gets changed to a one.
This means that before $\neg a$ is true it holds that $a\land\neg b$, and at the moment where $\neg a$ holds, also $b\land\LTLnext\LTLglobally(a\leftrightarrow b)$ has to be true.
Since the traces only encode the first $l$ bits of a natural number we have to prevent an overflow which is done by requiring that $B\neq 0$.
\end{proof}

\begin{observation}
\label{proofs2:exphierarchy:SObservation}
If $t_\pi$ and $t_{\pi'}$ encode the two numbers $T, T' \leq 2^l$ then the formula  $\exists\pi''\in S.\LTLglobally((a_{\pi''}\leftrightarrow t_\pi)\land(b_{\pi''}\leftrightarrow t_{\pi'}))$ is satisfied if and only if $T+1 = T'$.
\end{observation}

\begin{lemma}
Any set $G$ satisfying the condition in $\varphi_G$ contains for all $0\le A < B \le 2^l$ a trace $\pi$ encoding $A$ with $a$ and $B$ with $b$.
\end{lemma}
\begin{proof}
Because $<$ is the transitive closure of increment it suffices to show that $G$ represents the transitive closure of the relation represented by $S$.
Every trace $\pi$ from the counter substructure should be contained in $G$ if and only if it is in $S$ or it encodes $A, B$ such that there is a $C$ with $A < C < B$.
The second condition is equal to requiring that there exists two other traces $\pi',\pi''\in G$ such that the following hold:
$C$ is encoded with $b$ on $\pi'$ and with $a$ on $\pi''$;
 $A$ is encoded on $\pi'$ with $a$; and $B$ is encoded on $\pi''$ with $b$.
Therefore it is equivalent to $\exists\pi'\pi''\in G.\LTLsquare((b_{\pi'}\leftrightarrow a_{\pi''})\land(a_{\pi}\leftrightarrow a_{\pi'})\land(b_{\pi}\leftrightarrow b_{\pi''}))$.
\end{proof}

After proving that $S$ as well as $G$ contain the traces to compare numbers, we prove that the formulas $\varphi_i$ specify valid computation fragments of $M$.
We do this by induction over the number encoded with $t$ (the timestamp).

We say that a set $X_i$ of traces \emph{describes a computation of $M$} if and only if all traces in $X_i$ that are marked with \textit{step} are connectable to a series of steps of $M$, all traces in $X_i$ that are marked with \textit{tape} are the corresponding letters written to the tape, and $X_i$ does not contain any other traces.

\begin{lemma}\label{proofs2:exphierarchy:indbase}
Every set $X_1$ satisfying $\varphi_1$ contains the initial configuration of $M$.
\end{lemma}

\begin{proof}
The base case here is the computation fragment of length zero, i.e.,  $t = 0$.
We have to show that $\varphi_1$ is satisfied if $X_1$ contains all traces representing the empty tape.
This is the case because $X_1$ can only be satisfied if it contains all traces labeled with \textit{init} which represent the blank symbol on every tape cell at timestamp 0 ($\varphi^\textit{con}_1$ line 6).
\end{proof}

\begin{lemma}\label{proofs2:exphierarchy:indstep}
For all $1 \le i\le k+1$ we have that if all $X_{j}$ with $j < i$ describe a computation branch of $M$ in alternation block $j$, then $X_i$ satisfies $\varphi_i$ iff it describes a continuation in alternation block $i$ of the computation described $X_{i-1}$.
\end{lemma}
\begin{proof}
We start by proving that $\textit{validS}(\pi, \pi')$ is satisfied if and only if $\pi'$ is a successor step of $\pi$.
As already mentioned, this requires that the encoded steps $s$ and $s'$ satisfy the following conditions:
\begin{itemize}[nosep]
\item $s$ ends in the state $s'$ starts from.
\item The head has moved accordingly to $s$.
\item $s'$ is one step later than $s$.
\item The required letter is on the tape at the current step and head position.
\item The letter written by $s$ is on the tape at the old head position but at the current step.
\end{itemize}

The first condition is trivially covered by $\LTLglobally(\textit{to}_\pi\leftrightarrow\textit{from}_{\pi'})$.
The head movement according to the transition taken in $s$ is encoded in $\pi$.
The case that the head should not move is easily satisfied by adding $\neg\LTLfinally r_\pi\land\neg\LTLfinally l_\pi\rightarrow\LTLglobally(h_\pi\leftrightarrow h_{\pi'})$ as requirement to satisfy $\textit{validS}$.
The correctness of the other two cases follows analogously with consideration of Observation~\ref{proofs2:exphierarchy:SObservation}.
Similarly, it follows that \textit{validS} can only be satisfied if the timestamp encoded with $t$ on $\pi'$ is the successor of the timestamp encoded on $\pi$.

To decide whether the correct letter $\sigma$ is on the tape cell $c$, we need to show that there is a trace $\tau$ that represents that the newest write to $c$ is $\sigma$.
Let $t'$ be the timestamp encoded on trace $\tau$.
The letter must be written before $s'$ is executed. Therefore, we know by induction hypothesis that such a $\tau$ exists if and only if $\sigma$ was written at $t'$ to $c$.
Lines 7 and 8 in the definition of the formula $\textit{validS}$ are satisfied if and only if $\tau$ has this property.

Additionally, $\tau$ has to represent the newest write to $c$.
For this we require that there exists no trace $\tau'$ which also represents a write to $c$ (line 9 in the definition of $\textit{validS}$) and has a time stamp between $t$ and $t'$ (line 10).

Finally, \textit{validS} may only be satisfied if $X_i$ contains a trace $\tau$ that represents that the letter encoded with \textit{sig'} on $\pi$ was written to the tape cell on the current head position and timestamp (line 11).
This concludes the correctness of $\textit{validS}$.

For the satisfaction of $\varphi_i$ all represented steps need to be connectable to a path.
This is done by making sure that every step has exactly one valid predecessor and exactly one successor.
If a particular step $\pi\in X_i$ is neither the step with the least nor the highest timestamp in its set $X_i$, then this is ensured by requiring $(\exists!\pi'\in X_i.\textit{validS}(\pi,\pi')) \land \exists!\pi'\in X_i.\textit{validS}(\pi',\pi)$ in $\varphi_i^{con}$ (and $\varphi_1^{con}$).

There are exactly two exceptions to this rule: $\pi_1$ and $\pi_2$.
Because the computation branch represented in $X_i$ should continue the computation branch in $X_{i-1}$ we have to ensure that the last trace from $X_{i-1}$ is a valid predecessor to~$\pi_1$.

At the same time this ensures that the computation in $X_i$ passes alternation block $i$.
By induction we know that the computation in $X_{i-1}$ passes alternation block $i-1$ and ends in a state of alternation block $i$.
For $X_1$ this is done by requiring that $\pi_1$ encodes 1 as timestamp and leaves the initial state of $M$.
As all states in $X_i$ must be labeled with~$e$ or all states in $X_i$ must be labeled with~$u$, the computation in $X_i$ can not leave alternation block $i$.
It ends in alternation block $i+1$ by requiring that $\pi_2$ is labeled with $\textit{alt}$ and therefore with the first state of alternation block $i+1$ or the rejecting / accepting state.
\end{proof}

\begin{lemma}\label{proofs:exphierarchy:correct}
$M$ accepts in less than $m$ steps if and only if $K\models \varphi$.
\end{lemma}

\begin{proof}
Assume that $K\models\varphi$.
We can construct an accepting computation tree for $M$.
We know that each instantiation represents a valid computation branch through one alternation block.
Additionally, every part of the computation tree corresponding to the pass through alternation block $i$, contains either all continuations of the preceding branch (if alternation block $i$ contains universal states and $X_i$ is universally quantified) or contains one continuation of the preceding branch (if $X_i$ is existentially quantified).
As we show in \Cref{proofs2:exphierarchy:indstep}, we can use the instantiations of $X_i$ to construct the corresponding part of the computation tree.
Each branch of the computation tree ends in the accepting state because the accepting state is a sink state, $\varphi$ is only satisfied if there is a step entering the accepting state in $X_{k+1}$ and each transition to a leaf is present in $X_{k+1}$.
A branch can not be longer than $m$ because this would violate $\varphi_Y$.

Assume that $M$ accepts after $k$ alternations.
We divide the computation at the alternations into $k+1$ partitions.
We use these partitions to instantiate $X_1, \dots, X_{k+1}$.
By \Cref{proofs2:exphierarchy:indstep} this satisfies $\varphi_i$ for each $i$, and by the order of the quantifiers in $\varphi$, all universally quantified sets correspond to universal computation parts, and similarly for existential computation parts.
As $M$ accepts in less than $m$ steps, it takes a step entering the accepting state.
This step must be marked with \textit{acc} and has a timestamp less than $m$.
It will therefore satisfy $\varphi_Y$.
\end{proof}

\paragraph*{\Cref{proofs:exphierarchy:complete} (restated)}
\textsc{MC[$\Sigma_k$-Hyper$^2$LTL, acyclic]}
is $\Sigma^{EXP}_{k+1}$-complete, and 
\textsc{MC[$\Pi_k$-Hyper$^2$LTL, acyclic]}
 is $\Pi^{EXP}_{k+1}$-complete.

\begin{proof}
Containment in the respective classes follows from \Cref{proofs:exphierarchy:lowerbound}.
The hardness in the classes follows from \Cref{proofs:exphierarchy:correct} and the fact that the given reduction is polynomial time computable in the size of the Turing machine.
\end{proof}

\end{document}